\documentclass[letterpaper,11pt]{article}

\usepackage[margin=1in]{geometry}
\usepackage{amsmath,amssymb,amsthm,amsfonts,mathtools}
\usepackage{braket}
\usepackage{bm}
\usepackage{bbold}
\usepackage{enumitem}
\usepackage[normalem]{ulem}
\usepackage{xcolor}
\usepackage{authblk}
\usepackage[title,page]{appendix}
\usepackage{algorithm}
\usepackage{algorithmic}
\usepackage{booktabs}
\usepackage{tikz}
\usepackage{tikz-cd}
\usepackage{quantikz}
\usetikzlibrary{arrows.meta, bending}
\definecolor{navy}{RGB}{23,63,108}

\usepackage{scalerel}
\usepackage{float}
\usepackage{graphicx}
\usepackage{verbatim}
\usepackage{cite}
\usepackage{mdframed}

\usepackage[colorlinks=true, allcolors=blue, linktocpage=true]{hyperref}

\newtheorem{theorem}{Theorem}[section]
\newtheorem{lemma}[theorem]{Lemma}

\newtheorem{proposition}[theorem]{Proposition}
\newtheorem{definition}[theorem]{Definition}

\newcommand{\seed}{\ket{\mathrm{seed}}}

\newcommand{\Id}{\mathbb{1}}
\DeclareMathOperator{\poly}{poly}

\makeatletter

\skip\@mpfootins = \skip\footins
\renewcommand\subsubsection{\@startsection{subsubsection}{3}{\z@}%
                                   {-2.3ex\@plus -1ex \@minus -.5ex}%
                                   {1ex \@plus .2ex \@minus .2ex}%
                                   {\normalfont\normalsize\bfseries}}
\renewcommand\paragraph{\@startsection{paragraph}{4}{\z@}%
                                   {1.75ex \@plus1ex \@minus.2ex}%
                                   {-1em}%
                                   {\normalfont\normalsize\bfseries}}
\renewcommand\subparagraph{\@startsection{subparagraph}{5}{\z@}%
                                   {1.75ex \@plus1ex \@minus .2ex}%
                                   {-1em}%
                                   {\normalfont\normalsize\itshape}}

\renewcommand{\@dotsep}{10000}

\def\fnum@figure{\textbf{\figurename\nobreakspace\thefigure}}
\def\fnum@table{\textbf{\tablename\nobreakspace\thetable}}

\allowdisplaybreaks
\numberwithin{equation}{section}

\makeatletter
\newcommand{\subalign}[1]{%
  \vcenter{%
    \Let@ \restore@math@cr \default@tag
    \baselineskip\fontdimen10 \scriptfont\tw@
    \advance\baselineskip\fontdimen12 \scriptfont\tw@
    \lineskip\thr@@\fontdimen8 \scriptfont\thr@@
    \lineskiplimit\lineskip
    \ialign{\hfil$\m@th\scriptstyle##$&$\m@th\scriptstyle{}##$\hfil\crcr
      #1\crcr
    }%
  }%
}
\makeatother

\begin{document}

\title{An efficient Hamiltonian-based quantum algorithm for characters of the symmetric group}

\author[1]{Dikshant Rathore}
\author[2]{Leo Zhou}

\affil[1]{\small \it Mani L. Bhaumik Institute for Theoretical Physics, Department of Physics and Astronomy, \authorcr 
\small \it  University of California, Los Angeles, CA 90095, USA}
\affil[2]{\small{\it{Electrical \& Computer Engineering Department, University of California, Los Angeles, CA 90095, USA}}}

\date{October 4, 2026}

\maketitle

\begin{abstract}
Consider a quantum state vector proportional to any given column of the character table of the symmetric group ($S_n$), corresponding to a superposition over the irreps weighted by the character values. It was recently argued that sampling from this state is classically hard under reasonable complexity-theoretic assumptions, while there exists an efficient quantum algorithm for preparing this character state using the quantum Fourier transform (QFT) over $S_n$~\cite{BravyiGossetHavlicekSchatzki2025}. We give an alternative, simpler algorithm to prepare this character state with a sequence of Hamiltonian evolutions, as well as an efficient implementation with reconfigurable qubits that uses only nearest-neighbor gates in 1D.
The gate complexity, accounting for Hamiltonian simulation errors, depends on the choice of the column:
For fixed target state error, the estimated sufficient total gate count ranges from as few as $O(n)$ gates for the $n$-cycle column, to $\widetilde O(n^{1.75})$ gates near the conjecturally classically hard parameter regime, to provably $\widetilde O(n^{2.5})$ gates in the worst case. In contrast, the QFT subroutine in the prior approach uses $\widetilde O(n^3)$ gates. Our gate count estimates come from leading-order state-dependent Trotter error analysis, and the proven worst-case bound relies on a QSVT-based Hamiltonian simulation method;
numerical benchmarks of the Trotter-based protocol suggest that the estimates are conservative. Furthermore, we generalize our protocol to give an explicit algorithm for the quantum character transform (QCT) over $S_n$ using $\widetilde O(n^{2.5})$ gates and $O(n)$ qubits. An application to entanglement entropy of symmetric orbifold conformal field theories is also discussed.

\end{abstract}
\newpage


\tableofcontents
\newpage

\section{Introduction}

Symmetry plays a central role in the description and simulation of quantum
many-body systems. It allows us to resolve the Hilbert space of a system into
its irreducible symmetry sectors. For a symmetry group $G$, its characters are
the fundamental quantities that connect these sectors to the action of its
group elements.

Here we focus on the symmetric group $S_n$ consisting of the $n!$ permutations of $n$
objects.
Its conjugacy classes, labeled by cycle types, classify permutations
up to relabeling.
Its irreducible representations (irreps), labeled by Young
diagrams, describe the possible symmetry sectors of the system. 
The normalized character table can be viewed as a Fourier transform between conjugacy classes and irreducible representations within the space of class functions.

In many-body physics, $S_n$ governs the exchange of identical particles, with its irreps organizing the spin and spatial symmetry sectors of quantum gases. It also appears in replica methods, where $S_n$ permutes the replicas introduced to study disordered systems and many-body entanglement~\cite{MezardParisiVirasoro1987,ZhouNahum2019}. Moreover, the characters of $S_n$ determine the Slater-basis amplitudes of Laughlin states
in the fractional quantum Hall effect~\cite{Dunne1993}.

A concrete setting where characters of $S_n$ have a physical interpretation is in the two-dimensional symmetric orbifold conformal field theories (CFTs).
A symmetric orbifold CFT is formed from $n$ identical copies of a conformal
field theory by gauging the $S_n$ symmetry that permutes them~\cite{DijkgraafMooreVerlindeVerlinde1997}. The gauged theory has a generalized symmetry realized by topological line
operators (defects) labeled by irreducible representations of $S_n$~\cite{Thorngren,GutperleLiRathoreRoumpedakis2024_2}. The action of these defects on twist operators is determined by the
characters of $S_n.$ Recent work has shown that the entanglement entropy in the presence of these
defects can be computed directly from these
characters~\cite{GutperleLiRathoreRoumpedakis2024}. In particular,
for a defect labeled by $\lambda$ placed at the entangling surface, the defect-dependent contribution to the entropy is
\begin{equation}\label{eq:defect-entropy-intro}
  \Delta S_\lambda
  =
  -\sum_{\nu\vdash n}
  \frac{[\chi^\lambda(\nu)]^2}{z_\nu}
  \ln\!\left([\chi^\lambda(\nu)]^2\right)\,,
  \qquad
  z_\nu=\frac{n!}{|C_\nu|}\,,
\end{equation}
where $C_\nu$ is the conjugacy class of cycle type $\nu$, which is specified as a partition of $n$ (denoted as $\nu\vdash n$).
By character orthogonality, the weights $\chi^\lambda(\nu)^2/z_\nu$ form a probability
distribution over cycle types and are precisely the squared entries of a row
of the normalized character table. Characters of $S_n$ also determine expectation values of worldsheet observables in the $\alpha$-states of Hurwitz theory~\cite{BintanjaGesteau2026}.

Despite their broad role, characters of $S_n$ are
computationally difficult. Given an irreducible representation $\lambda$ and a
conjugacy class $\nu$, exactly evaluating the corresponding character $\chi^\lambda(\nu)$ is
at least $\#\mathrm{P}$-hard in the worst
case~\cite{Hepler1994, CrichignoPrakash2024}. Even the simpler task of deciding whether the character is nonnegative is
$\mathrm{PP}$-complete~\cite{Ikenmeyer_2023}.

Quantum computers offer a different form of access to the character table. Rather than evaluating individual entries, one can consider a quantum state
whose amplitudes encode an entire column of the character table. For a
cycle type $\nu=(1^{m_1}\ldots n^{m_n})\vdash n$, where $m_s$ is the number of
cycles of length $s$ (so that $\sum_s sm_s = n$), the normalized column is
the \emph{character column state} 
\begin{equation}\label{eq:chi-state}
  \ket{\chi_\nu}
  :=
  \frac{1}{\sqrt{z_\nu}}
  \sum_{\lambda\vdash n}
  \chi^\lambda(\nu)\,\ket{\lambda}\,,
  \qquad
  z_\nu=\sum_{\lambda\vdash n}[\chi^\lambda(\nu)]^2=\prod_{s=1}^{n}s^{m_s}m_s!\,,
\end{equation}
and measuring it returns an irreducible representation $\lambda$ with
probability
\begin{equation}\label{eq:intro-row-sampling}
  P_\nu(\lambda)
  = \left|\braket{\lambda|\chi_\nu}\right|^2
  = \frac{[\chi^\lambda(\nu)]^2}{z_\nu}\,.
\end{equation}
We call the task of preparing $\ket{\chi_\nu}$ for a given cycle type
\emph{character state preparation} (CSP). Character orthogonality makes the states $\ket{\chi_\nu}$ an orthonormal
basis, so preparing all columns coherently defines a unitary,
\begin{equation}\label{eq:intro-qct}
  \mathrm{QCT}_{S_n}:\ \ket{\nu}\longmapsto\ket{\chi_\nu}\,,
\end{equation}
called the \emph{quantum character transform} (QCT)  over $S_n$ by Bastidas
\emph{et al.}~\cite{Bastidas_2025}. 
The QCT can be used for symmetry-adapted state preparation and phase estimation, in which a many-body state is
projected onto the irreducible representations of its symmetry group. While CSP prepares $\ket{\chi_\nu}$ for one specified cycle type $\nu$, the
QCT performs this preparation coherently on an arbitrary superposition of cycle types. 

Bravyi \emph{et al.}~\cite{BravyiGossetHavlicekSchatzki2025} showed that CSP can be
implemented using the quantum Fourier transform over
$S_n.$ 
They also identified suitable character-sampling problems as candidates for quantum advantage (e.g., sampling from $\ket{\chi_\nu}$ when $\nu$ contains $\Theta(\sqrt{n}/\log n)$ cycles), conditional on an average-case hardness conjecture and the noncollapse of the polynomial hierarchy. 
The best currently known implementation of the quantum Fourier transform over $S_n$ requires $\widetilde O(n^3)$ gates and
$\widetilde O(n^{3/2})$ qubits~\cite{Bruinsma2026}. 
Bastidas \emph{et al.} \cite{Bastidas_2025} introduced the QCT as a primitive for symmetry-adapted quantum algorithms and gave implementations for $S_2$ and $S_3$, but no scalable circuit for the QCT over general $S_n$ was provided. 

In this work, we bypass the quantum Fourier transform by giving a more direct Hamiltonian-based algorithm for CSP (Theorem~\ref{thm:intro-csp}), and extend it coherently to realize the full QCT.
Using QSVT for the Hamiltonian simulation subroutine, our algorithm uses $\widetilde O(n^{5/2})$ gates and $O(n)$
qubits both for CSP in the worst case (Theorem~\ref{thm:intro-qsvt}) and for the full QCT (Theorem~\ref{thm:intro-qct}).
To our knowledge, our $\widetilde O(n^{5/2})$ is the
first explicit polynomial gate bound for the QCT over general $S_n$.
As an application, we show that inverse QCT can be used to efficiently estimate the entanglement entropy \eqref{eq:defect-entropy-intro} in symmetric orbifold CFTs to within $O(\sqrt{n})$ additive error.

We also propose an efficient Trotter-based implementation that is particularly suitable for reconfigurable qubits in 1D or 2D geometry (e.g., neutral atom arrays~\cite{Bluvstein_2023}).
The proposed implementation leverages the fast fermionic permutation protocol with reconfigurable qubits from Maskara \emph{et al.}~\cite{Maskara_2025}, which makes the Hamiltonian simulation gates nearest-neighbor.
For example, the required fermionic permutation can be implemented as a 2D reflection with $O(n)$ gates.
At the $\Theta(\sqrt n/\log n)$-cycle regime proposed as a candidate for quantum advantage in character sampling~\cite{BravyiGossetHavlicekSchatzki2025}, our leading-order state-dependent Trotter error analysis for uniform cycles $\nu=(s^m)$ suggests that $\ket{\chi_\nu}$ can be prepared with
$\widetilde O(n^{7/4}/\delta)$ gates, where $\delta$ is the target state error (Theorem~\ref{thm:intro-trotter}).
As a complementary result, we give a classical algorithm to sample from $P_\nu(\lambda)$ that is efficient when $\nu=(s^m)$ is exactly uniform, but becomes inefficient under perturbations from uniformity (Theorem~\ref{thm:intro-classical}).
Our numerical simulations of the Trotterized CSP circuit up to $n=48$ suggest that the true gate count needed for a fixed state error changes modestly when a uniform cycle type is perturbed, and the analytical bounds overestimate the count for many cycle types by various powers of $n$.
Hence, we believe this Trotter-based implementation of CSP is a promising candidate to realize near-term quantum advantage with reconfigurable qubits.

\subsection{Background and prior work}
\label{sec:intro-prior}

\paragraph{Spin-chain encoding of $S_n$ characters.}
Crichigno and Prakash~\cite{CrichignoPrakash2024} showed that the characters
of $S_n$ appear as amplitudes of states in the XX Heisenberg spin chain. Under the
Jordan-Wigner transformation, the chain of $2n$ sites becomes $2n$ fermionic
modes, and a Young diagram $\lambda\vdash n$ is encoded as the Fock state
$\ket\lambda$ with $n$ occupied modes determined by its shape. For each
cycle length $s$, the current operator
\begin{equation}\label{eq:psi-g}
  J_s
  =
  \sum_{i=1}^{2n-s}a_{i+s}^\dagger a_i
\end{equation}
hops a fermion by $s$ sites. For a cycle type
$\nu=(\nu_1,\ldots,\nu_\ell)$, they showed that
\begin{equation}\label{eq:intro-cp-identity}
  J_{\nu_1}J_{\nu_2}\cdots J_{\nu_\ell}\ket{1^n0^n}
  =
  \sum_{\lambda\vdash n}\chi^\lambda(\nu)\ket\lambda
  =
  \sqrt{z_\nu}\,\ket{\chi_\nu}\,.
\end{equation}
This is the Murnaghan-Nakayama rule in fermionic form: each hop by $s$
sites adds a border strip of length $s$ to the Young diagram, with the sign
supplied by fermionic statistics. The identity reduces CSP to applying a
prescribed sequence of current operators to a product state. The current operators
are not unitary, however, so they cannot be applied directly as gates. Nevertheless, as we show later in Theorem~\ref{thm:intro-csp}, their normalized action can be implemented by Hamiltonian evolution.

\paragraph{Classical and quantum algorithms for CSP.}
Bravyi \emph{et al.}~\cite{BravyiGossetHavlicekSchatzki2025} provided both
a classical and a quantum algorithm for realizing the state $\ket{\chi_\nu}$ based on the above encoding.
The classical algorithm represents each current $J_s$ as a matrix product operator of bond dimension $s+2$ and applies the $\ell$ currents in
\eqref{eq:intro-cp-identity} one at a time. The resulting matrix product
state (MPS) for $\ket{\chi_\nu}$ has bond dimension at most
$D\le\prod_{j=1}^{\ell}(\nu_j+2)\le(n+2)^\ell$, so any character
$\chi^\lambda(\nu)$ can be read off exactly in time $O(n\,\ell\, D^3)=n^{O(\ell)}$.
This is polynomial for a bounded number of cycles and quasipolynomial for
$\ell=O(\log n)$, the typical number of cycles of a random permutation, but
exponential in the worst case.

Their quantum algorithm efficiently prepares $\ket{\chi_\nu}$ for any $\nu$ but, unlike their classical MPS algorithm, does not enable evaluation of individual $\chi^\lambda(\nu)$.
This CSP algorithm works by first preparing a uniform superposition over the conjugacy class $C_\nu$,
then applying the quantum Fourier transform over $S_n$~\cite{Beals_1997},
and finally uncomputing the representation-matrix registers. 
The algorithm runs in $\poly(n)$ time, but no explicit upper bound is known on the exponent; in any case, the QFT subroutine uses
$\widetilde O(n^3)$ gates and $\widetilde O(n^{3/2})$
qubits according to the best known implementation~\cite{Bruinsma2026}.

\paragraph{Character sampling as a candidate for quantum advantage.}
Bravyi \emph{et al.} also studied the classical complexity of sampling from the character distribution $P_\nu(\lambda)=\left|\braket{\lambda|\chi_\nu}\right|^2$. A distribution over a finite set $\Omega$ is
called \emph{granular} if every probability is an integer multiple of
$\gamma/|\Omega|$ for some $\gamma\geq1$. Since the characters are integers and $\lambda\vdash n$, $P_\nu(\lambda)$ of Eq.~\eqref{eq:intro-row-sampling} is granular over the $p(n)= \exp(\pi \sqrt{2n/3}-O(\log n))$ partitions of $n$ whenever
$z_\nu\leq p(n)$, with granularity parameter $\gamma=p(n)/z_\nu$. 
Bravyi \emph{et al.} prove that this granularity condition holds for cycle types satisfying
\begin{equation}\label{eq:intro-bghs-regime}
  \ell(\nu)\leq A\frac{\sqrt n}{\log n}\,,
\end{equation}
where $\ell(\nu)$ is the number of cycles and $A\approx \pi\sqrt{2/3}$ is a universal positive constant.

For such cycle types, a classical sampler with $\ell_1$ error $\varepsilon$
can be converted into a $\mathrm{BPP}^{\mathrm{NP}}$ algorithm that computes
$P_\nu(\lambda)$ exactly for at least a $1-20\varepsilon$ fraction of
partitions $\lambda$. Classical
hardness of sampling would, therefore, follow if the exact computation of
$P_\nu(\lambda)$ were hard on average over $\lambda$, together with the
assumption that the polynomial hierarchy does not collapse. This
average-case hardness remains conjectural. Their result, therefore, identifies
cycle types with $O(\sqrt n/\log n)$ cycles as a candidate regime for
quantum advantage.

For the dual problem, in which one fixes $\lambda$ and samples
$g\in S_n$ with probability $[\chi^\lambda(g)]^2/n!$, granularity holds
automatically with $\gamma=1$. Bravyi \emph{et al.} regard this direction as
less promising for quantum advantage, 
because their MPS algorithm evaluates the character in quasipolynomial time on all but a vanishing fraction of uniformly random permutations, which makes the average-case hardness conjecture less likely to hold for this problem.

\subsection{Summary of main results}
\label{sec:intro-main-results}
As shown in Eq.~\eqref{eq:intro-cp-identity}, CSP can be reduced to applying the current
operators $J_{\nu_i}$, one for each cycle of $\nu$, to the reference state
$\ket{1^n0^n}$. However, the current operators are not unitary. To implement the transfer mechanism $\ket{\phi} \to J_s\ket{\phi}$, we add a single ancilla qubit in $\ket{0}_a$ and evolve with the
Hamiltonian
\begin{equation}\label{eq:intro-Hs}
  H_s =  \ket{1}\!\bra{0}_a \otimes   J_s
  + \ket{0}\!\bra{1}_a \otimes   J_s^\dagger \,.
\end{equation}
On the states that the algorithm visits, $J_s^\dagger J_s$ has a known
integer eigenvalue, so $e^{-iH_st}$ is an exact two-level Rabi rotation between $\ket{0}_a \otimes  \ket{\phi}$ and $ \ket{1}_a \otimes J_s \ket{\phi}$. This leads to a Hamiltonian-based algorithm for preparing $\ket{\chi_\nu}$:
\begin{theorem}[Character state preparation]
\label{thm:intro-csp}
Given a cycle type $\nu = (\nu_1,\ldots, \nu_\ell)\vdash n$, with entries obeying $\nu_1 \ge \nu_2 \ge \ldots \ge \nu_\ell$. 
Let $r_j = \left|\{ k \le j: \nu_k = \nu_j\}\right|$ and $t_j = \pi/(2\sqrt{\nu_j r_j})$. Then
\begin{equation}
\ket{0}_a \otimes \ket{\chi_\nu} = (iX_a \otimes\Id) e^{-it_\ell  H_{\nu_\ell}} \cdots (iX_a \otimes\Id) e^{-it_2 H_{\nu_2}} (iX_a \otimes\Id) e^{-it_1 H_{\nu_1}} \big( \ket{0}_a \otimes \ket{1^n 0^n} \big)\,.
\end{equation}
\end{theorem}

The CSP circuit is shown in Figure~\ref{fig:intro-circuits}(a). The proof of the above theorem is in Section~\ref{sec:csp}.
The algorithm needs $2n+1$ qubits and $\ell$ evolutions with known times,
and its cost is that of simulating $e^{-iH_st}$. We give two implementations of this simulation, visualized in  Figure~\ref{fig:intro-circuits}(b) and (c) respectively. The first is based on Trotter decomposition.

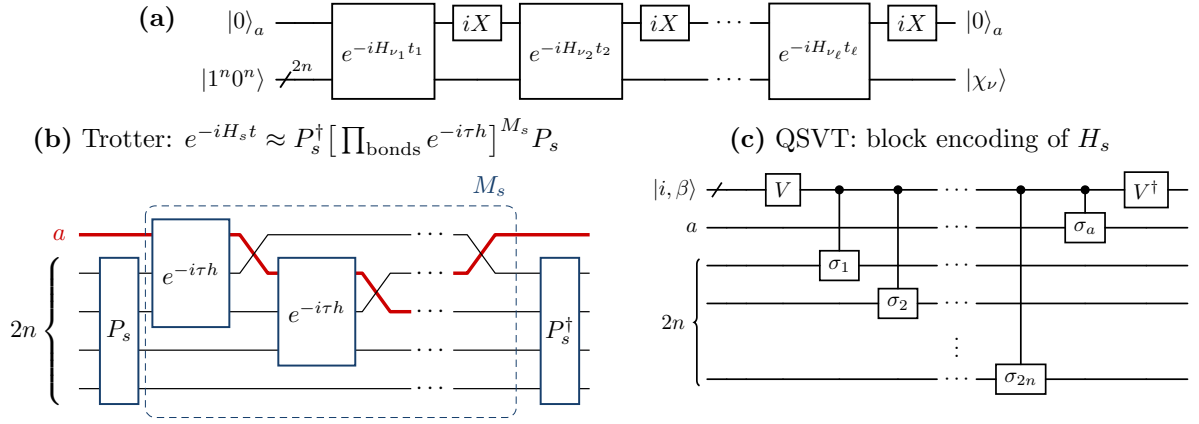
\begin{figure}[t]
\centering
{\small\textbf{(a)}}\hspace{6pt}\raisebox{-0.5\height}{\resizebox{0.66\linewidth}{!}{%
\begin{quantikz}[row sep={0.85cm,between origins}, column sep=0.28cm, font=\small, transparent, every gate/.style={draw=navy,fill=white,thick,inner xsep=4pt}, wire types={q,q}]
\lstick{$\ket{0}_a$} & &
  & \gate[2]{e^{-iH_{\nu_1}t_1}} & \gate{iX}
  & \gate[2]{e^{-iH_{\nu_2}t_2}} & \gate{iX}
  & \ \cdots\
  & \gate[2]{e^{-iH_{\nu_\ell}t_\ell}} & \gate{iX}
  & \rstick{$\ket{0}_a$} \\
\lstick{$\ket{1^n0^n}$} & \qwbundle{2n} &
  & & & & & \ \cdots\ & & & \rstick{$\ket{\chi_\nu}$}
\end{quantikz}}}\\[3pt]
\begin{minipage}[t]{0.55\linewidth}
\centering
{\small\textbf{(b)} Trotter: $e^{-iH_st}\approx P_s^\dagger\big[\prod_{\mathrm{bonds}}e^{-i\tau h}\big]^{M_s}P_s$}\\[3pt]
\resizebox{0.86\linewidth}{!}{%
\begin{tikzpicture}[x=1cm, y=0.55cm, font=\small,
  anc/.style={draw=red!80!black, line width=1.4pt, line join=round},
  sys/.style={draw=black, line width=0.5pt, line join=round},
  box/.style={draw=navy, fill=white, thick}]
\draw[anc] (0,0) -- (2.3,0) -- (2.7,-1) -- (4.05,-1) -- (4.45,-2) -- (4.75,-2);
\draw[anc] (5.35,-1) -- (5.55,-1) -- (5.95,0) -- (7.3,0);
\draw[sys] (0,-1) -- (2.3,-1) -- (2.7,0) -- (4.75,0);
\draw[sys] (5.35,0) -- (5.55,0) -- (5.95,-1) -- (7.3,-1);
\draw[sys] (0,-2) -- (4.05,-2) -- (4.45,-1) -- (4.75,-1);
\draw[sys] (5.35,-2) -- (7.3,-2);
\foreach \y in {-3,-4} {\draw[sys] (0,\y) -- (4.75,\y); \draw[sys] (5.35,\y) -- (7.3,\y);}
\foreach \y in {0,-1,-2,-3,-4} \node at (5.05,\y) {$\cdots$};
\draw[box] (0.3,-0.6) rectangle (0.85,-4.4); \node at (0.575,-2.5) {$P_s$};
\draw[box] (6.6,-0.6) rectangle (7.15,-4.4); \node at (6.875,-2.5) {$P_s^\dagger$};
\draw[box] (1.05,0.4) rectangle (2.15,-2.4); \node at (1.6,-1) {$e^{-i\tau h}$};
\draw[box] (2.85,-0.6) rectangle (3.95,-3.4); \node at (3.4,-2) {$e^{-i\tau h}$};
\draw[navy, densely dashed, rounded corners] (0.95,0.75) rectangle (6.25,-4.75);
\node[navy, anchor=south east] at (6.25,0.75) {$M_s$};
\node[anchor=east, text=red!80!black] at (-0.05,0) {$a$};
\node[anchor=east] at (-0.05,-2.5) {$2n\left\{\vrule height 0.95cm depth 0.95cm width 0pt\right.$};
\end{tikzpicture}}
\end{minipage}%
\begin{minipage}[t]{0.45\linewidth}
\centering
{\small\textbf{(c)} QSVT: block encoding of $H_s$}\\[3pt]
\resizebox{0.98\linewidth}{!}{%
\begin{quantikz}[row sep={0.62cm,between origins}, column sep=0.32cm, font=\small, transparent, every gate/.style={draw=navy,fill=white,thick,inner xsep=4pt}, wire types={q,q,q,q,n,q}]
\lstick{$\ket{i,\beta}$} & \qwbundle{} & & \gate{V} & \ctrl{2} & \ctrl{3} & \ \cdots\ & \ctrl{5} & \ctrl{1} & \gate{V^\dagger} & \\
\lstick{$a$} & & & & & & \ \cdots\ & & \gate{\sigma_a} & & \\
\lstick[4]{$2n$} & & & & \gate{\sigma_1} & & \ \cdots\ & & & & \\
 & & & & & \gate{\sigma_2} & \ \cdots\ & & & & \\
 & & & & & & \ \vdots\ & & & & \\
 & & & & & & \ \cdots\ & \gate{\sigma_{2n}} & & &
\end{quantikz}}
\end{minipage}
\caption{(a) Our CSP circuit: $\ell(\nu)$ exactly timed Hamiltonian evolutions by $H_{\nu_k}$ in \eqref{eq:intro-Hs}, each followed by an
$iX$ gate on the ancilla. 
(b) One Trotterized transfer on reconfigurable
qubits: after the regrouping $P_s$, the $2n-s$ bond terms are
nearest-neighbor, and the ancilla is swapped along the chain so that each
bond gate acts on three adjacent qubits. (c) Block encoding of $H_s$: $V$
puts the register $\ket{i,\beta}$ in uniform superposition over the
Pauli-string labels, and site $k$ receives the Pauli
$\sigma_k\in\{I,X,Y,Z\}$ the register selects: $X$ or $Y$ at $i$ and $i+s$,
$Z$ in between, $I$ elsewhere.
}
\label{fig:intro-circuits}
\end{figure}

\paragraph{Trotter with reconfigurable qubits.}
Reordering the $2n$ modes by their residue modulo $s$ makes every hop
nearest-neighbor on one of $s$ chains. In the reordered basis, $H_s$ is a
sum of $2n-s$ local bond terms, each coupling two adjacent sites to the
ancilla, and one first-order Trotter step costs $O(n)$ gates. The
reordering is a fermionic permutation $P_s$, which can cost $\Theta(n^2)$ swap gates
on a fixed one-dimensional layout. We show in Lemma~\ref{lem:reflection}
that $P_s$ is precisely the two-dimensional reflection of Maskara
\emph{et al.}~\cite{Maskara_2025}. Their reconfigurable-qubit protocol
implements this structured fermionic permutation with $O(n)$ Clifford gates
in constant depth, using $O(n)$ auxiliary qubits, mid-circuit measurement,
and classical feedforward. The operator-norm Trotter
bound then gives $O(n^4\log^2n/\delta)$ gates for any cycle type at state
error $\delta$ (Section~\ref{sec:trotter-error}). This bound is loose,
because it maximizes the commutator error over all states while the
algorithm visits only a small family of them. For uniform cycle types, we evaluate the leading-order error on those states directly.

\begin{theorem}[Informal: leading-order state-dependent Trotter bound for CSP]
\label{thm:intro-trotter}
For a uniform cycle type $\nu=(s^m)$ with $n=sm$, allocating
\begin{equation}\label{eq:intro-state-dependent}
  O\!\left[
    \frac{n}{\delta}
    \left(m^{3/2}+\frac{m^2}{s^{3/2}}\right)
  \right]
\end{equation}
gates makes the leading $1/M$ contribution to the first-order Trotter error
at most $\delta$ on the states visited by the algorithm. At $m=\Theta(\sqrt n/\log n)$, the cycle count of the Bravyi et al. granularity
regime \eqref{eq:intro-bghs-regime}, this is $\widetilde O(n^{7/4}/\delta)$.
\end{theorem}
For details, see Section~\ref{sec:state-dependent} and Appendix~\ref{app:state-dependent}. Eq.~\eqref{eq:intro-state-dependent} is a state-dependent bound on the sufficient gate count to make the leading-order Trotter error small. Based on numerical simulations up to $n=48$, we believe the bound to be a loose overestimate on the actual gate count required (see Section~\ref{sec:numerics}).

\paragraph{QSVT-based implementation.}
In the second implementation, we write $H_s$ as a sum of $4(2n-s)$ Pauli strings. Each string connects sites $i$ and $i+s$ with a product of $Z$, and acts as $X$ or $Y$ on the endpoints in addition to the ancilla.
A register of $O(\log n)$ qubits in uniform
superposition over the string labels selects which string is applied, at a
cost of $O(n\log n)$ gates (Proposition~\ref{prop:structured-select}). This
block-encodes $H_s$ with normalization $2n-s$, and we apply the quantum singular value
transformation (QSVT)~\cite{Gilyen_2019} to simulate each transfer with a number of queries that grows only logarithmically in $1/\delta$.

\begin{theorem}[QSVT-based CSP; also Theorem~\ref{thm:qsvt-csp}]
\label{thm:intro-qsvt}
For every cycle type $\nu\vdash n$ and desired state error $\delta$, CSP can be
implemented with
\begin{equation}\label{eq:intro-qsvt}
  O\!\left(
    n^{5/2}\log n
    +
    n^2\log n\log\frac{n}{\delta}
  \right)
  =
  \widetilde O(n^{5/2})
\end{equation}
one- and two-qubit gates and $2n+O(\log n)$ qubits, where the simplified $\widetilde O(n^{5/2})$ form assumes fixed or inverse-polynomial $\delta$.
\end{theorem}

\paragraph{The quantum character transform.}
To implement the QCT \eqref{eq:intro-qct}, we run our QSVT-based CSP coherently controlled by a superposition of cycle types, $\ket{\nu}\ket{1^n0^n}\to \ket{\nu}\ket{\chi_\nu}$, and then erase the cycle types to get $\ket{0}\ket{\chi_\nu}$ in order to maintain coherence in the output superposition of the character states.
The erasure can be done in concert with the CSP subroutine: After the transfers at length $s$, the system is an eigenvector of $\mathcal N_s=J_sJ_s^\dagger/s$ with eigenvalue
$m_s$, so the $\ket{\nu}$ register can be erased by a sequence of phase estimation on $\mathcal N_s$.
Simulating $\mathcal N_s$ directly costs $\widetilde O(n^3)$ gates, because
its block encoding has normalization $(2n-s)^2/s$. We avoid this by
observing that $H_s$ has eigenvalues of magnitude at most $\sqrt n$ on the
states the algorithm visits, and amplifying them with uniform singular value
amplification~\cite{Gilyen_2019} before squaring. All in all, the gate count sufficient for implementing QCT is similar to the QSVT-based implementation of CSP above.

\begin{theorem}[Quantum character transform; also Theorem~\ref{thm:qct-total}]
\label{thm:intro-qct}
On the input subspace of valid cycle-type labels with clean work registers,
$\mathrm{QCT}_{S_n}$ can be implemented using
$\widetilde O(n^{5/2}\log \delta^{-1})$ gates and $O(n)$ qubits with operator-norm error $\delta$.
\end{theorem}
The extra register for encoding cycle types $\ket{\nu}$ needs at most $\sum_{s=1}^n O(\log (n/s)) = O(n)$  qubits [see Eq.~\eqref{eq:register-count}], and thus the qubit count remains $O(n)$.

\paragraph{Classically tractable cycle types.}
We are interested in the potential of our efficient CSP/QCT algorithms for realizing quantum advantage by sampling from $\ket{\chi_\nu}$. The Bravyi \emph{et al.} argument for quantum advantage requires cycle types $\nu$
in the regime \eqref{eq:intro-bghs-regime}, which ensures granularity but not average-case hardness. We show that the number of cycles alone does not determine classical hardness due to the following classical algorithm:

\begin{theorem}[Classical samplers; see Theorems~\ref{thm:equal-cycle-classical}
and~\ref{thm:almost-equal-classical}]
\label{thm:intro-classical}
For uniform cycles $\nu=(s^m)$, the distribution $P_\nu$ of
\eqref{eq:intro-row-sampling} can be sampled exactly by a classical
algorithm in expected time $\widetilde O(n^2)$, and each probability
$P_\nu(\lambda)$ can be computed exactly in polynomial time. For
almost-uniform cycles $\nu=(s^m,r_1,\ldots,r_b)$, an exact rejection
sampler runs in expected time $n^{4b+O(1)}$, and each probability can be
computed in time $n^{2b+O(1)}$.
\end{theorem}
The uniform-cycle sampler is built from the $s$-quotient of $\lambda$: the
sizes of its $s$ components are drawn from a multinomial distribution, and
each component from the Plancherel measure by Robinson-Schensted insertion
of a random permutation. Uniform cycles with $\Theta(\sqrt n/\log n)$ cycles
lie inside the regime \eqref{eq:intro-bghs-regime} and are nevertheless
classically tractable, as are perturbations of them with a constant number
of exceptional cycles. Any superpolynomial quantum advantage in this
neighborhood must therefore come from a number of exceptional cycles that
grows with $n$. We do not have a state-dependent Trotter bound for such cycle types; nevertheless, numerical benchmarking up to $n=48$ suggests that the required Trotter steps change only modestly with the number of exceptional cycles (see Figure~\ref{fig:perturbation}).

\paragraph{Application: entanglement entropy in CFTs.}
Applying the inverse QCT to an irreducible representation $\ket\lambda$ and
measuring returns a cycle type $\nu$ with probability
$P_\lambda(\nu)=[\chi^\lambda(\nu)]^2/z_\nu$, which is the weight appearing
in the defect entropy \eqref{eq:defect-entropy-intro}. Section~\ref{sec:rows}
shows that samples of $\nu$ can be used to estimate $\Delta S_\lambda$ with additive error at most
$\ln p(n) \sim O(\sqrt n)$, where $p(n)$ is the number of partitions of integer $n$.
On the other hand, numerical evaluations of $\Delta S_\lambda$ for all possible $\lambda$ up to $n=30$ suggest that typical values of $|\Delta S_\lambda|$ exceed $\ln p(n)$ at large $n$, and the growth of its magnitude is consistent with $\Theta(n)$ (see Figure~\ref{fig:defect-entropy-growth}).
This suggests that for a generic $\lambda$, our quantum algorithm gives an efficient estimator of $\Delta S_\lambda$ with relative error that vanishes with $n$. 

\subsection{Paper organization}
Section~\ref{sec:background} reviews the symmetric group, its character
table, and the fermionic encoding. Section~\ref{sec:csp} describes how CSP can be realized with Hamiltonian evolutions, based on a transfer mechanism and a ladder property of $J_s$.
Section~\ref{sec:ham-sim} implements the Hamiltonian evolutions $e^{-iH_st}$ in two ways: (1) using
Trotter decomposition on reconfigurable qubits, with operator-norm and state-dependent error bounds, as well as numerical benchmarks, and (2) using block-encoding and QSVT, which give the best-known provable worst-case bound. Section~\ref{sec:qct} gives a generalization of our algorithm to implement the QCT, complete with QSVT-based implementation and resource estimate, in addition to a concrete application to the defect entropy of symmetric orbifold CFTs. 
Section~\ref{sec:classical-hardness} contains the classical samplers for uniform and almost-uniform cycles.
Detailed proofs of the ladder property and of the state-dependent Trotter bounds are in the appendices.

\section{Symmetric group character states}\label{sec:background}

This section collects the background for the rest of the paper. We recall
the symmetric group and its character table, define the character column
states and the quantum character transform precisely, and describe the
fermionic encoding of the character table on which our algorithm is built.

\subsection{The symmetric group and its characters}
\label{sec:sn}

The symmetric group $S_n$ is the group of permutations of $n$ objects, and
$|S_n|=n!$. Every permutation decomposes into disjoint cycles, and its cycle
lengths in nonincreasing order form its \emph{cycle type}. For example,
$(1\,3\,5)(2)(4)\in S_5$ has cycle type $(3,1,1)=(3,1^2)$.
In general, a cycle type $\nu$ is a
\emph{partition} of $n$, a nonincreasing sequence of positive integers
that sum to $n$:
\begin{equation}
  \nu=(\nu_1\ge\nu_2\ge\cdots\ge\nu_\ell)\,,
  \qquad
  \sum_{j=1}^{\ell}\nu_j=n\,,
\end{equation}
written $\nu\vdash n$. We write $\ell(\nu)$ for the number of cycles (or parts). When several cycles share a length, we record multiplicities and write
\begin{equation}
  \nu=(n^{m_n}\cdots 2^{m_2}1^{m_1})\,,
  \qquad
  \sum_{s=1}^{n}s\,m_s=n\,,
  \label{eq:multiplicity-notation}
\end{equation}
where $m_s$ is the number of cycles of length $s$. In this notation the cycle
type $(3,2)$ is $(3^12^1)$, and the partitions of $4$ are $(4)$, $(3,1)$,
$(2^2)$, $(2,1^2)$, and $(1^4)$.

The conjugacy class of $g\in S_n$ is $C(g)=\{hgh^{-1}:h\in S_n\}$.
Conjugating by $h$ relabels the objects on which $g$ acts, which changes the
contents of the cycles of $g$ but not their lengths, and any two permutations
with the same cycle type are related by such a relabeling. The conjugacy
classes of $S_n$ are therefore in bijection with the partitions of $n$, and
we write $C_\nu$ for the class of cycle type $\nu$.

Partitions of $n$ also label the irreducible representations of $S_n$. For
$n=3$, the three partitions correspond to the trivial representation $(3)$,
the standard representation $(2,1)$, and the sign representation $(1^3)$.
In this work, we do not need the representations themselves, only their
traces. Writing $\rho^\lambda(g)$ for the matrix representing $g$ in the
irrep $\lambda$, the \emph{character} of $g$ in that irrep is
\begin{equation}
  \chi^\lambda(g)=\operatorname{tr}\rho^\lambda(g)\,.
\end{equation}
Traces are invariant under conjugation, so the character depends on $g$ only
through its cycle type, and we write $\chi^\lambda(\nu)$. The characters of
$S_n$ are integers.

Since partitions of $n$ index both the conjugacy classes and the irreps, the
characters assemble into a square table. For $n=3$,
\begin{center}
\begin{tabular}{c|ccc}
 & $(1^3)$ & $(2,1)$ & $(3)$ \\
\hline
$\chi^{(3)}$   & $1$ & $1$  & $1$  \\
$\chi^{(2,1)}$ & $2$ & $0$  & $-1$ \\
$\chi^{(1^3)}$ & $1$ & $-1$ & $1$
\end{tabular}
\end{center}
with rows labeled by irreps and columns by cycle types. The first column
records the dimensions $d_\lambda=\chi^\lambda(1^n)$ of the irreps.

Distinct columns of the table are orthogonal, and the norm of a column is
set by the size of the corresponding conjugacy class,
\begin{equation}
  \sum_{\lambda\vdash n}\chi^{\lambda}(\nu)\,\chi^{\lambda}(\mu)
  =
  \delta_{\nu\mu}\,z_\nu\,,
  \qquad
  z_\nu=\prod_{s=1}^{n}s^{m_s}m_s!=\frac{n!}{|C_\nu|}\,,
  \label{eq:colorth}
\end{equation}
where $z_\nu$ is the order of the centralizer of any element of $C_\nu$. The
three columns of the table above have squared norms $6$, $2$, and $3$.

\subsection{Character column states and the quantum character transform}
\label{sec:columnstate}

Setting $\mu=\nu$ in \eqref{eq:colorth} shows that the column of the
character table labeled by $\nu$ has squared norm $z_\nu$. Normalizing it
gives the character column state $\ket{\chi_\nu}$ of Eq.~\eqref{eq:chi-state},
\begin{equation}
  \ket{\chi_\nu}
  =
  \frac{1}{\sqrt{z_\nu}}
  \sum_{\lambda\vdash n}\chi^\lambda(\nu)\,\ket{\lambda}\,,
  \label{eq:target}
\end{equation}
where $\ket\lambda$ is a computational basis state labeled by the irrep
$\lambda$. Measuring $\ket{\chi_\nu}$ in this basis returns $\lambda$ with
probability
\begin{equation}
  P_\nu(\lambda)
  =
  \frac{[\chi^\lambda(\nu)]^2}{z_\nu}\,.
  \label{eq:rowsampling}
\end{equation}
Although $[\chi^\lambda(\nu)]^2$ is a nonnegative integer, Ikenmeyer \emph{et al.} proved that the function $(\lambda,\nu)\mapsto[\chi^\lambda(\nu)]^2$
is not in $\#\mathrm{P}$ unless the polynomial hierarchy collapses to its
second level~\cite{Ikenmeyer_2023}. 
This rules out a counting interpretation
of the probabilities \eqref{eq:rowsampling}, but not a quantum algorithm
that samples from them without computing them.

For a finite group $G$ with conjugacy classes $C$ and irreps $\Gamma$, the
quantum character transform of Bastidas \emph{et al.}~\cite{Bastidas_2025}
is
\begin{equation}
  \mathrm{QCT}_G
  =
  \sum_{C}\sum_{\Gamma}
  \sqrt{\frac{|C|}{|G|}}\,
  \chi_\Gamma(C)\,
  \ket{\Gamma}\!\bra{C}\,.
  \label{eq:qctgeneral}
\end{equation}
Character orthogonality makes its columns orthonormal, so it is unitary. For
$G=S_n$ the prefactor is $\sqrt{|C_\nu|/n!}=1/\sqrt{z_\nu}$ by
\eqref{eq:colorth}, and
\begin{equation}
  \mathrm{QCT}_{S_n}:\quad  \ket{\nu}\longmapsto\ket{\chi_\nu}\,.
  \label{eq:qct}
\end{equation}
Character state preparation for a fixed $\nu$ produces one column of this
unitary. Implementing the preparation coherently over all cycle types gives
the full transform.

\subsection{The fermionic encoding}
\label{sec:encoding}

Characters of $S_n$ can be read off as amplitudes of a state on an open XX
spin chain. This encoding was discovered by Crichigno and
Prakash~\cite{CrichignoPrakash2024}, and we follow its presentation in Bravyi
\emph{et al.}~\cite{BravyiGossetHavlicekSchatzki2025}.

The chain has $2n$ sites. Under the Jordan--Wigner transformation its states
are those of $2n$ fermionic modes with operators $a_i,a_i^\dagger$ obeying
\begin{equation}
  \{a_i,a_j^\dagger\}=\delta_{ij}\,,
  \qquad
  \{a_i,a_j\}=\{a_i^\dagger,a_j^\dagger\}=0\,,
\end{equation}
and we write $\hat n_i=a_i^\dagger a_i$. Everything below takes place in the
sector of $n$ particles, the space of $2n$-bit strings of Hamming weight $n$.

A partition $\lambda\vdash n$ is encoded as the Fock state $\ket\lambda$ with
occupied modes
\begin{equation}
  \operatorname{supp}(\lambda)
  =
  \{\lambda_j+n-j+1 : j=1,\dots,n\}\,,
  \label{eq:maya}
\end{equation}
where $\lambda$ has been padded with zeros to $n$ parts. The parts are
nonincreasing while $n-j+1$ strictly decreases, so the $n$ values in
\eqref{eq:maya} are distinct and lie in $\{1,\dots,2n\}$. Equivalently, the
bit string is the boundary path of the Young diagram of $\lambda$ inside an
$n\times n$ box, with a step up recorded as $1$ and a step to the right as
$0$, as in Figure~\ref{fig:maya}. The empty partition gives the \emph{seed
state}
\begin{equation}
  \seed=\ket{\varnothing}=\ket{1^n0^n}\,,
\end{equation}
with modes $1,\dots,n$ occupied and modes $n+1,\dots,2n$ empty.

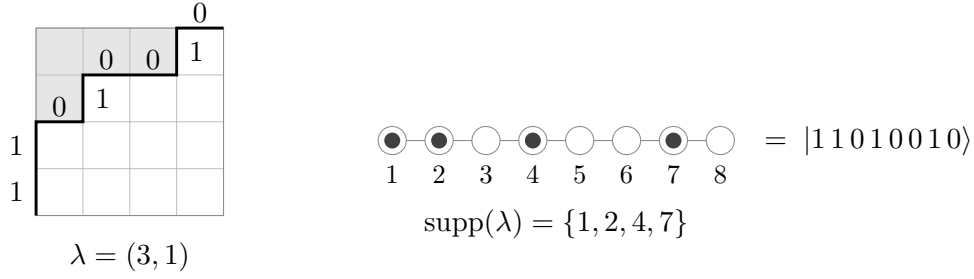
\begin{figure}[t]
\centering
\begin{tikzpicture}[scale=0.62]

\fill[black!10] (0,3) rectangle (3,4);
\fill[black!10] (0,2) rectangle (1,3);

\draw[black!25, very thin] (0,0) grid (4,4);
\draw[black!45] (0,0) rectangle (4,4);

\draw[line width=1.1pt]
  (0,0) -- (0,2) -- (1,2) -- (1,3) -- (3,3) -- (3,4) -- (4,4);

\node at (-0.42,0.5)  {$1$};
\node at (-0.42,1.5)  {$1$};
\node at ( 0.5,2.32)  {$0$};
\node at ( 1.42,2.5)  {$1$};
\node at ( 1.5,3.32)  {$0$};
\node at ( 2.5,3.32)  {$0$};
\node at ( 3.42,3.5)  {$1$};
\node at ( 3.5,4.32)  {$0$};

\node[anchor=north] at (2,-0.35) {$\lambda = (3,1)$};

\begin{scope}[xshift=6.6cm, yshift=1.6cm]
  \foreach \i/\occ in {1/1,2/1,3/0,4/1,5/0,6/0,7/1,8/0} {
    \draw[black!45] (\i,0) circle (0.30);
    \node[font=\small] at (\i,-0.72) {$\i$};
  }
  \foreach \i in {1,2,4,7} { \fill[black!75] (\i,0) circle (0.17); }
  \foreach \i in {1,...,7} { \draw[black!45] (\i+0.30,0) -- (\i+0.70,0); }

  \node[anchor=west] at (8.55,0) {$\;=\;|1\,1\,0\,1\,0\,0\,1\,0\rangle$};
  \node[anchor=north] at (4.5,-1.25)
    {$\operatorname{supp}(\lambda) = \{1,2,4,7\}$};
\end{scope}

\end{tikzpicture}
\caption{The encoding \eqref{eq:maya} for $\lambda = (3,1)$ at $n = 4$. The
boundary of the Young diagram inside the $4\times 4$ box is traced from the
bottom left corner to the top right corner. Each step up is recorded as a $1$
and each step to the right as a $0$, giving a bit string of length $2n$ and
Hamming weight $n$. The ones sit at the modes $\lambda_j + n - j + 1$, here
$\{1,2,4,7\}$ after padding $\lambda$ to $(3,1,0,0)$. The empty partition gives
$|\mathrm{seed}\rangle = |1^n0^n\rangle$.}
\label{fig:maya}
\end{figure}

The operator that moves between these states is the \emph{current operator}
\begin{equation}
  J_s
  :=
  \sum_{i=1}^{2n-s}a_{i+s}^\dagger a_i\,,
  \qquad
  s=1,\dots,n\,,
  \label{eq:Js}
\end{equation}
whose Hermitian part $J_1+J_1^\dagger$ is the Hamiltonian of the open XX
chain.\footnote{The name follows \cite{CrichignoPrakash2024}, where the
generating function $\exp\big(\sum_kt^kJ_k/k\big)$ is identified as a vertex
operator.} The current operators are not unitary, and on the open chain a
current does not commute with its adjoint,
\begin{equation}
  [J_s,J_s^\dagger]\neq0\,.
  \label{eq:JJdag}
\end{equation}
Currents of different lengths do commute,
\begin{equation}
  [J_s,J_{s'}]=0\,,
  \label{eq:Jcommute}
\end{equation}
so for a cycle type $\nu=(\nu_1,\dots,\nu_\ell)$ the product
$J_\nu:=J_{\nu_1}J_{\nu_2}\cdots J_{\nu_\ell}$ does not depend on the order
of the factors.

Applying $J_\nu$ to the seed produces an entire column of the character
table,
\begin{equation}
  J_\nu\,\seed
  =
  \sum_{\lambda\vdash n}\chi^\lambda(\nu)\,\ket{\lambda}\,,
  \label{eq:identity}
\end{equation}
which is the identity \eqref{eq:intro-cp-identity} of the introduction. It
is the fermionic form of the Murnaghan-Nakayama
rule~\cite{CrichignoPrakash2024}; Bravyi \emph{et al.} derive it from the
Frobenius character formula, in which $\chi^\lambda(\nu)$ is the coefficient
of a monomial in a product of power sums times a Vandermonde determinant.
Comparing with \eqref{eq:target} and \eqref{eq:colorth},
\begin{equation}
  \ket{\chi_\nu}
  =
  \frac{J_\nu\,\seed}{\sqrt{z_\nu}}\,,
  \qquad
  \big\|J_\nu\,\seed\big\|^2=z_\nu\,,
  \label{eq:targetfromseed}
\end{equation}
so preparing $\ket{\chi_\nu}$ amounts to applying $J_\nu$ to the seed and
normalizing.

Consider $n=3$ as an illustrative example. The seed state is $\seed=\ket{111000}$ and the three partitions encode
as $\ket{(3)}=\ket{110001}$, $\ket{(2,1)}=\ket{101010}$, and
$\ket{(1^3)}=\ket{011100}$. The three cycle types give $J_{(3)}=J_3$,
$J_{(2,1)}=J_2J_1$, and $J_{(1^3)}=J_1^3$, and
\begin{align}
  J_3\,\seed
    &=\ket{(3)}-\ket{(2,1)}+\ket{(1^3)}\,, \nonumber\\
  J_2J_1\,\seed
    &=\ket{(3)}-\ket{(1^3)}\,, \nonumber\\
  J_1^3\,\seed
    &=\ket{(3)}+2\,\ket{(2,1)}+\ket{(1^3)}\,,
\end{align}
reproducing the three columns of the character table shown in
Section~\ref{sec:sn}.

\section{Character State Preparation}
\label{sec:csp}

By Eq.~\eqref{eq:targetfromseed}, preparing the character column state
$\ket{\chi_\nu}$ amounts to applying $J_\nu$ to the seed state $\seed=\ket{1^n0^n}$ and normalizing. The
current operators are not unitary, so $J_\nu$ cannot be applied directly as a
quantum gate, and $J_\nu\seed$ does not have unit norm.
Instead, we realize the normalized map $\seed \mapsto \ket{\chi_\nu}$
by unitary evolution on the system together with a single ancilla qubit,
``depositing'' the current operators $J_{\nu_1},\ldots,J_{\nu_\ell}$ on the state one cycle length at a time in nonincreasing order of length.
Each deposit is a single exact rotation driven by
\begin{equation}
  H_s = \ket{1}\bra{0}_a \otimes J_s + \ket{0}\bra{1}_a \otimes J_s^\dag\,,
  \label{eq:Hs}
\end{equation}
in an effective two-level system, with the evolution time chosen so that the transfer is complete. The algorithm
uses $2n$ system qubits and one ancilla and prepares $\ket{\chi_\nu}$
deterministically, without post-selection.

\subsection{The transfer mechanism}
\label{sec:transfer}

We start by describing the key building block: a transfer mechanism that takes $\ket{\phi}\to J_s\ket{\phi}/ \left \| J_s\ket{\phi}\right\|$, which $H_s$ implements via an effective two-level rotation. 

\begin{lemma}[Transfer mechanism]
\label{lem:transfer}
Let $A$ be any operator and
suppose $\ket{\phi}$ is an eigenvector of $A^\dagger A$ with eigenvalue
$\zeta > 0$, i.e.,
\begin{equation}
  A^\dagger A \ket{\phi} = \zeta \ket{\phi}\,,
  \label{eq:transfer-eigen}
\end{equation}
Set $H = \ket{1}\bra{0}_a \otimes A + \ket{0}\bra{1}_a \otimes A^\dagger$ and $t^* = \pi/(2\sqrt{\zeta})$, then
\begin{equation}
(iX_a \otimes\Id) e^{-it^*H}  \ket{0}_a\otimes \ket{\phi} = \ket{0}_a \otimes \frac{A\ket{\phi}}{\sqrt{\zeta}}\,.
\end{equation}
\end{lemma}
\begin{proof}
Write $\ket{v} := A\ket{\phi}/\sqrt{\zeta}$ which is a normalized state.
On the two-dimensional subspace spanned by $\ket{0}_a \ket{\phi}$ and
$\ket{1}_a \ket{v}$, the Hamiltonian acts as $\sqrt{\zeta} X$: it sends
$\ket{0}_a \ket{\phi} \mapsto \sqrt\zeta \ket{1}_a \ket{v}$ due to
\eqref{eq:transfer-eigen}, and $\ket{1}_a \ket{v} \mapsto \sqrt\zeta \ket{0}_a \ket{\phi}$
because $A^\dagger\ket{v} = A^\dagger A\ket{\phi}/\sqrt\zeta = \sqrt\zeta\ket{\phi}$.
The evolution stays in this subspace, and thus
\begin{equation}
  e^{-itH} \ket{0}_a\ket{\phi}
  = \cos(t\sqrt\zeta) \ket{0}_a \ket{\phi}
  - i\sin(t\sqrt\zeta)\ket{1}_a \ket{v} .
  \label{eq:transfer-evolution}
\end{equation}
At the time
\begin{equation}
  t^\ast = \frac{\pi}{2\sqrt\zeta}
  \label{eq:transfer-time}
\end{equation}
the rotation is at its half-period, and the transfer is exact,
$e^{-it^\ast H}\ket{0}_a  \ket{\phi}= -i \ket{1}_a \ket{v}$. Applying $iX_a$ to
the ancilla removes the phase and returns it to $\ket{0}_a$, so the transfer
unitary
\begin{equation} 
  W := (iX_a \otimes\Id)\,e^{-it^\ast H}
  \label{eq:transfer-W}
\end{equation}
acts as $ W  \ket{0}_a\ket{\phi}  = \ket{0}_a\ket{v} $, as claimed.
\end{proof}

Observe that under the action of the transfer unitary \eqref{eq:transfer-W}, the ancilla begins and ends in $\ket{0}_a$, and the system is left in the normalized image of $\ket{\phi}$ under $A$. The only
input this mechanism needs is the eigenvalue $\zeta$, which sets the
evolution time \eqref{eq:transfer-time}.

\subsection{The ladder property}
\label{sec:ladder}

The transfer mechanism needs the input state to be an eigenvector of $A^\dagger A$.
With $A = J_s$, this requirement is supplied by the following structure of the current
operators.

Fix a cycle length $s$, and suppose the currents of every length greater than
$s$ have already been deposited, producing a normalized state $\ket{\Phi^{>s}} \propto \prod_{i:\,\nu_i > s} J_{\nu_i} \seed$.
Let us denote the states obtained by depositing $j$ further copies of $J_s$ as
\begin{equation}
  \ket{\Phi_{s,j}} := \frac{J_s^{\,j} \ket{\Phi^{>s}}}{\sqrt{s^{\,j}\, j!}}\,,
  \qquad j = 0, 1, \ldots, m_s\,.
  \label{eq:ladder-states}
\end{equation}
In Appendix~\ref{app:eigenvalue-proof}, we prove that the above state is normalized, as well as the following proposition:

\begin{proposition}[Ladder property]
\label{prop:ladder}
For every $1\leq j\leq m_s$, the state $\ket{\Phi_{s,j-1}}$ is an eigenvector of
$J_s^\dagger J_s$ with eigenvalue $sj$,
\begin{equation}
  J_s^\dagger J_s \ket{\Phi_{s,j-1}} = sj\,\ket{\Phi_{s,j-1}}\,.
  \label{eq:ladder-eigen}
\end{equation}
Equivalently, $J_s$ and $J_s^\dagger$ act as raising and lowering operators on
these states,
\begin{equation}
  J_s\ket{\Phi_{s,j-1}} = \sqrt{sj}\ket{\Phi_{s,j}}\,,
  \qquad
  J_s^\dagger\ket{\Phi_{s,j}} = \sqrt{sj}\ket{\Phi_{s,j-1}}\,.
  \label{eq:ladder-raiselower}
\end{equation}
\end{proposition}

The eigenvalue $sj$ is exactly the input that the transfer mechanism needs. Write
$t_{s,j}$ for the evolution time of the $j$-th transfer at cycle length $s$,
that is, the deposit of the $j$-th copy of $J_s$. At this transfer, the singular
value is $\sqrt{sj}$, so by \eqref{eq:transfer-time}
\begin{equation}
  t_{s,j} \;=\; \frac{\pi}{2\sqrt{sj}}\,.
  \label{eq:ladder-time}
\end{equation}

To prove Proposition~\ref{prop:ladder}, a key quantity is the following commutator of the current operator:
\begin{equation}
  [J_s^\dagger, J_s]
  = \sum_{j=1}^{s}\hat n_j  - \sum_{j=2n-s+1}^{2n}\hat n_j\,.
  \label{eq:comm-boundary}
\end{equation}
This equals $s$ on any state whose lowest $s$ modes are occupied and whose
highest $s$ modes are empty. Every state $\ket{\Phi_{s,j-1}}$ that the algorithm
encounters has this form. Immediately before a length-$s$ transfer, the total
length of hops from the already-deposited currents is at most $n-s$, since the present transfer still
contributes $s$ and all cycle lengths must sum to $n$. Appendix~\ref{app:eigenvalue-proof}
shows that this bound keeps the lowest $s$ modes occupied and the highest $s$
modes empty. Given this, $J_s^\dagger$ lowers each $\ket{\Phi_{s,j}}$ and the eigenvalue
\eqref{eq:ladder-eigen} follows by induction on $j$.

\subsection{The Hamiltonian-based CSP algorithm}
\label{sec:csp-algorithm}

We are now ready to prove Theorem~\ref{thm:intro-csp}, since
the transfer mechanism and the ladder property combine directly into an algorithm for character state preparation.

As visualized in Figure~\ref{fig:intro-circuits}(a), our Hamiltonian-based CSP algorithm uses $2n$ system qubits and a single ancilla, initialized in $\ket{0}_a\otimes\seed$.
For any given conjugacy class
$\nu = (n^{m_n}\cdots 2^{m_2}1^{m_1})$, we process the cycle lengths in
decreasing order $s = n, n-1, \ldots,2, 1$. For each cycle length $s$ with multiplicity $m_s\ge 1$, we apply
the $m_s$ transfer unitaries
\begin{equation}
  W_{s,1},\; W_{s,2},\; \ldots,\; W_{s,m_s}\,,
 \quad \text{where}
  \qquad
  W_{s,j} := (iX_a \otimes\Id)\, e^{-it_{s,j} H_s}\,, 
  \label{eq:csp-transfers}
\end{equation}
with $H_s$ from \eqref{eq:Hs} and $t_{s,j}=\pi/(2\sqrt{sj})$ as in \eqref{eq:ladder-time}. 
Theorem~\ref{thm:intro-csp} is the claim that this procedure produces the desired state, as
\begin{align} 
\prod_{s=n}^1 (W_{s,m_s} \cdots W_{s,2} W_{s,1}) \ket{0}_a\seed
&= (W_{1,m_1} \cdots  W_{1,1}) \cdots (W_{n,m_n} \cdots  W_{n,1})\ket{0}_a\seed  \nonumber \\
&= \ket{0}_a\ket{\chi_\nu}\,,
\label{eq:transfer-result}
\end{align}
where the products with $m_s=0$ are treated as the identity.

To see Eq.~\eqref{eq:transfer-result}, we note that before $W_{s,j}$ is applied,
the system is in $\ket{\Phi_{s,j-1}}$, which by Proposition~\ref{prop:ladder}
is an eigenvector of $J_s^\dagger J_s$ with eigenvalue $sj$. 
By the transfer mechanism in Lemma~\ref{lem:transfer}, taken with $A = J_s$ and $\zeta = sj$, we have $W_{s,j} \big (\ket{0}_a\ket{\Phi_{s,j-1}} \big) = \ket{0}_a\ket{\Phi_{s,j}}$. Applying the $m_s$ transfers in turn advances the
system from $\ket{\Phi^{>s}}$ to $\ket{\Phi_{s,m_s}}$, depositing $m_s$ copies
of $J_s$, and moving to the next value of $s$ deposits the currents of the next
length. After all $\ell=\sum_sm_s$ transfers the system holds
$J_\nu\seed/\|J_\nu\seed\|=\ket{\chi_\nu}$, with the ancilla returned to
$\ket{0}_a$.

The normalization is supplied by the algorithm itself. Each transfer $W_{s,j}$
divides by the singular value $\sqrt{sj}$, so the full sequence divides
$J_\nu\seed$ by
\begin{equation}
  \prod_{s=1}^{n}\prod_{j=1}^{m_s}\sqrt{sj}
  =
  \prod_{s=1}^{n}\sqrt{s^{m_s}\,m_s!}
  =
  \sqrt{z_\nu}
  =
  \|J_\nu\seed\|\,,
  \label{eq:norm-telescope}
\end{equation}
which is the normalization in \eqref{eq:targetfromseed}.

\section{Hamiltonian simulation for $H_s$}
\label{sec:ham-sim}

The CSP algorithm requires implementing the time evolution $e^{-it_{s,j}H_s}$
for each transfer, with $H_s$ as in \eqref{eq:Hs}. This section gives the
two implementations in Theorems~\ref{thm:intro-trotter} and \ref{thm:intro-qsvt} of Section~\ref{sec:intro-main-results}. We
first describe a Trotter decomposition suitable for reconfigurable-qubit hardware such as neutral atom arrays, where reordering the fermions makes every term of $H_s$ nearest-neighbor (Section~\ref{sec:chain-decomp}), 
bound its error in operator norm (Section~\ref{sec:trotter-error}) and on the states our algorithm visits (Section~\ref{sec:state-dependent}), and benchmark with numerical simulations (Section~\ref{sec:numerics}). 
Then, we give a QSVT-based implementation
(Section~\ref{sec:qsvt}), which gives a complete bound with
logarithmic error dependence.

\subsection{Implementation via Trotter decomposition}

\subsubsection{Nearest-neighbor form of $H_s$ using reconfigurable qubits}
\label{sec:chain-decomp}

\paragraph{Chain decomposition.}
Each term $a_{i+s}^\dagger a_i$ in $J_s$ connects sites that whose indices are the same modulo $s$. The operator, therefore, decomposes into $s$ independent
chains (Figure~\ref{fig:chain-decomp}),
\begin{align}\label{eq:chain-decomp}
    J_s = \sum_{\mu=1}^{s} J_s^{(\mu)}\,, \qquad
    J_s^{(\mu)} := \sum_{r=0}^{L_\mu - 2} a_{\mu+(r+1)s}^\dagger\, a_{\mu+rs}\,,
\end{align}
where $L_\mu = \lfloor(2n-\mu)/s\rfloor + 1 \approx 2n/s$ is the length of chain
$\mu$. The chains act on disjoint mode sets, so $[J_s^{(\mu)},
(J_s^{(\mu')})^\dagger] = 0$ for $\mu \ne \mu'$.

In the physical mode ordering, each bond $a_{\mu+(r+1)s}^\dagger a_{\mu+rs}$
carries a Jordan--Wigner string of length $s-1$ over the intervening modes,
which belong to other chains. The fermionic permutation $P_s$ introduced later in \eqref{eq:Ps-perm} reorders the modes so that each chain is
contiguous, and in that basis the string cancels and $J_s^{(\mu)}$ becomes a
nearest-neighbor hopping operator. Each term of $H_s$ then acts on two
adjacent site qubits and the ancilla, see \eqref{eq:bond-term-main} below.

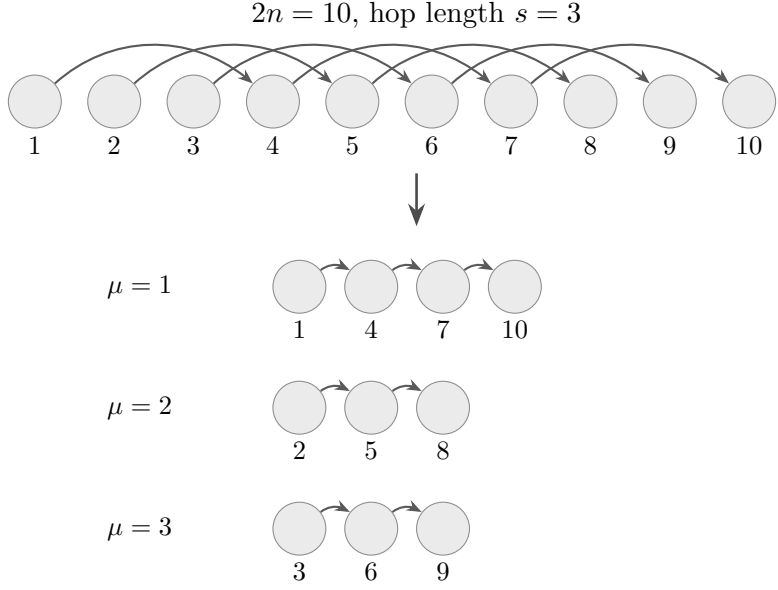
\begin{figure}[t]
\centering
\begin{tikzpicture}[
  site/.style={circle, draw=black!45, fill=black!8, minimum size=7mm, inner sep=0pt},
  hop/.style={-{Stealth[length=2mm]}, black!65, line width=0.8pt},
  lbl/.style={font=\small},
  lbl/.style={font=\small}
]

\foreach \i in {1,...,10} {
  \node[site] (t\i) at (\i*1.05, 3) {};
  \node[lbl] at (\i*1.05, 2.42) {$\i$};
}
\foreach \i in {1,...,7} {
  \pgfmathtruncatemacro{\j}{\i+3}
  \draw[hop] (t\i) to[bend left=42] (t\j);
}
\node at (6.1, 4.15) {$2n=10$, hop length $s=3$};

\draw[-{Stealth[length=3mm]}, black!70, line width=1pt] (6.1,2.05) -- (6.1,1.35);

\node[lbl, anchor=east] at (3.0,0.55) {$\mu=1$};
\foreach \k/\lab in {1/1, 2/4, 3/7, 4/10} {
  \node[site] (a\k) at (3.6 + \k*0.95, 0.55) {};
  \node[lbl] at (3.6 + \k*0.95, -0.02) {$\lab$};
}
\foreach \k in {1,2,3} {
  \pgfmathtruncatemacro{\kk}{\k+1}
  \draw[hop] (a\k) to[bend left=38] (a\kk);
}

\node[lbl, anchor=east] at (3.0,-1.05) {$\mu=2$};
\foreach \k/\lab in {1/2, 2/5, 3/8} {
  \node[site] (b\k) at (3.6 + \k*0.95, -1.05) {};
  \node[lbl] at (3.6 + \k*0.95, -1.62) {$\lab$};
}
\foreach \k in {1,2} {
  \pgfmathtruncatemacro{\kk}{\k+1}
  \draw[hop] (b\k) to[bend left=38] (b\kk);
}

\node[lbl, anchor=east] at (3.0,-2.65) {$\mu=3$};
\foreach \k/\lab in {1/3, 2/6, 3/9} {
  \node[site] (c\k) at (3.6 + \k*0.95, -2.65) {};
  \node[lbl] at (3.6 + \k*0.95, -3.22) {$\lab$};
}
\foreach \k in {1,2} {
  \pgfmathtruncatemacro{\kk}{\k+1}
  \draw[hop] (c\k) to[bend left=38] (c\kk);
}

\end{tikzpicture}
\caption{Chain decomposition \eqref{eq:chain-decomp} of the current operator
$J_s$, shown for $2n=10$ and $s=3$. On the full chain $J_s$ hops each mode to
the one $s$ sites to its right. Grouping modes by residue class modulo $s$
splits $J_s$ into $s$ independent chains $J_s^{(\mu)}$, here the residue classes
$\{1,4,7,10\}$, $\{2,5,8\}$ and $\{3,6,9\}$, each a nearest-neighbor hopping
chain of length $L_\mu \approx 2n/s$. The figure shows the decomposition at the
level of mode couplings. In the physical ordering, each bond still carries a
Jordan--Wigner string over the intervening modes, which the permutation $P_s$ in \eqref{eq:Ps-perm} removes.}
\label{fig:chain-decomp}
\end{figure}

\paragraph{Fermionic permutation.}
To make the chains of \eqref{eq:chain-decomp} local, we reorder the modes by
residue class modulo $s$. Let $P_s$ be the fermionic permutation that groups
them,
\begin{align}
\label{eq:Ps-perm}
    P_s: \quad (1,\, 1{+}s,\, 1{+}2s,\, \ldots;\;\;
               2,\, 2{+}s,\, \ldots;\;\; \ldots;\;\;
               s,\, 2s,\, 3s,\, \ldots)\,,
\end{align}
so that after $P_s$ each chain occupies a contiguous block of modes. Writing
$H_s' = (\Id_a \otimes P_s)\, H_s\, (\Id_a \otimes P_s^\dagger)$,
\begin{align}\label{eq:permute-evolve}
    e^{-itH_s} = (\Id_a \otimes P_s^\dagger)\, e^{-itH_s'}\,
    (\Id_a \otimes P_s)\,.
\end{align}
Under the Jordan-Wigner encoding, $P_s$ is not a product of local gates, since
reordering the modes rearranges the Jordan-Wigner strings, and on a fixed
one-dimensional layout it costs $\Theta(n^2)$ fermionic swaps. We instead use
the reconfigurable-qubit construction of Maskara
\emph{et al.}~\cite{Maskara_2025} described in Section~\ref{sec:intro-prior}.
For a general fermionic permutation of $2n$ modes, their construction uses
$O(n\log n)$ Clifford gates in depth $O(\log n)$. The permutation $P_s$ is
cheaper because it has more structure. Lemma~\ref{lem:reflection} below shows
that grouping the $2n$ modes by residue modulo $s$ is the \emph{2D
reflection} of~\cite{Maskara_2025} on a $(2n/s)\times s$ grid, or such a
reflection, with at most $s$ modes deleted when $s\nmid2n$. Their Corollary~2
and Theorem~3 compile this reflection into $O(n)$ CNOT gates in depth $O(1)$,
using $O(n)$ ancillas, mid-circuit measurement, and classical feedforward.

\begin{lemma}\label{lem:reflection}
Let $s \le 2n$ and write $2n = Ls + t$ with $0 \le t < s$. If $t = 0$, the
permutation $P_s$ is the 2D reflection $m_0(r,c) = rL_c + c \mapsto m_1(r,c) =
cL_r + r$ on the $L_r\times L_c$ grid with $(L_r, L_c) = (L,\, s)$. If $t > 0$,
$P_s$ is this reflection on the $(L{+}1)\times s$ grid with $s - t$ modes
deleted.
\end{lemma}
\begin{proof}
Index the modes $0,\ldots,2n-1$ and write each as $i = \kappa s + \rho$ with
$\rho = i \bmod s$ and $\kappa = \lfloor i/s\rfloor$. The permutation $P_s$
orders modes by residue $\rho$, and within a residue class by $\kappa$.
Suppose first $t = 0$, so $L = 2n/s$ and each residue class has exactly $L$
modes. The modes with residue less than $\rho$ fill the first $\rho L$ output
positions, and mode $\kappa s + \rho$ is the $\kappa$-th within its class, so
\begin{align}
    P_s(i) = \rho L + \kappa\,.
    \label{eq:Ps-coords}
\end{align}
Setting $r = \kappa$, $c = \rho$, $L_r = L$, $L_c = s$, we have $m_0(r,c) =
\kappa s + \rho = i$, and the reflection sends this to $m_1(r,c) = \rho L +
\kappa$, which matches \eqref{eq:Ps-coords}. Hence, $P_s$ is the stated
reflection.
If $t > 0$, residue class $\rho$ has $\lfloor (2n-1-\rho)/s\rfloor + 1$ modes,
which equals $L+1$ for $\rho < t$ and $L$ for $\rho \ge t$. Thus, $P_s$ is the
reflection on the full $(L{+}1)\times s$ grid with the $s - t$ absent modes,
one from each residue class $\rho \ge t$, deleted.
\end{proof}

For a uniform cycle type $\nu=(s^m)$ every transfer uses the same
permutation $P_s$, which commutes with the ancilla correction $iX_a$, so the
intermediate permutations cancel and only one permutation at the start and
one at the end remain,
\begin{align}\label{eq:U-equal-cycle}
    W_{s,m}\cdots W_{s,1}
    =
    (\Id_a \otimes P_s^\dagger)\,
    \big[(iX_a \otimes\Id)\,e^{-it_{s,m}H_s'}\big]\cdots
    \big[(iX_a \otimes\Id)\,e^{-it_{s,1}H_s'}\big]\,
    (\Id_a\otimes P_s)\,.
\end{align}
For a general cycle type, one permutation $P_{s'}P_s^\dagger$ is needed at
each change of cycle length, so there are at most $\ell(\nu)+1$ permutations in total.

\paragraph{Bond terms.}
In the permuted basis, $H_s'$ is a sum of nearest-neighbor bond terms, each
coupling two adjacent site qubits $(q, q+1)$ to the ancilla $a$,
\begin{align}\label{eq:bond-term-main}
    h_{q,a} = \sigma_a^+ \otimes \sigma_{q+1}^+ \sigma_q^- 
            + \sigma_a^- \otimes \sigma_{q+1}^- \sigma_q^+ \,,
\end{align}
with $\sigma^+ = |1\rangle\langle 0|$ and $\sigma^- = |0\rangle\langle 1|$. In
the ordering (ancilla $a$, site $q$, site $q{+}1$), each $h_{q,a}$ acts as $X$
on the two-dimensional subspace $\mathrm{span}\{|010\rangle, |101\rangle\}$ and
as zero on its complement. The ancilla is shared by every bond. On a one-dimensional layout, it is moved
next to the bond being rotated at a cost of $O(1)$ swaps per bond, so that
each bond rotation acts on three adjacent qubits, as drawn in
Figure~\ref{fig:intro-circuits}(b).

The bond rotation $e^{-i\theta h_{q,a}}$ is implemented by conjugating a single
controlled rotation. Define
\begin{align}\label{eq:bond-compress}
    T_{q,a} := \mathrm{CNOT}_{q\to a}\,\mathrm{CNOT}_{q\to(q+1)}\, X_q ,
\end{align}
which maps $|010\rangle \mapsto |000\rangle$ and $|101\rangle \mapsto
|010\rangle$, collapsing the active subspace onto
$\mathrm{span}\{|000\rangle,|010\rangle\}$. Let $V(\theta)$ rotate on this pair,
\begin{align}\label{eq:bond-rotation}
    V(\theta) :\quad
    \begin{aligned}
        |000\rangle &\mapsto \cos\theta\,|000\rangle - i\sin\theta\,|010\rangle\,,\\
        |010\rangle &\mapsto \cos\theta\,|010\rangle - i\sin\theta\,|000\rangle\,,
    \end{aligned}
\end{align}
while acting as identity on the rest of the states.
Then
\begin{align}\label{eq:bond-decomp}
    e^{-i\theta h_{q,a}} = T_{q,a}^\dagger\, V(\theta)\, T_{q,a}\,,
\end{align}
as verified by tracking $|010\rangle$ and $|101\rangle$ through the
conjugation (Figure~\ref{fig:bond-circuit}). Each bond rotation is a constant
number of one- and two-qubit gates.

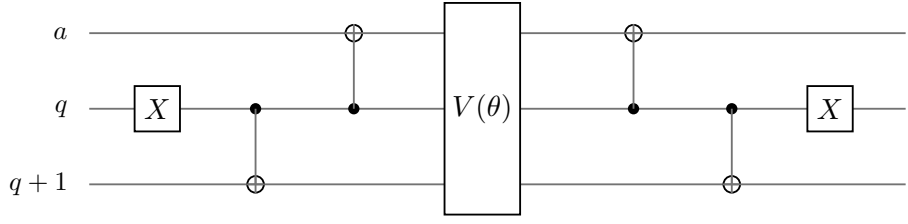
\begin{figure}[t]
\centering
\begin{tikzpicture}[
  wire/.style={draw, black!55, line width=0.7pt},
  gate/.style={draw, fill=white, minimum size=6mm, inner sep=2pt, line width=0.7pt},
  ctrl/.style={circle, fill=black, inner sep=1.5pt},
  targ/.style={circle, draw, line width=0.7pt, inner sep=2.2pt},
  lbl/.style={font=\small}
]

\def\ya{2.0}
\def\yq{1.0}
\def\yqq{0.0}
\def\xL{0.4}
\def\xR{11.2}

\foreach \yy in {\yq,\yqq,\ya} { \draw[wire] (\xL,\yy) -- (\xR,\yy); }
\node[lbl, anchor=east] at (\xL-0.15,\ya)  {$a$};
\node[lbl, anchor=east] at (\xL-0.15,\yq)  {$q$};
\node[lbl, anchor=east] at (\xL-0.15,\yqq) {$q+1$};

\def\xa{1.3}
\def\xb{2.6}
\def\xc{3.9}
\def\xd{5.6}
\def\xe{7.6}
\def\xf{8.9}
\def\xg{10.2}

\node[gate] at (\xa,\yq) {$X$};

\node[ctrl] at (\xb,\yq) {};
\node[targ] (t1) at (\xb,\yqq) {};
\draw[wire] (\xb,\yq) -- (\xb,\yqq);
\draw[wire] (t1.west) -- (t1.east); \draw[wire] (t1.south) -- (t1.north);

\node[ctrl] at (\xc,\yq) {};
\node[targ] (t2) at (\xc,\ya) {};
\draw[wire] (\xc,\yq) -- (\xc,\ya);
\draw[wire] (t2.west) -- (t2.east); \draw[wire] (t2.south) -- (t2.north);

\node[gate, minimum width=10mm, minimum height=28mm] at (\xd,\yq) {$V(\theta)$};

\node[ctrl] at (\xe,\yq) {};
\node[targ] (t3) at (\xe,\ya) {};
\draw[wire] (\xe,\yq) -- (\xe,\ya);
\draw[wire] (t3.west) -- (t3.east); \draw[wire] (t3.south) -- (t3.north);

\node[ctrl] at (\xf,\yq) {};
\node[targ] (t4) at (\xf,\yqq) {};
\draw[wire] (\xf,\yq) -- (\xf,\yqq);
\draw[wire] (t4.west) -- (t4.east); \draw[wire] (t4.south) -- (t4.north);

\node[gate] at (\xg,\yq) {$X$};

\end{tikzpicture}
\caption{Circuit for a single bond rotation
$e^{-i\theta h_{q,a}} = T_{q,a}^\dagger\, V(\theta)\, T_{q,a}$ on the ancilla $a$ and two site qubits $q, q{+}1$.
The compression $T_{q,a}$ of \eqref{eq:bond-compress} maps the active subspace
onto $\mathrm{span}\{|000\rangle,|010\rangle\}$, where $V(\theta)$ performs the
rotation \eqref{eq:bond-rotation}, after which $T_{q,a}^\dagger$ uncomputes.}
\label{fig:bond-circuit}
\end{figure}

\paragraph{Trotter decomposition.}
All bond terms share the ancilla, so two bonds that do not share a site qubit do not commute (Appendix~\ref{app:commutators}). We, therefore, use the first-order
Lie-Trotter formula. Dividing the evolution time $t_{s,j}$ into $M$ steps of
size $\tau = t_{s,j}/M$,
\begin{align}\label{eq:trotter-main}
    e^{-it_{s,j} H_s'} \approx
    \left(\prod_{(q,q+1)} e^{-i\tau\, h_{q,a}}\right)^{\!M}\,,
\end{align}
with the approximation error analyzed in Sections~\ref{sec:trotter-error} and~\ref{sec:state-dependent}. Each
step applies one rotation per bond, and the $s$ chains have
$\sum_\mu (L_\mu - 1) = 2n - s$ bonds in total, so a Trotter step costs $O(n)$
gates.

\subsubsection{Trotter error analysis}\label{sec:trotter-error}

We first bound the error of the first-order Trotter decomposition
\eqref{eq:trotter-main} in operator norm. For a single transfer at length $s$,
the evolution $e^{-it_{s,j}H_s'}$ is approximated by $M$ applications of the
bond product $S(\tau) = \prod_{(q,q+1)} e^{-i\tau h_{q,a}}$ with $\tau =
t_{s,j}/M$. The first-order Trotter bound~\cite{ChildsSuTranWiebeZhu2021} gives
\begin{align}\label{eq:trotter-opnorm}
    \big\| S(\tau)^M - e^{-it_{s,j}H_s'} \big\|
    \;\le\; \frac{t_{s,j}^2}{2M}
    \sum_{w < w'} \big\| [h_w, h_{w'}] \big\|
    \;=\; \frac{t_{s,j}^2}{2M}\, C(s,n)\,,
\end{align}
where $C(s,n) := \sum_{w < w'}\|[h_w, h_{w'}]\|$ is the sum of bond
commutator norms.

Two bond terms commute if and only if they lie in the same chain and share a
site qubit. Every other pair, including bonds in different chains, fails to
commute through the shared ancilla, with commutator norm $1$
(Appendix~\ref{app:commutators}). The shared ancilla, therefore, couples almost
all of the $\binom{2n-s}{2}$ bond pairs, and subtracting the $2n-2s$ same-chain
adjacent pairs that do commute gives
\begin{align}\label{eq:C-exact}
    C(s,n) = \binom{2n-s}{2} - (2n - 2s) = \Theta(n^2)\,.
\end{align}
This value depends only on $s$ and $n$, not on the rest of the cycle type.
Combining \eqref{eq:trotter-opnorm} and \eqref{eq:C-exact}, the $j$-th transfer
at length $s$ is simulated to error $\varepsilon_{s,j}$ using
\begin{align}\label{eq:M-per-transfer}
    M_{s,j} = \frac{t_{s,j}^2\, C(s,n)}{2\varepsilon_{s,j}}
            = \Theta\!\left(\frac{n^2}{s\, j\, \varepsilon_{s,j}}\right)
\end{align}
Trotter steps, each costing $O(n)$ gates.

Summing over all transfers of the CSP algorithm, we split the total error
budget $\delta$ uniformly, assigning $\varepsilon_{s,j} = \delta/\ell$ to each
of the $\ell = \sum_s m_s$ transfers. Using $t_{s,j}^2 = \pi^2/(4sj)$, the total
Trotter step count is
\begin{align}\label{eq:M-total}
    M_{\mathrm{tot}}
    = \sum_{s}\sum_{j=1}^{m_s} M_{s,j}
    = \frac{\ell}{2\delta}\sum_{s} C(s,n)\sum_{j=1}^{m_s} t_{s,j}^2
    = \frac{\pi^2\, \ell}{8\delta}
      \sum_{s} \frac{C(s,n)\, H_{m_s}}{s}\,,
\end{align}
where $H_m = \sum_{j=1}^{m} 1/j$ is the $m$-th harmonic number. Since $C(s,n) =
\Theta(n^2)$ uniformly in $s$, and $\sum_s H_{m_s}/s \le \big(\sum_s
1/s\big)\max_s H_{m_s} = O(\log^2 n)$, with $\ell \le n$,
\begin{align}\label{eq:M-total-scaling}
    M_{\mathrm{tot}} = O\!\left(\frac{n^3 \log^2 n}{\delta}\right)\,.
\end{align}
Each Trotter step costs $O(n)$ gates, so the operator-norm bound gives a total
gate count of $O(n^4 \log^2 n/\delta)$ for any cycle type.

This bound is loose for two reasons. First, $C(s,n)$ is an operator norm over
the whole Hilbert space, while the error actually incurred is the norm of the
same commutator sum acting on the particular state the algorithm passes through,
which lies in the $n$-particle sector. Second, $C(s,n)$ bounds each commutator
separately and discards the cancellation between them. On the algorithm's states,
the non-commuting terms partially interfere, so the true error falls well below
the sum of the individual norms. Section~\ref{sec:state-dependent} exploits the first point for uniform cycle
types, where the states the algorithm visits can be characterized exactly.
Section~\ref{sec:qsvt} avoids the Trotter error altogether.

\subsubsection{State-dependent Trotter error for uniform cycles}
\label{sec:state-dependent}

Consider a uniform cycle type $\nu=(s^m)$ with $n=ms$. The operator-norm
bound of Section~\ref{sec:trotter-error} maximizes the Trotter error over
every input state, while the algorithm only passes through the ladder states
$\ket{\Phi_{s,r}}$ in Eq.~\eqref{eq:ladder-states}. On these states, the
commutator sum is much smaller, and for uniform cycles it can be evaluated
exactly.

Write $\ket{\rho_r}=\ket{\Phi_{s,r}}$ for the normalized state after $r$
transfers, and define $\mathcal A$ by writing the first-order Trotter
commutator sum as $\sum_{w<w'}[h_w,h_{w'}]=Z_a\otimes\mathcal A$
(Appendix~\ref{app:state-dependent}). At the $k$-th transfer, let
\begin{equation}
  \Gamma_r
  :=
  \left\|
    \mathcal A\ket{\rho_r}
  \right\|\,,
  \qquad
  \Delta_k
  :=
  \max\{\Gamma_{k-1},\Gamma_k\}\,.
\end{equation}
The leading state-dependent error is
\begin{equation}
  \varepsilon_{k,\mathrm{lead}}^{\mathrm{SD}}
  \leq
  \frac{t_k^2}{2M_k}\Delta_k\,.
  \label{eq:state-error}
\end{equation}

Appendix~\ref{app:state-dependent} evaluates $\Gamma_r$ by separating
commutators between different chains from commutators within the same chain.
It proves
\begin{equation}
  \Gamma_r^2
  =
  \frac{r(s^2-1)}{3}
  +
  s\sum_{d=0}^{r}
  \binom{r}{d}
  \left(\frac{1}{s}\right)^d
  \left(1-\frac{1}{s}\right)^{r-d}
  F(d)\,,
  \label{eq:gamma-decomp}
\end{equation}
where $F(d)$ is a one-chain norm defined in Eq.~\eqref{eq:app-same}. The
appendix also proves
\begin{equation}
  F(d)\leq d^2.
\end{equation}
This gives
\begin{equation}
  \Gamma_r^2
  \leq
  \frac{r(s^2-1)}{3}
  +
  r\left(1-\frac{1}{s}\right)
  +
  \frac{r^2}{s}\,,
\end{equation}
and therefore
\begin{equation}
  \Gamma_r
  =
  O\!\left(
    s\sqrt r+\frac{r}{\sqrt s}
  \right)\,.
  \label{eq:gamma-scaling}
\end{equation}

Using $t_k=\pi/(2\sqrt{ks})$, Eq.~\eqref{eq:state-error} becomes
\begin{equation}
  \varepsilon_{k,\mathrm{lead}}^{\mathrm{SD}}
  =
  O\!\left[
    \frac{1}{M_k}
    \left(
      \frac{1}{\sqrt k}
      +
      \frac{1}{s^{3/2}}
    \right)
  \right]\,.
  \label{eq:M-state}
\end{equation}
The first term comes from commutators between different chains. The second
comes from commutators within the same chain.

Assigning leading error $\delta/m$ to each transfer and summing over all
$m$ transfers gives
\begin{equation}\label{eq:Mtot-SD}
  M_{\mathrm{tot}}
  =
  O\!\left[
    \frac{1}{\delta}
    \left(
      m^{3/2}
      +
      \frac{m^2}{s^{3/2}}
    \right)
  \right]
\end{equation}
Trotter steps. Each step contains $s(2m-1)=2n-s=O(n)$ bond evolutions. The
leading-order gate count is therefore
\begin{equation}
  G_{\mathrm{lead}}^{\mathrm{SD}}(s^m)
  =
  O\!\left[
    \frac{n}{\delta}
    \left(
      m^{3/2}
      +
      \frac{m^2}{s^{3/2}}
    \right)
  \right]\,.
  \label{eq:gate-state}
\end{equation}

This bound gives two regimes, with one conditional improvement:
\begin{itemize}
  \item If $m=O(s^3)$, the same-chain term does not change the scaling, and
  \begin{equation}
    G_{\mathrm{lead}}^{\mathrm{SD}}(s^m)
    =
    O\!\left(
      \frac{n m^{3/2}}{\delta}
    \right)\,.
  \end{equation}
  It is strictly lower order when $m=o(s^3)$ and $s$ grows.

  \item If $m\gg s^3$, the same-chain term cannot be dropped using the
  proved bound $F(d)\leq d^2$. The resulting bound is
  \begin{equation}
    G_{\mathrm{lead}}^{\mathrm{SD}}(s^m)
    =
    O\!\left(
      \frac{n m^2}{\delta s^{3/2}}
    \right)\,.
  \end{equation}

  \item If $F(d)=O(d)$, which we have not proved, the same-chain
  contribution is $O(r)$ rather than $O(r+r^2/s)$, and
  $G_{\mathrm{lead}}^{\mathrm{SD}}(s^m)=O(nm^{3/2}/\delta)$ for every
  uniform cycle type.
\end{itemize}

At the cycle count $m=\Theta(\sqrt n/\log n)$ of the granularity regime
\eqref{eq:intro-bghs-regime}, the cycle length is
$s=n/m=\Theta(\sqrt n\log n)$, so $m=O(s^3)$ and the first regime applies,
\begin{equation}
  G_{\mathrm{lead}}^{\mathrm{SD}}
  =
  O\!\left(\frac{n\,m^{3/2}}{\delta}\right)
  =
  O\!\left(\frac{n^{7/4}}{\delta\,\log^{3/2}n}\right)\,,
  \label{eq:gate-state-sqrt}
\end{equation}
which is the estimate claimed in Theorem~\ref{thm:intro-trotter}. 

We remark that these results control only the leading $1/M_k$ contribution to the state-dependent Trotter error. The higher-order remainder has not yet been
bounded.

For fixed $s$ and $\delta$, with
$m=n/s$, the bound on the total number of Trotter steps in \eqref{eq:Mtot-SD} (which is $G/n$) grows as $n^2$, and can be improved to $n^{1.5}$ assuming $F(d)=O(d)$.
In contrast,  numerical simulations up to $n=30$ suggest that $n^{1.2}$ Trotter steps are sufficient (see Figure~\ref{fig:family-numerics} of Section~\ref{sec:numerics}), indicating that even these state-dependent bounds might be conservative overestimates.

\subsubsection{Numerical simulations}
\label{sec:numerics}

We simulate the complete Trotterized CSP circuit exactly,
and compare the number of Trotter steps it needs with the bounds of
Section~\ref{sec:trotter-error} and \ref{sec:state-dependent}.

\paragraph{Setup.}
For a cycle type $\nu\vdash n$, we apply the transfers in the order of
Section~\ref{sec:csp-algorithm} and Trotterize every transfer with the same
number $M$ of first-order steps \eqref{eq:trotter-main}. We find the smallest
$M_\ast$ for which the final state $\ket{\psi_M}$ satisfies
\begin{equation}\label{eq:numerics-threshold}
  \min_\varphi\|\ket{0_a}\ket{\chi_\nu}-e^{i\varphi}\ket{\psi_M}\| = \sqrt{2-2\,\left|(\langle0_a|\bra{\chi_\nu})\ket{\psi_{M}}\right|}\le10^{-2}\,,
\end{equation}
and plot the total number of Trotter steps
$M_{\mathrm{tot}}=\ell(\nu)\,M_\ast.$ 

To scale up our numerical simulations to larger $n$, we take advantage of the fact that the Trotterized dynamics does not use the full $\binom{2n}{n}$-dimensional Fock space. Instead, we work in the eigenspaces of the displacement operator $\widehat{\mathsf D}$ of
Appendix~\ref{app:eigenvalue-proof}. 

Let $\mathcal F_k$ be the eigenspace
of $\widehat{\mathsf D}$ with eigenvalue $k$. Its basis states
$\ket{p_1<\cdots<p_n}$ with $\sum_r(p_r-r)=k$ are labeled by the partitions
$(p_n-n,\dots,p_1-1)$ of $k$, so $\dim\mathcal F_k=p(k)$ for $k\le n$. By
Eq.~\eqref{eq:displacement-comm}, $J_s$ maps $\mathcal F_D$ into
$\mathcal F_{D+s}$, and the same holds for each of its bond terms separately.
Hence $H_s$, and every bond rotation in the Trotter product
\eqref{eq:trotter-main}, leaves the subspace
\begin{equation}
\ket0_a \otimes \mathcal F_D\;\oplus\; \ket1_a \otimes \mathcal F_{D+s}
\end{equation}
invariant. A transfer at length $s$ that starts at displacement $D$ therefore
acts on a space of dimension $p(D)+p(D+s)$, which is at most
$2p(48)\approx3\times10^5$ in our runs.

At finite $M$, part of the amplitude remains in $\ket0_a\otimes \mathcal F_D$
at the end of the transfer, and the ancilla flip moves it to
$\ket1_a\otimes \mathcal F_D$. The remaining transfers have total length
$n-D-s$, and each one changes the displacement by at most its length, so this
component ends in $\mathcal F_k$ with $k\le n-s$.

Since $\ket{\chi_\nu}\in\mathcal F_n$ by Lemma~\ref{lem:displacement}, this
component is orthogonal to the target and does not contribute to
$\langle\chi_\nu|\psi_M\rangle$. We therefore discard it without renormalizing the retained vector, so that each transfer
begins in a single eigenspace $\ket0_a\otimes\mathcal F_D$. In this basis, each
bond rotation acts with the Jordan-Wigner sign of the corresponding hop. As a
check, we compare with a brute-force simulation in the full Fock space for
$n=4$ and $n=5$. The target states and the overlaps
$\langle\chi_\nu|\psi_M\rangle$ agree to $10^{-15}$.

\begin{figure}[t]
\centering
\includegraphics[width=\textwidth]{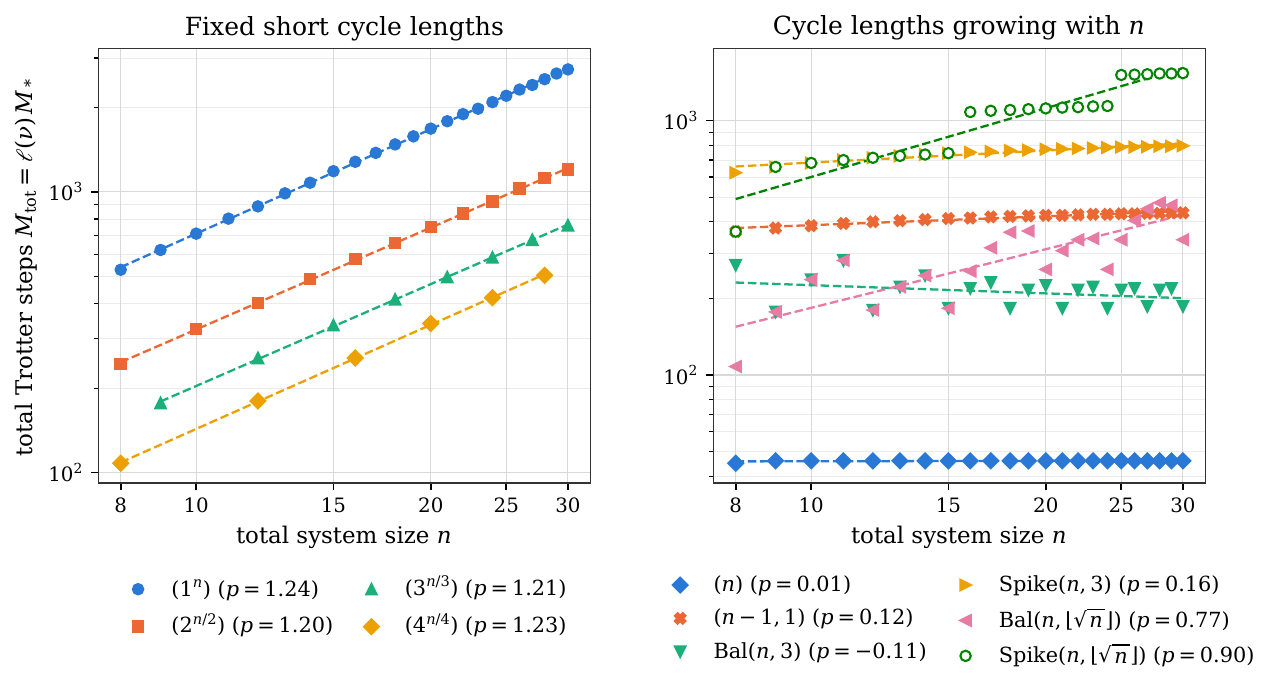}
\caption{Total Trotter steps $M_{\mathrm{tot}}=\ell(\nu)M_\ast$ needed for the
complete CSP circuit to reach final-state error $10^{-2}$,
Eq.~\eqref{eq:numerics-threshold}, for $8\le n\le30$. Markers are data, dashed
lines are least-squares fits $M_{\mathrm{tot}}\propto n^p$.
\textbf{Left}: uniform cycle types $(s^{n/s})$, $s=1,\dots,4$. The fits give
$p=1.20$--$1.24$, far below the operator-norm bound~\eqref{eq:M-total}, which
scales as $n^3\log n$, and the leading-order state-dependent
bound~\eqref{eq:Mtot-SD}, which scales as $n^2$ at fixed $s$.
\textbf{Right}: cycle types whose lengths grow with $n$. $\mathrm{Bal}(n,c)$ has $c$
cycles whose lengths differ by at most one, and
$\mathrm{Spike}(n,c)=(n-c+1,1^{c-1})$ has one long cycle and $c-1$ fixed
points. The fits give $|p|\le0.16$ at fixed cycle count and $p=0.77$ and
$0.90$ for $\lfloor\sqrt n\rfloor$ cycles. Only the operator-norm bound  \eqref{eq:M-total} applies
here: 
it scales as $n$ for $(n)$ and $\mathrm{Bal}(n,3)$, as $n^2$ for
$(n-1,1)$ and $\mathrm{Spike}(n,3)$, and as $n^2\log n$ and $n^{5/2}\log n$
for the $\lfloor\sqrt n\rfloor$ families. For $9\le n\le15$,
$\lfloor\sqrt n\rfloor=3$, so these coincide with the $c=3$ families.}
\label{fig:family-numerics}
\end{figure}

\paragraph{Uniform cycles of fixed length.}
For uniform cycle types of fixed length $s$, left panel of
Figure~\ref{fig:family-numerics}, $M_{\mathrm{tot}}$ grows as $n^{1.2}$ for
every $s$. The operator-norm bound \eqref{eq:M-total-scaling} predicts
$n^3\log^2n$, and evaluated for $(1^{30})$ at $\delta=10^{-2}$ it exceeds
the observed count by a factor of $10^4$. The leading state-dependent bound
of Section~\ref{sec:state-dependent} predicts between $n^{3/2}$ and $n^2$ at
fixed $s$, and the observed growth is below this too. The number of steps per transfer grows only as $n^{0.2}$ over the range, so most of the growth of $M_{\mathrm{tot}}$ is the growth of the number of transfers $n/s.$

\paragraph{Cycle lengths growing with $n$.}
For cycle types whose lengths grow with $n$, right panel of
Figure~\ref{fig:family-numerics}, the observed Trotter step count depends mainly on the number of cycles over the tested range.
At fixed cycle count, $M_{\mathrm{tot}}$ stays bounded over
the whole range, and for the $\lfloor\sqrt n\rfloor$-cycle families it
changes only when $\lfloor\sqrt n\rfloor$ does. At the same cycle count $c$, the spiked family $\mathrm{Spike}(n,c)$ costs three to four times more than the balanced family $\mathrm{Bal}(n,c)$, and the two grow at a similar rate. 

We remark that the families in Figure~\ref{fig:family-numerics} probe how the Trotter cost depends on the spread of the cycle lengths, not classical hardness. As mentioned in Theorem~\ref{thm:intro-classical}, we have found a classical algorithm to sample from $\ket{\chi_\nu}$ for uniform cycles, or almost-uniform cycles with a small number of exceptions. The spiked family is almost-uniform with a single exceptional cycle relative to the block $1^{c-1}$, so there is a polynomial-time sampler for every $n$ (Theorem~\ref{thm:almost-equal-classical}). The balanced family is uniform at $n=9,16,25$ (so Theorem~\ref{thm:equal-cycle-classical} applies), and almost-uniform with up to $\lfloor\sqrt n\rfloor/2$ exceptional cycles (which would make it difficult for the known classical methods).

\begin{figure}[ht]
\centering
\includegraphics[width=\textwidth]{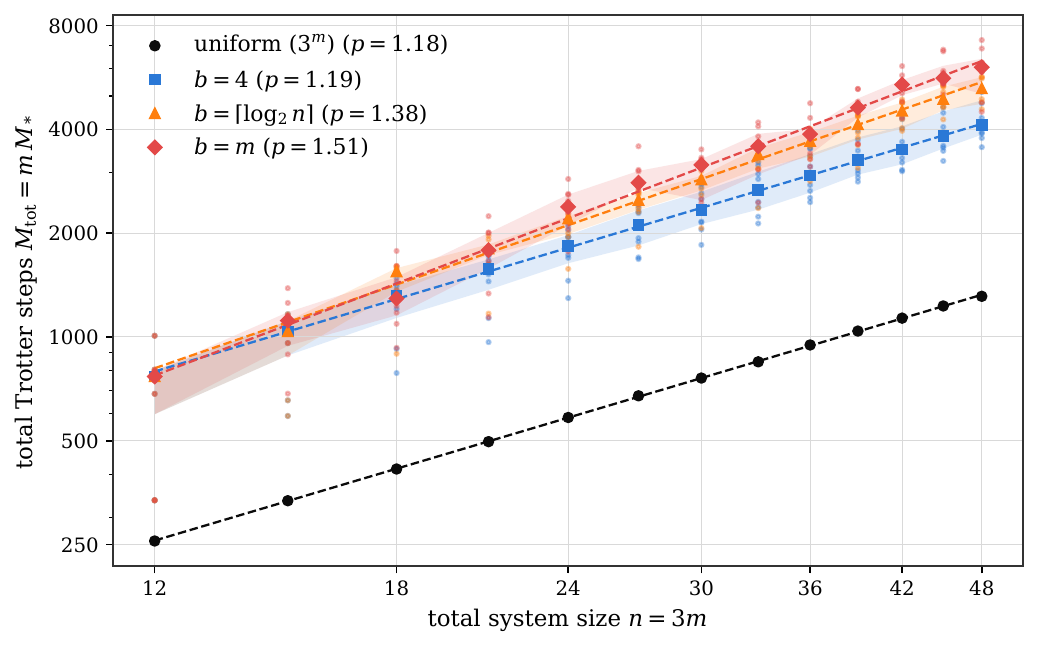}
\caption{Total Trotter steps $M_{\mathrm{tot}}=mM_\ast$ for the uniform
family $(3^m)$ and for random replacements of $b$ of its $m$ cycles at fixed
$n=3m$. A replacement draws $b$ lengths uniformly among the partitions of
$3b$ into $b$ parts with no part equal to $3$. For $b=4$ and
$b=\lceil\log_2 n\rceil$ the same eight multisets are used at every $n$,
and for $b=m$ they are drawn independently at each $n$. Markers are medians
over the eight draws, bands are the interquartile range, and faint dots are
the individual draws. Dashed lines are least-squares fits
$M_{\mathrm{tot}}\propto n^p$ of the medians over $12\le n\le 48$.  At $n=12$ the three values of $b$ coincide.
For the uniform family, the fit gives $p=1.18$, whereas the leading-order
bound~\eqref{eq:Mtot-SD} scales as $n^2$ at fixed $s$ and the operator-norm
bound~\eqref{eq:M-total} as $n^3\log n$. For the replaced ensembles, only the
operator-norm bound applies, again $n^3\log n$, while the fits give $n^{1.19}$,
$n^{1.38}$, and $n^{1.51}$.}
\label{fig:perturbation}
\end{figure}

\paragraph{Random replacements at fixed $n$ and $m$.}
The state-dependent bound of Theorem~\ref{thm:intro-trotter} is specific to
uniform cycles. To test how the cost changes when the cycle type is
perturbed, we keep $n=3m$ and the number of transfers $m$ fixed and replace
$b$ of the $m$ cycles of $(3^m)$ by a random multiset of $b$ lengths with
the same total $3b$ and no length equal to $3$. We take $b=4$,
$b=\lceil\log_2 n\rceil$, and $b=m$, and certify $M_{\rm tot}$ for eight
random draws at each $n$ up to $48$. 
Figure~\ref{fig:perturbation} shows
the result, and we make the following observations:
\begin{itemize}
\item Replacing a constant number of cycles multiplies the cost by a
constant and leaves the exponent unchanged. The fit for $b=4$ gives $1.19$
against $1.18$ for $(3^m)$.
\item The fitted exponent increases with the number of replaced cycles, to $1.38$
for $b=\lceil\log_2 n\rceil$ and $1.51$ for $b=m$, where none of the
original cycles is kept.
\item All three ensembles stay well below the exponent $3$ of the
operator-norm bound \eqref{eq:M-total} for $M_{\mathrm{tot}}$ at
$\ell=n/3$. They are also well below the exponent $2$ predicted by proxy from the state-dependent bound for the uniform cycles. Since each Trotter step costs $O(n)$ gates, the corresponding gate-count exponents for these tested families are $2.19$, $2.38$, and $2.51$.
\end{itemize}
The distinction between constant and growing $b$ matters for classical hardness. In Section~\ref{sec:classical-hardness} we show that cycle types with a bounded number of exceptional cycles are classically tractable. The number of replaced cycles $b$ upper-bounds the number of exceptional cycles in the sense of Section~\ref{sec:almost-equal-classical}, so the $b=4$ ensemble is classically polynomial-time tractable, while the ensembles with $b$ growing with $n$ are not covered by the algorithm in Theorem~\ref{thm:almost-equal-classical}.

\paragraph{Limitations.}
For $n\le48$, $\pi \sqrt{2n/3}/\log n<5$, so the granularity condition on the cycle count $\Theta(\sqrt n/\log n)$ from \eqref{eq:intro-bghs-regime} does not define a well-resolved asymptotic regime at these sizes.
The numerics do not establish scaling for growing cycle counts near this threshold.
The fits span a factor of at most $4$ in $n$,
and the exponents may change at larger $n$.
Nevertheless, within the families studied, our analytical and proven bounds overestimate the Trotter cost by orders of magnitude, and the
cost depends on the number of cycles and on how unequal their lengths are.

\subsection{Implementation via quantum singular value transformation}
\label{sec:qsvt}

The Trotter construction has three limitations: its dependence on the target
error is inverse-linear, its worst-case bound is $O(n^4\log^2n/\delta)$, and
its state-dependent improvement is a leading-order bound for uniform cycles
only. We now replace each Trotterized transfer with QSVT-based Hamiltonian
simulation, which gives a complete bound of $\widetilde O(n^{5/2})$ gates
for every cycle type with polylogarithmic dependence on the error. We first
construct an explicit block-encoding of $H_s$, then bound the cost of CSP
for a fixed cycle type, and then extend the construction coherently over all
cycle types. Section~\ref{sec:qct-cost} completes the QCT by erasing the
cycle-type registers.

We assume a gate set of arbitrary one- and two-qubit
gates. Constant-arity controlled gates are decomposed at constant cost.
Arithmetic and comparisons on $O(\log n)$ qubits are implemented reversibly.
Compilation into a fixed discrete gate set introduces additional
polylogarithmic precision factors. The gate counts below are circuit-size
bounds. We assume throughout that $0<\delta<1$.

\paragraph{Block-Hamiltonian simulation.}

An $(\omega,\mathfrak a,0)$-block-encoding of a Hermitian operator $H$ is a
unitary $U_H$ satisfying
\begin{equation}
  \left(
    \bra{0}^{\otimes\mathfrak a}\otimes I
  \right)
  U_H
  \left(
    \ket{0}^{\otimes\mathfrak a}\otimes I
  \right)
  =
  \frac{H}{\omega}\,.
  \label{eq:block-encoding-definition}
\end{equation}
The parameter $\omega$ is the normalization of the block-encoding, and
$\mathfrak a$ is the number of work qubits.

We use the block-Hamiltonian simulation result of
Ref.~\cite[Theorem~58 and Corollary~60]{Gilyen_2019}. Given an exact
$(\omega,\mathfrak a,0)$-block-encoding of $H$, QSVT constructs a unitary
$V_H(t)$ satisfying
\begin{equation}
  \left\|
    \left(
      \bra{0}^{\otimes(\mathfrak a+2)}\otimes I
    \right)
    V_H(t)
    \left(
      \ket{0}^{\otimes(\mathfrak a+2)}\otimes I
    \right)
    -
    e^{-itH}
  \right\|
  \leq
  \eta\,.
  \label{eq:qsvt-block-error}
\end{equation}
The required number of uses of $U_H$, $U_H^\dagger$, and their controlled
versions is
\begin{equation}
  Q_{\mathrm{sim}}(\omega|t|,\eta)
  =
  \Theta\!\left(
    \omega|t|
    +
    \frac{\log(1/\eta)}
    {\log\!\left(
      e+\dfrac{\log(1/\eta)}{\omega|t|}
    \right)}
  \right)\,.
  \label{eq:qsvt-optimal-query-count}
\end{equation}
The construction uses two additional block-encoding qubits and $O(1)$ further
work qubits. It also uses $O(\mathfrak a\,Q_{\mathrm{sim}}(\omega|t|,\eta))$
two-qubit gates outside the block-encoding queries.

Eq.~\eqref{eq:qsvt-block-error} controls the encoded block. We next
convert this into a bound on the complete output state, including all work
qubits. Define
\begin{equation}
  V_{00}(t)
  :=
  \left(
    \bra{0}^{\otimes(\mathfrak a+2)}\otimes I
  \right)
  V_H(t)
  \left(
    \ket{0}^{\otimes(\mathfrak a+2)}\otimes I
  \right)\,.
\end{equation}
For every normalized input state $\ket{\psi}$,
Eq.~\eqref{eq:qsvt-block-error} gives
$\|(V_{00}(t)-e^{-itH})\ket{\psi}\|\le\eta$, and since $e^{-itH}$ is unitary
the triangle inequality gives $\|V_{00}(t)\ket{\psi}\|\ge1-\eta$. The
component of $V_H(t)\ket{0}^{\otimes(\mathfrak a+2)}\ket{\psi}$ whose work
register is orthogonal to $\ket{0}^{\otimes(\mathfrak a+2)}$ has squared
norm $1-\|V_{00}(t)\ket{\psi}\|^2$. Orthogonality of the two work-register
subspaces therefore implies
\begin{align}
  &\left\|
    V_H(t)\ket{0}^{\otimes(\mathfrak a+2)}\ket{\psi}
    -
    \ket{0}^{\otimes(\mathfrak a+2)}e^{-itH}\ket{\psi}
  \right\|^2
  \nonumber\\
  &\qquad
  =
  \left\|
    \bigl(V_{00}(t)-e^{-itH}\bigr)\ket{\psi}
  \right\|^2
  +
  1-\left\|V_{00}(t)\ket{\psi}\right\|^2
  \leq
  \eta^2+1-(1-\eta)^2
  =
  2\eta\,.
  \label{eq:qsvt-block-to-state}
\end{align}
Setting $\eta=\varepsilon^2/2$ makes the simulated evolution
$\varepsilon$-accurate in state norm. Since
$\log(1/\eta)=2\log(1/\varepsilon)+O(1)$, the coarser bound
\begin{equation}
  Q_{\mathrm{state}}(\omega|t|,\varepsilon)
  =
  O\!\left(
    \omega|t|+\log\frac{1}{\varepsilon}
  \right)
  \label{eq:gilyen-simple}
\end{equation}
is sufficient for all estimates below.

\paragraph{Block-encoding the transfer Hamiltonian.}

Recall that
\begin{equation}
  H_s
  =
  \sum_{i=1}^{2n-s}h_i,
  \qquad
  h_i
  =
  \sigma_a^+\otimes a_{i+s}^\dagger a_i
  +
  \sigma_a^-\otimes a_i^\dagger a_{i+s}\,.
  \label{eq:qsvt-transfer-bonds}
\end{equation}
Here, the subscript $a$ denotes the transfer ancilla. Under the
Jordan--Wigner transformation,
\begin{align}
  h_i
  &=
  \frac{1}{4}
  X_aX_iX_{i+s}Z_i^{(s)}
  +
  \frac{1}{4}
  X_aY_iY_{i+s}Z_i^{(s)}
  +
  \frac{1}{4}
  Y_aY_iX_{i+s}Z_i^{(s)}
  -
  \frac{1}{4}
  Y_aX_iY_{i+s}Z_i^{(s)}\,,
  \label{eq:qsvt-bond-pauli}
\end{align}
where
\begin{equation}
  Z_i^{(s)}
  :=
  \prod_{k=i+1}^{i+s-1}Z_k \,.
\end{equation}
For $s=1$, this product is the identity. Thus, every bond term $h_i$ is a sum
of four Pauli strings, each with coefficient of magnitude $1/4$.

Pauli strings belonging to different bonds cannot coincide. Indeed, if
$i\neq j$, then the endpoint pairs $\{i,i+s\}$ and $\{j,j+s\}$ differ. At an
endpoint belonging to only one of these pairs, one string contains $X$ or
$Y$, while the other contains either $I$ or $Z$. The Pauli one-norm of $H_s$
is therefore exactly
\begin{equation}
  \omega_s
  =
  (2n-s)
  \left(
    4\cdot\frac{1}{4}
  \right)
  =
  2n-s\,.
  \label{eq:pauli-1norm}
\end{equation}

Write the Pauli decomposition as
\begin{equation}
  H_s
  =
  \sum_{i=1}^{2n-s}
  \sum_{\beta\in\{00,01,10,11\}}
  c_{i,\beta}P_{i,\beta}\,,
\end{equation}
where the four values of $\beta$ label the terms in the order displayed in
\eqref{eq:qsvt-bond-pauli}. Each coefficient has magnitude $1/4$, and
$c_{i,11}=-1/4$.

The preparation circuit creates the uniform superposition
\begin{equation}
  \mathrm{PREP}_s\ket{0}
  =
  \frac{1}{\sqrt{4(2n-s)}}
  \sum_{i=1}^{2n-s}
  \sum_{\beta\in\{00,01,10,11\}}
  \ket{i,\beta}\,.
  \label{eq:qsvt-prep}
\end{equation}
The selection circuit is
\begin{equation}
  \mathrm{SELECT}_s
  =
  \sum_{i=1}^{2n-s}
  \sum_{\beta\in\{00,01,10,11\}}
  \ket{i,\beta}\!\bra{i,\beta}
  \otimes
  \operatorname{sgn}(c_{i,\beta})P_{i,\beta}\,.
  \label{eq:qsvt-select}
\end{equation}
Unused computational-basis labels are defined to act as the identity. They
have zero amplitude in the state prepared by $\mathrm{PREP}_s$.

Define
\begin{equation}
  U_s
  :=
  \left(
    \mathrm{PREP}_s^\dagger\otimes I
  \right)
  \mathrm{SELECT}_s
  \left(
    \mathrm{PREP}_s\otimes I
  \right)\,.
  \label{eq:qsvt-block-encoding}
\end{equation}
This is the circuit of Figure~\ref{fig:intro-circuits}(c), with
$V=\mathrm{PREP}_s$ and the controlled Paulis forming $\mathrm{SELECT}_s$.
Its upper-left block is
\begin{align}
  (\bra{0}\otimes I)U_s(\ket{0}\otimes I)
  &=
  \frac{1}{4(2n-s)}
  \sum_{i=1}^{2n-s}
  \sum_{\beta\in\{00,01,10,11\}}
  \operatorname{sgn}(c_{i,\beta})P_{i,\beta}
  =
  \frac{H_s}{2n-s}
  =
  \frac{H_s}{\omega_s}\,.
  \label{eq:qsvt-block-verification}
\end{align}
Consequently, $U_s$ is an exact
$(\omega_s,\mathfrak a_s,0)$-block-encoding of $H_s$, with
\begin{equation}
  \mathfrak a_s
  =
  \left\lceil\log_2(2n-s)\right\rceil+2
  =
  O(\log n)
\end{equation}
work qubits.

A generic exact circuit for the range-state preparation
\eqref{eq:qsvt-prep} uses $O(n)$ one- and two-qubit gates. This cost is included in the gate estimates below.
\begin{proposition}[Cost of $\mathrm{SELECT}_s$]
\label{prop:structured-select}
For fixed $s$, the circuit $\mathrm{SELECT}_s$ can be implemented using
$O(n\log n)$ gates and $O(\log n)$ work qubits. The same gate count holds if
we add a fixed number of control qubits.
\end{proposition}

\begin{proof}
The label register stores the bond index $i$ and the two-bit label $\beta$.
We first compute a bit $v$ that is equal to one exactly when
$1\leq i\leq 2n-s$. When $v=1$, we compute the other endpoint
$i_{\mathrm R}=i+s$, and when $v=0$ we set $i_{\mathrm R}=0$. The register
holding $i_{\mathrm R}$ has $\lceil\log_2(2n+1)\rceil$ qubits, which is
enough because $i_{\mathrm R}\leq2n$ for every valid bond.

We then run through the physical sites $k=1,\ldots,2n$. At each site, the
circuit checks whether $k=i$, $k=i_{\mathrm R}$, or $i<k<i_{\mathrm R}$. It
applies the endpoint Pauli operators when $k=i$ or $k=i_{\mathrm R}$ and
applies $Z_k$ when $i<k<i_{\mathrm R}$. The choice of endpoint Pauli is
determined by $\beta$, as in Eq.~\eqref{eq:qsvt-bond-pauli}.

After all physical sites have been processed, the circuit applies $X_a$ for
$\beta=00,01$ and $Y_a$ for $\beta=10,11$. For $\beta=11$, it also applies
the minus sign in Eq.~\eqref{eq:qsvt-bond-pauli}. All these operations are
turned off when $v=0$. The circuit then uncomputes $i_{\mathrm R}$ and $v$.

Each comparison of two $O(\log n)$-bit numbers uses $O(\log n)$ gates. We
perform a constant number of comparisons at each of the $2n$ sites. The
total gate count is therefore $O(n\log n)$. The same comparison registers
are reused at every site, so the circuit needs only $O(\log n)$ work qubits.
Adding a fixed number of controls to the whole circuit changes the gate count
by only a constant factor in our gate model.
\end{proof}

Combining the $O(n)$ cost of $\mathrm{PREP}_s$ with
Proposition~\ref{prop:structured-select}, one use of $U_s$, $U_s^\dagger$, or
either circuit with a fixed number of extra controls costs
\begin{equation}
  g_s
  :=
  \operatorname{cost}(U_s)
  =
  O(n\log n)\,.
  \label{eq:qsvt-query-gates}
\end{equation}
This includes the controlled versions needed by QSVT and by the coherent
QCT circuit.

\paragraph{A fixed cycle type.}

The transfer that adds the $j$-th cycle of length $s$ has evolution time
$t_{s,j}=\pi/(2\sqrt{sj})$. Since $\omega_s=2n-s$,
Eq.~\eqref{eq:gilyen-simple} shows that a state error of at most
$\varepsilon_{s,j}$ requires
\begin{equation}
  Q_{s,j}
  =
  O\!\left(
    \frac{2n-s}{\sqrt{sj}}
    +
    \log\frac{1}{\varepsilon_{s,j}}
  \right)
  \label{eq:qsvt-per-transfer}
\end{equation}
uses of $U_s$ or $U_s^\dagger$.

Fix a cycle type $\nu=(n^{m_n}\cdots2^{m_2}1^{m_1})\vdash n$ with
$\ell(\nu)=\sum_sm_s$ transfers, and give each transfer error
$\varepsilon_{s,j}=\delta/\ell(\nu)$, so that the errors add to at most
$\delta$. Using $\sum_{j=1}^{m}1/\sqrt j\leq2\sqrt m$,
\begin{equation}
  Q_{\mathrm{CSP}}(\nu)
  =
  O\!\left(
    \sum_{s=1}^{n}
    (2n-s)\sqrt{\frac{m_s}{s}}
    +
    \ell(\nu)\log\frac{\ell(\nu)}{\delta}
  \right)\,.
  \label{eq:qsvt-total-queries}
\end{equation}
To bound the first sum, set $x_s:=s\,m_s$, so that $\sum_sx_s=n$. The
Cauchy--Schwarz inequality gives
\begin{align}
  \left(
    \sum_{s=1}^{n}
    \frac{(2n-s)\sqrt{x_s}}{s}
  \right)^2
  &\leq
  \left(
    \sum_{s=1}^{n}x_s
  \right)
  \left(
    \sum_{s=1}^{n}\frac{(2n-s)^2}{s^2}
  \right)
  \leq
  n\left(
    4n^2\sum_{s=1}^{n}\frac{1}{s^2}
  \right)
  =
  O(n^3)\,,
  \label{eq:qsvt-worst}
\end{align}
so the first sum in Eq.~\eqref{eq:qsvt-total-queries} is $O(n^{3/2})$.
Each query costs $O(n\log n)$ gates, and the extra $O(\log n)$ gates used by
QSVT per query are smaller than this. It follows that
\begin{equation}
  G_{\mathrm{CSP}}(\nu)
  =
  O\!\left(
    n^{5/2}\log n
    +
    n\ell(\nu)\log n
    \log\frac{\ell(\nu)}{\delta}
  \right)\,.
  \label{eq:qsvt-fixed-gates}
\end{equation}

\begin{theorem}[QSVT bound for CSP]
\label{thm:qsvt-csp}
For every cycle type $\nu\vdash n$ and $0<\delta<1$, the CSP circuit with
QSVT-simulated Hamiltonian transfers prepares $\ket{\chi_\nu}$ to state error at most
$\delta$ using
\begin{equation}
  G_{\mathrm{CSP}}^{\mathrm{worst}}
  =
  O\!\left(
    n^{5/2}\log n
    +
    n^2\log n\log\frac{n}{\delta}
  \right)
  \label{eq:qsvt-fixed-worst}
\end{equation}
one- and two-qubit gates and $2n+O(\log n)$ qubits.
\end{theorem}
\begin{proof}
Set $\ell(\nu)\leq n$ in \eqref{eq:qsvt-fixed-gates}. The qubit count is the
$2n$ system qubits, the transfer ancilla, the $\mathfrak a_s=O(\log n)$
block-encoding qubits, and the $O(\log n)$ work qubits of
Proposition~\ref{prop:structured-select} and of QSVT.
\end{proof}
This is an upper bound for our circuit, and we do not claim that it is
optimal.

\section{Quantum Character Transform}
\label{sec:qct}

We now generalize our CSP algorithm to the full Quantum Character Transform \eqref{eq:qct}. This is the unitary operation that maps coherent superpositions of cycle types to coherent superpositions of character column states
\begin{equation}
\mathrm{QCT}_{S_n} = \sum_{\nu\vdash n}\ket{\chi_\nu}\bra{\nu} = \sum_{\lambda,\nu\vdash n} \frac{\chi^\lambda(\nu)}{\sqrt{z_\nu}} \ket{\lambda}\bra{\nu}\,.
\label{eq:qct-explicit}
\end{equation}
The basic idea is to run our CSP algorithm via a control register encoding the cycle types, $\ket{\nu}\seed\to\ket{\nu}\ket{\chi_\nu}$, and uncompute the cycle type to get $\ket{0}\ket{\chi_\nu}$.
This realizes $\mathrm{QCT}_{S_n}$ up to a swapping of registers.
To improve efficiency, we interleave the uncomputing subroutines with the Hamiltonian simulation subroutines in CSP. 

Our full QCT circuit is summarized in Figure~\ref{fig:qct}.
The cycle type $\nu=(n^{m_n}\ldots 2^{m_2}1^{m_1})$ is held in count registers $C_n,\ldots,C_2, C_1$.
Since a cycle of length $s$ takes up $s$ of the $n$ elements, the multiplicity satisfies
$m_s \le \lfloor n/s\rfloor$, and so each cycle count register $C_s$ needs
$q_s := \lceil\log_2(\lfloor n/s\rfloor + 1)\rceil$ qubits to encode an integer in
$\{0,\ldots,\lfloor n/s\rfloor\}$. The label $\ket{\nu}$ is shorthand for the
product of these registers,
\begin{equation}
  \ket{\nu}_C := \ket{m_n}_{C_n}\cdots\ket{m_1}_{C_1}\,,
  \qquad
  \sum_{s=1}^n s\,m_s = n\,,
  \label{eq:qct-cycletype}
\end{equation}
and the constraint $\sum_s s\,m_s = n$ picks out the register configurations
that encode a genuine cycle type of $S_n$. The full input is
\begin{equation}
  \ket{\Psi_{\mathrm{in}}}
  = \ket{\nu}_C \ket{0}_a \seed_\Lambda = \ket{\nu}_C \ket{0}_a \ket{1^n0^n}_\Lambda
  \label{eq:qct-input}
\end{equation}
with $\Lambda$ the $2n$-qubit system register. The registers shrink with
$s$, and their total size is linear in $n$, since
\begin{equation}
  \sum_{s=1}^{n}q_s
  \le
  n+\sum_{s=1}^{n}\log_2\Big(1+\frac ns\Big)
  =
  n+\log_2\prod_{s=1}^{n}\frac{n+s}{s}
  =
  n+\log_2\binom{2n}{n}
  <
  3n\,
  \label{eq:register-count}
\end{equation}
Hence, our QCT implementation uses $2n+O(n)$ qubits, compared to the $2n+1$ qubits sufficient for CSP.
Cycle lengths are processed in decreasing order, and for each $s$ our QCT circuit performs a
deposit step followed by an erase step.

\paragraph{Deposit.}
For a fixed cycle type, recall our CSP algorithm \eqref{eq:transfer-result} applies exactly $m_s$ transfers at
length $s$. In superposition, however, $m_s$ is not fixed, so we instead apply the
transfer $W_{s,j} = (iX_a \otimes\Id) e^{-iH_{s} t_{s,j}}$  at every slot $j = 1,\ldots,L_s$ with $L_s = \lfloor n/s\rfloor$, each
conditioned on the control register $C_s$,
\begin{equation}
  V_{s,j}:\quad \ket{x}_{C_s}\ket{\psi}
  \;\longmapsto\;
  \begin{cases}
    \ket{x}_{C_s}\,W_{s,j}\ket{\psi}, & j \le x,\\[2pt]
    \ket{x}_{C_s}\,\ket{\psi}, & j > x.
  \end{cases}
  \label{eq:controlled-transfer}
\end{equation}
Note the control applies to every gate component of $W_{s,j}$, including the $iX_a$ correction to the ancilla. 
The $s$-th deposit operation is then $D_s = \prod_{j=1}^{L_s} V_{s,j}$.
On the branch with multiplicity $m_s$, the slots $j\le m_s$ act and the slots
$j > m_s$ are identity, so exactly $m_s$ transfers are applied, and the system advances from $\ket{\Phi^{>s}}$ to $\ket{\Phi_{s,m_s}}$. 

\paragraph{Erase.}
After the deposit by $D_s$, the register $C_s$ still holds $\ket{m_s}$, which must be erased to $\ket{0}$ to preserve coherence between different cycle types.
The multiplicity $m_s$ can be recovered from the system state using the
operator
\begin{equation}
  \mathcal{N}_s := \tfrac{1}{s}\,J_s J_s^\dagger\,,
  \qquad \text{as} \qquad
  \mathcal{N}_s\ket{\Phi_{s,m}} = m\,\ket{\Phi_{s,m}}\,.
\label{eq:number-operator-eigenvalue}
\end{equation}
The eigenvalue follows from the raising and lowering relations
\eqref{eq:ladder-raiselower}: $J_s^\dagger\ket{\Phi_{s,m}} =
\sqrt{sm}\,\ket{\Phi_{s,m-1}}$ and then $J_s\ket{\Phi_{s,m-1}} =
\sqrt{sm}\,\ket{\Phi_{s,m}}$, so $J_sJ_s^\dagger\ket{\Phi_{s,m}} =
sm\ket{\Phi_{s,m}}$. Since $\ket{\Phi_{s,m_s}}$ is an eigenstate of $\mathcal N_s$ with eigenvalue $m_s$, the register $C_s$ can be erased by subtracting this eigenvalue from it.
Writing $F_s$ for
the quantum Fourier transform on $C_s$ and
$R_s = 2^{\lceil\log_2(L_s+1)\rceil}$, the subtraction can be implemented via the following phase estimation protocol:
\begin{equation}
  \mathrm{SUB}_s := F_s^\dagger\, K_s^-\, F_s\,,
  \qquad
  K_s^-:\ \ket{q}_{C_s}\ket{\psi} \mapsto
  e^{-2\pi i q\,\mathcal{N}_s/R_s}\,\ket{q}_{C_s}\ket{\psi}\,.
  \label{eq:qft-erasure}
\end{equation}
Note $\mathrm{SUB}_s\ket{m_s}_{C_s}\ket{\Phi_{s,m_s}} =
\ket{0}_{C_s}\ket{\Phi_{s,m_s}}$, since $\mathcal{N}_s$ acts as the number
$m_s$ on the system.


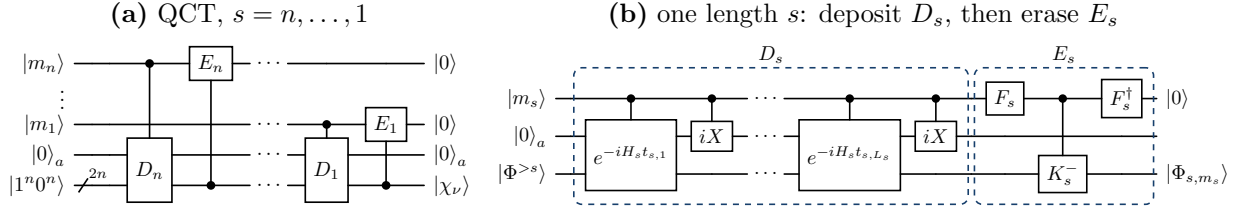
\begin{figure}[t]
\centering
\begin{minipage}[t]{0.39\linewidth}
\centering
{\small\textbf{(a)} QCT, $s=n,\ldots,1$}\\[3pt]
\resizebox{0.98\linewidth}{!}{%
\begin{quantikz}[row sep={0.52cm,between origins}, column sep=0.3cm, font=\small, transparent, every gate/.style={draw=navy,fill=white,thick,inner xsep=4pt}, wire types={q,n,q,q,q}]
\lstick{$\ket{m_n}$} & & & \ctrl{3} & \gate{E_n} & \ \cdots\ & & & \rstick{$\ket{0}$} \\
\lstick{$\vdots$} & & & & & & & & \\
\lstick{$\ket{m_1}$} & & & & & \ \cdots\ & \ctrl{1} & \gate{E_1} & \rstick{$\ket{0}$} \\
\lstick{$\ket{0}_a$} & & & \gate[2]{D_n} & & \ \cdots\ & \gate[2]{D_1} & & \rstick{$\ket{0}_a$} \\
\lstick{$\ket{1^n0^n}$} & \qwbundle{2n} & & & \ctrl{-4} & \ \cdots\ & & \ctrl{-2} & \rstick{$\ket{\chi_\nu}$}
\end{quantikz}}
\end{minipage}\hfill
\begin{minipage}[t]{0.61\linewidth}
\centering
{\small\textbf{(b)} one length $s$: deposit $D_s$, then erase $E_s$}\\[3pt]
\resizebox{0.98\linewidth}{!}{%
\begin{quantikz}[row sep={0.66cm,between origins}, column sep=0.26cm, font=\small, transparent, every gate/.style={draw=navy,fill=white,thick,inner xsep=4pt}, wire types={q,q,q}]
\lstick{$\ket{m_s}$} & & \ctrl{1}\gategroup[3,steps=5,style={dashed,rounded corners,draw=navy,inner xsep=2pt,inner ysep=2pt}]{$D_s$} & \ctrl{1} & \ \cdots\ & \ctrl{1} & \ctrl{1} & & \gate{F_s}\gategroup[3,steps=3,style={dashed,rounded corners,draw=navy,inner xsep=2pt,inner ysep=2pt}]{$E_s$} & \ctrl{2} & \gate{F_s^\dagger} & \rstick{$\ket{0}$} \\
\lstick{$\ket{0}_a$} & & \gate[2]{e^{-iH_st_{s,1}}} & \gate{iX} & \ \cdots\ & \gate[2]{e^{-iH_st_{s,L_s}}} & \gate{iX} & & & & & \\
\lstick{$\ket{\Phi^{>s}}$} & & & & \ \cdots\ & & & & & \gate{K_s^-} & & \rstick{$\ket{\Phi_{s,m_s}}$}
\end{quantikz}}
\end{minipage}
\caption{(a) Our QCT circuit. Each count register $\ket{m_s}$ controls the
deposit $D_s$ and is then erased by $E_s=\mathrm{SUB}_s$ using the system
state. (b) The block for one length $s$. Slot $j$ of the deposit fires when
$m_s\ge j$, so exactly $m_s$ transfers act and the system advances from
$\ket{\Phi^{>s}}$ to $\ket{\Phi_{s,m_s}}$. The erase step applies the
Fourier-controlled phase $K_s^-$ of Eq.~\eqref{eq:qft-erasure} between $F_s$
and $F_s^\dagger$, which returns the register to $\ket{0}$.
}
\label{fig:qct}
\end{figure}

\paragraph{Output.}
After every cycle length has been processed, all count registers are back in
$\ket{0}$ and the system holds the character column state,
\begin{equation}
  \ket{\Psi_{\mathrm{out}}}
  = \ket{0}_{C_n}\cdots\ket{0}_{C_1}  \ket{0}_a \ket{\chi_\nu}_\Lambda\,.
  \label{eq:qct-output}
\end{equation}
Each deposit and erase step is unitary. On the valid cycle-type input subspace,
column orthogonality ensures that distinct states $\ket{\nu}$ are mapped to
orthogonal character states. Hence, the circuit implements the QCT on this
subspace and, by linearity, on any superposition of valid cycle types.

\subsection{QSVT-based implementation and resource estimate}
\label{sec:qct-cost}

We now describe a QSVT-based implementation of the QCT protocol outlined above. We first analyze the costs associated with the deposit operations, then the erasure operations, before giving the complete QCT gate count to prove Theorem~\ref{thm:intro-qct}.

\paragraph{Deposit operations $D_s$.}
The deposit operation $D_s$ consists of a sequence of controlled transfers $V_{s,j}$ \eqref{eq:controlled-transfer}.
The effect of the control by $\ket{m_s}_{C_s}$ is such that the transfer operation $W_{s,j}$ in slot $(s,j)$ is applied only when $j \le m_s$, and the circuit runs all slots $1\leq s\leq n$,
$1\leq j\leq L_s:=\lfloor n/s\rfloor$. 
The total number of slots (and transfer operations) is
\begin{equation}
  N_{\mathrm{slot}}(n)
  :=
  \sum_{s=1}^{n}L_s
  =
  \Theta(n\log n)\,.
  \label{eq:qct-slot-count}
\end{equation}

The part of the query count in \eqref{eq:gilyen-simple} that depends on the evolution times satisfies
\begin{align}
  \sum_{s,j} \omega_s |t_{s,j}| = \frac{\pi}{2}\sum_{s=1}^{n}
  \sum_{j=1}^{L_s}\frac{2n-s}{\sqrt{sj}}
  \leq
  \pi\sum_{s=1}^{n}
  \frac{2n-s}{\sqrt s}\sqrt{\frac ns}
  =
  \pi\sqrt n
  \sum_{s=1}^{n}\frac{2n-s}{s}
  =
  O(n^{3/2}\log n)\,.
\end{align}
We give every slot a state error budget of $\delta_{\mathrm{dep}}/N_{\mathrm{slot}}(n)$,
so that the total state error is at most $\delta_{\mathrm{dep}}$. The number of queries to $U_s$ is
\begin{equation}
  Q_{\mathrm{dep}}
  =
  O\!\left(
    n^{3/2}\log n
    +
    n\log n
    \log\frac{n\log n}{\delta_{\mathrm{dep}}}
  \right)\,,
  \label{eq:qct-controlled-queries}
\end{equation}
and multiplying by the cost of one query from \eqref{eq:qsvt-query-gates} gives
\begin{equation}
  G_{\mathrm{dep}}
  =
  O\!\left(
    n^{5/2}\log^2 n
    +
    n^2\log^2 n
    \log\frac{n\log n}{\delta_{\mathrm{dep}}}
  \right)\,.
  \label{eq:qct-controlled-gates}
\end{equation}
Computing and uncomputing the controls $[m_s\geq j]$ over all slots uses
$O(n\log^2 n)$ additional gates, which is smaller than the bound above.

The coherent circuit, therefore, implements
\begin{equation}
  \ket{\nu}_{C} \ket{0}_a \seed_{\Lambda}
  \longmapsto
  \ket{\nu}_{C} \ket{0}_a \ket{\chi_\nu}_{\Lambda}
\end{equation}
with state error at most $\delta_{\mathrm{dep}}$ for every $\nu\vdash n$.
The gate count includes every controlled slot.

The deposit steps of the QCT therefore use $\widetilde O(n^{5/2})$ gates. 
\paragraph{Singular-value-amplified erasure $E_s$.}
Next, we describe and analyze the implementation of the erasure steps $E_s$ within the same budget as \eqref{eq:qct-controlled-gates} above.

After the deposit $D_s$ for all current operators at length $s$, the count register $C_s$ contains the
multiplicity $\ket{m_s}$ and the system is in $\ket{\Phi_{s,m_s}}$. To erase $\ket{m_s}$, we use phase estimation on the operator $\mathcal N_s = \frac{1}{s}J_sJ_s^\dagger$, since $\mathcal N_s\ket{\Phi_{s,m}} = m\ket{\Phi_{s,m}}$ as explained earlier in \eqref{eq:number-operator-eigenvalue}.

The register $C_s$ has $q_s=\lceil\log_2(L_s+1)\rceil$ qubits and dimension
$R_s:=2^{q_s}$. We need the following quantum Fourier transform for the $C_s$ register
\begin{equation}
  F_s\ket{x} =\frac{1}{\sqrt{R_s}}
  \sum_{y=0}^{R_s-1}
  e^{2\pi ixy/R_s}\ket{y},
\end{equation}
which can be implemented with at most $O(q_s^2)=O(\log^2 n)$ gates.
Then the phase-estimation or subtraction operation $\mathrm{SUB}_s=F_s^\dagger K_s^-F_s$ of
Eq.~\eqref{eq:qft-erasure}, with
\begin{equation}
  K_s^-:\
  \ket{y}\ket{\psi}
  \longmapsto
  e^{-2\pi iy\mathcal N_s/R_s}
  \ket{y}\ket{\psi}\,,
  \label{eq:qft-erasure-resource}
\end{equation}
maps $\ket{m_s}_{C_s}\ket{\Phi_{s,m_s}}_\Lambda$ to $\ket{0}_{C_s}\ket{\Phi_{s,m_s}}_\Lambda$,
because the phase in $K_s^-$ cancels the Fourier phase of $\ket{m_s}$ by
\eqref{eq:number-operator-eigenvalue}. The outstanding issue is an efficient implementation of the evolution generated by $\mathcal N_s$.

A direct construction block-encodes $J_sJ_s^\dagger$ as the square of the
block-encoding $U_s$ and divides by $s$. Its normalization is
$\omega_s^2/s=(2n-s)^2/s$, and using it at every length gives
$\widetilde O(n^3)$ gates. We do better because the algorithm only accesses a small part of the spectrum of $H_s$.

\subparagraph{$\Diamond$ The part of $H_s$'s spectrum used by the algorithm.}

The square of the transfer Hamiltonian is
\begin{equation}
  H_s^2
  =
  \ket{0}\!\bra{0}_a \otimes J_s^\dagger J_s
  +
  \ket{1}\!\bra{1}_a \otimes J_sJ_s^\dagger\,.
  \label{eq:H-square}
\end{equation}
For every $1\leq j\leq L_s$, the two-dimensional space
\begin{equation}
  \mathcal K_{s,j}
  :=
  \operatorname{span} \left\{
    \ket{0}_a\ket{\Phi_{s,j-1}}_\Lambda,
    \ket{1}_a\ket{\Phi_{s,j}}_\Lambda
  \right\}
\end{equation}
is invariant under $H_s$, on which its eigenvalues are
$E_{s,j}^{\pm}=\pm\sqrt{sj}$, and $\ket{1}_a \ket{\Phi_{s,0}}_\Lambda$ has eigenvalue
zero. Since $j\leq L_s$,
\begin{equation}
  \left|E_{s,j}^{\pm}\right|
  \leq
  \sqrt{sL_s}
  \leq
  \sqrt n\,.
  \label{eq:relevant-eigenvalue-bound}
\end{equation}
The norm of $H_s$ can be larger. The bound \eqref{eq:relevant-eigenvalue-bound}
only holds for the eigenvalues of $H_s$ on the states that the algorithm
visits.

\subparagraph{$\Diamond$ Amplifying these eigenvalues.}

The block-encoding $U_s$ contains $H_s/\omega_s$ with $\omega_s=2n-s$. We
use uniform singular value amplification~\cite[Theorem~30]{Gilyen_2019}.
For $\gamma>1$ and $\Delta,\eta\in(0,1/2)$, that result gives an odd
polynomial $P$ with $|P(x)|\leq1$ on $[-1,1]$ and
\begin{equation}
  \left|
    \frac{P(x)}{\gamma x}-1
  \right|
  \leq
  \eta
  \qquad
  \text{for }
  |x|\leq\frac{1-\Delta}{\gamma}\,,
\end{equation}
of degree $O(\gamma\log(\gamma/\eta))$ for fixed $\Delta$. We choose
\begin{equation}
  \gamma_s
  :=
  \frac{\omega_s}{3\sqrt n}\,,
  \qquad
  \Delta:=\frac12\,.
  \label{eq:amplification-parameters}
\end{equation}
For $n>9$, $\gamma_s>1$ for every $s$, and the finitely many smaller values
of $n$ can be implemented directly. For every eigenvalue $E$ in
\eqref{eq:relevant-eigenvalue-bound},
\begin{equation}
  \left|\frac{E}{\omega_s}\right|
  \leq
  \frac{\sqrt n}{\omega_s}
  =
  \frac{1}{3\gamma_s}
  <
  \frac{1-\Delta}{\gamma_s}\,,
  \label{eq:usva-interval}
\end{equation}
so the amplification theorem applies to every eigenvalue used by the
algorithm.

Let $p_s$ be the resulting polynomial and $B_s:=p_s(H_s/\omega_s)$. QSVT
gives a block-encoding of $B_s$, and on the relevant eigenvalues
\begin{equation}
  p_s\!\left(\frac{E}{\omega_s}\right)
  =
  \frac{E}{3\sqrt n}(1+e)\,,
  \qquad
  |e|\leq\eta_s\,.
  \label{eq:usva-relative-error}
\end{equation}
The polynomial stays bounded by one on all of $[-1,1]$, so the QSVT circuit
is a valid unitary even where $H_s$ has larger eigenvalues outside the
states used by the algorithm.

We now square the amplified block-encoding. By the product rule for
block-encodings~\cite[Lemma~53]{Gilyen_2019}, two copies of its work
register give a block-encoding of $B_s^2$, which we use as a block-encoding
of
\begin{equation}
  \widetilde{\mathcal N}_s
  :=
  \frac{9n}{s}B_s^2
  \qquad
  \text{with normalization}
  \qquad
  \alpha_s^{(\mathcal N)}
  =
  \frac{9n}{s}\,.
  \label{eq:amplified-number-operator}
\end{equation}
This is smaller than the direct normalization $(2n-s)^2/s$ by a factor of
order $n$. Because $p_s$ is odd, $B_s^2$ is a polynomial in $H_s^2$ and
preserves the two ancilla sectors in \eqref{eq:H-square}. We flip the
transfer ancilla to $\ket{1}_a$ before the Fourier-controlled evolution and
back afterwards, so that the operator acting is
$\tfrac1sJ_sJ_s^\dagger=\mathcal N_s$.

\subparagraph{$\Diamond$ Amplification error.}

On $\mathcal K_{s,j}$, where $E^2=sj$, Eq.~\eqref{eq:usva-relative-error}
gives the eigenvalue
$\widetilde j=\tfrac{9n}{s}\big[\tfrac{E}{3\sqrt n}(1+e)\big]^2=j(1+e)^2$
of $\widetilde{\mathcal N}_s$. If $\eta_s\leq1/2$, then
\begin{equation}
  |\widetilde j-j|
  \leq
  3\eta_sj
  \leq
  3\eta_s\frac{n}{s}\,.
  \label{eq:number-eigenvalue-error}
\end{equation}
The Fourier-controlled evolution is implemented one register qubit at a
time, with evolution times $\theta_{s,b}:=2\pi2^b/R_s$ for $0\leq b<q_s$,
whose sum is $2\pi(R_s-1)/R_s<2\pi$. The phase error from replacing $j$ by
$\widetilde j$ is therefore at most $2\pi|\widetilde j-j|\leq6\pi\eta_sn/s$.
Choosing
\begin{equation}
  \eta_s
  :=
  \frac{\delta_{\mathrm{er}}s}{12\pi n^2}
  \label{eq:amplification-error-budget}
\end{equation}
makes the amplification error at most $\delta_{\mathrm{er}}/(2n)$ at each
length, and at most $\delta_{\mathrm{er}}/2$ in total. Since
$\gamma_s=O(\sqrt n)$, this choice gives polynomial degree
\begin{equation}
  d_s
  =
  O\!\left(
    \sqrt n\log\frac{n}{\delta_{\mathrm{er}}}
  \right)\,,
  \label{eq:amplification-degree}
\end{equation}
so one query to the amplified block-encoding uses $O(d_s)$ queries to $U_s$
or $U_s^\dagger$.

\subparagraph{$\Diamond$ Simulation error and gate count.}

The number of Fourier-controlled simulations is
$N_{\mathrm{phase}}(n):=\sum_{s=1}^{n}q_s=\Theta(n)$ by
\eqref{eq:register-count}. We give each of them a state error
$\delta_{\mathrm{er}}/(2N_{\mathrm{phase}}(n))$, so that these errors add
to at most $\delta_{\mathrm{er}}/2$. At fixed $s$, the total number of
queries to the block-encoding of $B_s^2$ is
\begin{equation}
  O\!\left(
    \frac ns
    +
    q_s\log
    \frac{N_{\mathrm{phase}}(n)}{\delta_{\mathrm{er}}}
  \right)\,,
  \label{eq:erasure-amplified-queries-per-s}
\end{equation}
where the first term follows from $\alpha_s^{(\mathcal N)}=9n/s$ and
$\sum_b\theta_{s,b}<2\pi$, and the second is the precision cost of the
$q_s$ simulations. Using $\sum_sn/s=O(n\log n)$ and $\sum_sq_s=O(n)$
together with \eqref{eq:amplification-degree}, the total number of queries
to the original block-encodings $U_s$ is
\begin{equation}
  Q_{\mathrm{er}}
  =
  O\!\left[
    n^{3/2}
    \log\frac{n}{\delta_{\mathrm{er}}}
    \left(
      \log n+
      \log\frac{n}{\delta_{\mathrm{er}}}
    \right)
  \right]\,.
  \label{eq:erasure-base-queries}
\end{equation}
Each query costs $O(n\log n)$ gates by \eqref{eq:qsvt-query-gates}, so
\begin{equation}
  G_{\mathrm{er}}
  =
  O\!\left[
    n^{5/2}\log n
    \log\frac{n}{\delta_{\mathrm{er}}}
    \left(
      \log n+
      \log\frac{n}{\delta_{\mathrm{er}}}
    \right)
  \right]\,.
  \label{eq:erasure-gate-count}
\end{equation}
The Fourier transforms use $O(\sum_sq_s^2)=O(n)$ gates, and the
comparisons, QSVT phase gates, and reflections have smaller gate counts than
\eqref{eq:erasure-gate-count}.

\paragraph{Complete QCT gate count upper bound.}

Let $\mathcal P_{\mathrm{valid},0}$ project onto the states for which
$\sum_ss\,m_s=n$, the system is in the seed state, and all transfer and
work qubits are zero. On this input space, the exact circuit (where all Hamiltonian simulations are error-free) acts as
\begin{equation}
  U_{\mathrm{QCT}}:
  \quad
  \ket{\nu}_{C}\ket{0}_a \ket{0}_{\mathrm w}\seed_{\Lambda}
  \longmapsto
  \ket{0}_{C}
  \ket{0}_a\ket{0}_{\mathrm w}\ket{\chi_\nu}_{\Lambda}\,.
\end{equation}

\begin{theorem}[Quantum character transform; also Theorem~\ref{thm:intro-qct}]
\label{thm:qct-total}
For $0<\delta<1$, the circuit $\tilde{U}_{\rm QCT}$ described above with $\delta_{\mathrm{dep}}=\delta_{\mathrm{er}}=\delta/2$
satisfies
\begin{equation}
  \left\|
    \left(
      \widetilde U_{\mathrm{QCT}}-U_{\mathrm{QCT}}
    \right)
    \mathcal P_{\mathrm{valid},0}
  \right\|
  \leq
  \delta
  \label{eq:qct-restricted-error}
\end{equation}
and uses
\begin{align}
  G_{\mathrm{QCT}}
  &=
  O\!\left(
    n^{5/2}\log^2 n
    +
    n^2\log^2 n
    \log\frac{n\log n}{\delta}
  \right)
  +
  O\!\left[
    n^{5/2}\log n
    \log\frac{n}{\delta}
    \left(
      \log n+
      \log\frac{n}{\delta}
    \right)
  \right]
  \label{eq:qct-final-gate-count}
\end{align}
gates and $2n+O(n)$ qubits. For $\delta=1/\operatorname{poly}(n)$,
\begin{equation}
  G_{\mathrm{QCT}}
  =
  \widetilde O(n^{5/2})\,.
  \label{eq:qct-final-tilde}
\end{equation}
\end{theorem}

\begin{proof}
By telescoping the difference of $\tilde U_{QCT}$ and $U_{\rm QCT}$ and using the triangle inequality, we can add the errors from individual unitary steps.
The amplification error is bounded uniformly on this entire spectral subspace, and the simulation bound \eqref{eq:qsvt-block-to-state} is uniform over normalized inputs, including inputs entangled with the count registers. Thus, the same bounds hold for arbitrary initial states in $\mathcal P_{\mathrm{valid},0}$.
In our QSVT construction of $\tilde U_{QCT}$ above, we have budgeted the total deposit error to be $\delta_{\mathrm{dep}}$, and the total erasure error to be $\delta_{\mathrm{er}}$.
Choosing $\delta_{\rm dep}=\delta_{\rm er}=\delta/2$ and adding them yields \eqref{eq:qct-restricted-error}. 
The gate count to guarantee this error level is the sum of
\eqref{eq:qct-controlled-gates} and \eqref{eq:erasure-gate-count}. The
qubits are the $2n$ system qubits, the $O(n)$ multiplicity-register qubits
of \eqref{eq:register-count}, and $O(\log n)$ reusable qubits for the
indices, transfer ancilla, arithmetic, block-encodings, and QSVT.
\end{proof}

\subsection{Application: character row state and defect entropy from the inverse QCT}
\label{sec:rows}

The QCT algorithm in Theorem~\ref{thm:qct-total} run in reverse also gives access to the rows of the normalized character table. 
Indeed, the reversed circuit effectively implements the inverse of QCT~\eqref{eq:qct-explicit}, which when applied to a given Fock state $\ket\lambda$ returns
\begin{equation}
U_{\rm QCT}^{\dag} \ket{0}_C \ket{0}_a\ket{0}_{\rm w}\ket{\lambda}_\Lambda 
=
\Big( \sum_{\nu\vdash n} \frac{1}{\sqrt{z_\nu}} \chi^\lambda(\nu)\ket{\nu}_C\Big)\otimes \ket{0}_a \ket{0}_{\rm w} \seed_\Lambda\,.
\end{equation}
This is the $\lambda$-th row of the normalized character table, at the cost of one QCT.

Measuring the count registers in this state returns a cycle type $\nu$ with probability
\begin{equation}
  P_\lambda(\nu)=\frac{[\chi^\lambda(\nu)]^2}{z_\nu}\,.
  \label{eq:row-distribution}
\end{equation}
These are the weights in the defect entropy~\eqref{eq:defect-entropy-intro}
of a symmetric orbifold CFT with a universal defect labeled by $\lambda$ at
the entangling surface~\cite{GutperleLiRathoreRoumpedakis2024}. Writing
$[\chi^\lambda(\nu)]^2=z_\nu P_\lambda(\nu)$ inside the logarithm,
\begin{equation}
  \Delta S_\lambda
  =
  H(P_\lambda)-\mathbb E_{\nu\sim P_\lambda}\!\left[\ln z_\nu\right]
  \;\le\; 0\,,
  \label{eq:defect-entropy-split}
\end{equation}
where $H$ is the Shannon entropy and the sign follows from
$[\chi^\lambda(\nu)]^2\ge1$ on the support of $P_\lambda$. Since $P_\lambda$ is supported on the $p(n)=e^{O(\sqrt n)}$ partitions of
$n$, the first term is at most $\ln p(n)=O(\sqrt n)$, while $\ln z_\nu = \sum_{s} [m_s \ln s + \ln (m_s!)]$ is
given in closed form by the cycle multiplicities of $\nu$.
Therefore, samples of $\nu$
from the inverse QCT can be used to efficiently compute a lower bound $L_\lambda := -\mathbb{E}_\nu \ln z_\nu$, where $0\le \Delta S_\lambda - L_\lambda \le \ln p(n)$.
This quantity can also be used as a biased estimator for $\Delta S_\lambda$, with additive error $|L_\lambda - \Delta S_\lambda| \le \ln p(n) = O(\sqrt n)$, which corresponds to a vanishing relative error whenever $|\Delta S_\lambda|$ grows faster than $\sqrt n$.

In Figure~\ref{fig:defect-entropy-growth}, we evaluated
\eqref{eq:defect-entropy-intro} numerically for all $\lambda\vdash n$ with
$n\le30$, and find that
$|\Delta S_\lambda|$ grows roughly linearly in $n$ for generic $\lambda$. For
$5\le n\le30$, linear fits give slopes $0.336$ for the Plancherel average and
$0.450$ for the maximum over $\lambda$. At $n=30$, $|\Delta S_\lambda|>\ln p(n)$
for $31\%$ of the irreps, which carry $96\%$ of the Plancherel weight. 
Only shapes close to the trivial or sign representation stay bounded, e.g.,
$|\Delta S_{(n-1,1)}|\approx0.83$ for all $n$.

\begin{figure}[h!]
\centering
\includegraphics[width=\textwidth]{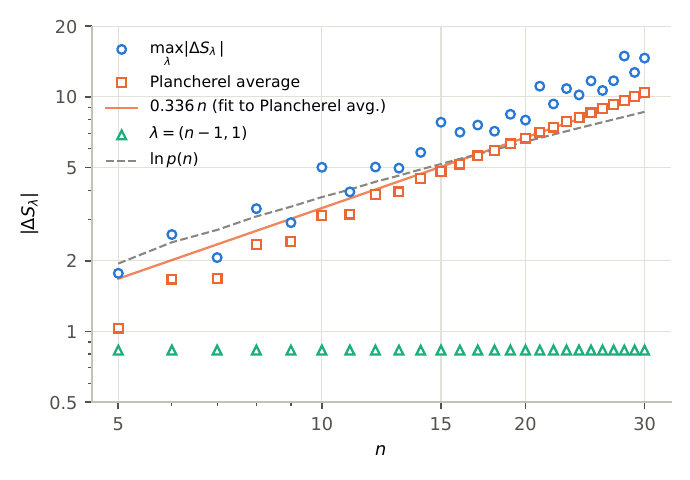}
\caption{Growth of the defect entropy $|\Delta S_\lambda|$ of
Eq.~\eqref{eq:defect-entropy-intro} with $n$, computed from the full character table of $S_n$ for $5\le n\le30$.
The loglog plot shows the maximum $|\Delta S_\lambda|$ over all $\lambda\vdash n$, the average over $\lambda$ weighted by
the Plancherel measure $d_\lambda^2/n!$, 
and the value at the standard representation $\lambda=(n-1,1)$. 
The largest possible value of the entropy term $H(P_\lambda)$ in Eq.~\eqref{eq:defect-entropy-split}, plotted as $\ln p(n)$, bounds the additive error in estimating ${\Delta S}_\lambda$ by $-\mathbb{E}_\nu [\ln z_\nu]$. 
The maximum $\Delta S_\lambda$ is attained at staircase-like shapes, with prominent peaks at the perfect staircases $n=6,10,15,21,28$.}
\label{fig:defect-entropy-growth}
\end{figure}

\section{Classical computational complexity}
\label{sec:classical-hardness}

The argument of Bravyi \emph{et al.}~\cite{BravyiGossetHavlicekSchatzki2025} for quantum advantage in character sampling, reviewed in
Section~\ref{sec:intro-prior}, applies to cycle types with at most
$A\sqrt n/\log n$ cycles. This section shows that the number of cycles alone
does not determine classical hardness: uniform cycle types in this regime,
and their perturbations by a bounded number of exceptional cycles, admit
exact polynomial-time classical samplers.

\subsection{Conditions for quantum advantage}

Recall from Section~\ref{sec:intro-prior} that the hardness argument of
Bravyi \emph{et al.} for sampling
from $P_\nu(\lambda)=[\chi^\lambda(\nu)]^2/z_\nu$ needs $P_\nu$ to be
granular, which holds in the regime \eqref{eq:intro-bghs-regime}.
In order for the quantum advantage to hold, their argument also needs to assume the noncollapse of the polynomial hierarchy, which is standard, and an average-case hardness conjecture that
computing $P_\nu(\lambda)$ exactly is hard for a uniformly random
$\lambda$ with a $1-1/\poly(n)$ success probability.

We observe that this average-case hardness conjecture can only hold if the granularity $\gamma=p(n)/z_\nu$ is at
most polynomial in $n$, since an algorithm that always outputs zero is
correct on at least a $1-1/\gamma$ fraction of the
partitions~\cite{BravyiGossetHavlicekSchatzki2025}. Since
$z_\nu=\prod_s s^{m_s}m_s!\le n^{\ell(\nu)}$, this requires
\begin{equation}
  \ell(\nu)\;\ge\;\frac{\ln p(n)}{\ln n}-O(1)\,,
  \qquad
	\text{where} \qquad
  \ln p(n)=\pi\sqrt{2n/3}-O(\log n)\,.
\end{equation}
The candidate cycle types therefore lie near the upper end of the regime
\eqref{eq:intro-bghs-regime}.

Bravyi \emph{et al.}'s classical MPS algorithm computes any
$\chi^\lambda(\nu)$ (and thus $P_\nu(\lambda)$) exactly in time $n^{O(\ell(\nu))}$, which is polynomial
for a bounded number of cycles but becomes superpolynomial for cycle types with $\ell(\nu)=\Theta(\sqrt n/\log n)$ cycles. We
now show that some of these cycle types are classically easy.

\subsection{Uniform cycles are classically tractable}

We first show that uniform cycle types have an exact polynomial-time
classical sampler, including when the number of cycles lies in the regime
\eqref{eq:intro-bghs-regime}.

\begin{theorem}[Classical sampling for uniform cycles]
\label{thm:equal-cycle-classical}
Let $\nu=(s^m)$ with $n=ms$. The distribution
\begin{equation}
  P_{(s^m)}(\lambda)
  =
  \frac{[\chi^\lambda(s^m)]^2}{s^m m!}
\end{equation}
can be sampled exactly by a classical randomized algorithm with expected
running time $\widetilde O(n^2)$. For any given $\lambda\vdash n$, the
probability $P_{(s^m)}(\lambda)$ can be computed exactly in polynomial
time.
\end{theorem}

\begin{proof}
The Murnaghan-Nakayama rule computes $\chi^\lambda(s^m)$ by successively
removing $m$ border strips of length $s$ from the Young diagram of $\lambda$
\cite[Section~7.17]{Stanley1999}. A border strip is a connected set of
boundary boxes that contains no $2\times2$ block and whose removal leaves a
valid Young diagram.

These removals are described by the $s$-core and $s$-quotient of $\lambda$
\cite[Section~2.7]{JamesKerber1981}. The $s$-core is the diagram that remains
after removing as many length-$s$ border strips as possible, and does not
depend on the order of removal. The $s$-quotient is an $s$-tuple of
partitions $(\lambda^{(0)},\ldots,\lambda^{(s-1)})$. The core and quotient
determine $\lambda$ uniquely, and they can be computed from $\lambda$, and
$\lambda$ reconstructed from them, in polynomial time using the standard
abacus description.

If the $s$-core of $\lambda$ is nonempty, then $\lambda$ cannot be reduced to
the empty diagram by removing $m$ strips of length $s$, and the
Murnaghan-Nakayama rule gives $\chi^\lambda(s^m)=0$.

Suppose instead that the $s$-core is empty, and let
$d_i:=|\lambda^{(i)}|$, so that $\sum_{i=0}^{s-1}d_i=m$. Removing a
length-$s$ border strip from $\lambda$ corresponds to removing one corner
box from one component of the $s$-quotient. Let $f^\mu$ denote the number of
standard Young tableaux of shape $\mu$, which is the number of orders in
which the corner boxes of $\mu$ can be removed. The boxes of
$\lambda^{(i)}$ can be removed in $f^{\lambda^{(i)}}$ orders, and the
removal orders from the different quotient components can be interleaved in
$\binom{m}{d_0,\ldots,d_{s-1}}$ ways. The product of the
Murnaghan-Nakayama signs does not depend on the order in which the strips
are removed \cite[Section~2.7]{JamesKerber1981}. Hence,
\begin{equation}
  \chi^\lambda(s^m)
  =
  \begin{cases}
    \displaystyle
    \pm
    \binom{m}{d_0,\ldots,d_{s-1}}
    \prod_{i=0}^{s-1}f^{\lambda^{(i)}}\,,
    &
    \text{if the $s$-core of $\lambda$ is empty}\,,
    \\[2ex]
    0\,,
    &
    \text{otherwise}\,.
  \end{cases}
  \label{eq:equal-cycle-character-factorization}
\end{equation}
Since $z_{(s^m)}=s^mm!$, squaring gives
\begin{equation}
  P_{(s^m)}(\lambda)
  =
  \frac{1}{s^m}
  \binom{m}{d_0,\ldots,d_{s-1}}
  \prod_{i=0}^{s-1}
  \frac{\big(f^{\lambda^{(i)}}\big)^2}{d_i!}\,.
  \label{eq:equal-cycle-probability-factorization}
\end{equation}

The first two factors in
Eq.~\eqref{eq:equal-cycle-probability-factorization} form a multinomial
distribution: they give the probability that $m$ independent uniform choices
from $\{0,\ldots,s-1\}$ produce counts $(d_0,\ldots,d_{s-1})$. It remains
to sample a partition $\mu\vdash d$ with probability
\begin{equation}
  \frac{(f^\mu)^2}{d!}\,.
  \label{eq:plancherel-probability}
\end{equation}
The Robinson-Schensted correspondence is a bijection between permutations
of $d$ objects and pairs of standard Young tableaux of the same
shape~\cite[Section~7.11]{Stanley1999}. For a fixed shape $\mu$ there are
$(f^\mu)^2$ such pairs, so choosing a uniformly random permutation and
applying the Robinson--Schensted algorithm samples $\mu$ with probability
\eqref{eq:plancherel-probability}.

Eq.~\eqref{eq:equal-cycle-probability-factorization} now gives an exact
classical sampler:
\begin{enumerate}
  \item Choose $m$ independent uniform values from $\{0,\ldots,s-1\}$. Let
  $d_i$ be the number of times that $i$ occurs.

  \item For every $i$, choose a uniformly random permutation of $d_i$
  objects, apply the Robinson-Schensted algorithm, and let $\lambda^{(i)}$
  be the resulting shape.

  \item Reconstruct $\lambda$ from the empty $s$-core and the quotient
  $(\lambda^{(0)},\ldots,\lambda^{(s-1)})$.
\end{enumerate}
The probability that this algorithm outputs $\lambda$ is exactly
Eq.~\eqref{eq:equal-cycle-probability-factorization}.

Row insertion of a permutation of $d$ objects takes $O(d^2)$ steps, and
$\sum_id_i^2\le(\sum_id_i)^2\le n^2$, so the second step takes $O(n^2)$
steps. The uniform choices in the first step can be made exactly from random
bits in expected time $O(n\log n)$, and the abacus reconstruction in the
third step takes $O(n\log n)$ steps. The expected running time is therefore
$\widetilde O(n^2)$.

Finally, suppose that $\lambda$ is given. Its $s$-core and $s$-quotient can
be computed in polynomial time. A nonempty core gives probability zero.
Otherwise, each $f^{\lambda^{(i)}}$ follows from the hook-length
formula~\cite{FrameRobinsonThrall1954},
\begin{equation}
  f^\mu
  =
  \frac{|\mu|!}
  {\displaystyle\prod_{x\in\mu}h(x)}\,,
\end{equation}
where $h(x)$ is the hook length of the box $x$, and substitution into
Eq.~\eqref{eq:equal-cycle-probability-factorization} gives the exact
probability with polynomially many arithmetic operations on integers of
polynomially many bits.
\end{proof}

Theorem~\ref{thm:equal-cycle-classical} holds for every $s$ and $m$, in
particular for $m=\Theta(\sqrt n/\log n)$. Uniform cycle types therefore lie
inside the granularity regime \eqref{eq:intro-bghs-regime} while remaining
classically tractable, with every output probability efficiently computable
and an exact classical sampler for the distribution. They do not satisfy the
average-case hardness condition of Bravyi \emph{et al.} and cannot provide
the sampling advantage considered there.

\subsection{Almost-uniform cycles}
\label{sec:almost-equal-classical}

We now consider cycle types for which all but a small number of cycles have
the same length. Write
\begin{equation}
  \nu_j
  :=
  (s^m,r_1,\ldots,r_j)\,,
  \qquad
  n_j
  :=
  ms+\sum_{t=1}^{j}r_t\,,
  \qquad
  0\leq j\leq b\,,
  \label{eq:almost-equal-prefix}
\end{equation}
so that $\nu_0=(s^m)$ and $\nu_b$ is the cycle type we want to sample. The
order of $r_1,\ldots,r_b$ is arbitrary and fixed throughout the algorithm,
and the lengths $r_j$ need not be close to $s$. Here ``almost uniform''
means that only $b$ cycles have been separated from the uniform block. The decomposition of $\nu$ into a uniform block and exceptional cycles is not unique. Since Theorem~\ref{thm:almost-equal-classical} below holds for every decomposition, its bound applies with the smallest possible $b$, namely $b(\nu):=\ell(\nu)-\max_s m_s$, the number of cycles whose length is not the most frequent one. For example, the spiked cycle types $(n-c+1,1^{c-1})$ of Section~\ref{sec:numerics} have $b(\nu)=1$ for every $1<c<n$.

\begin{theorem}[Classical sampling for almost-uniform cycles]
\label{thm:almost-equal-classical}
For $\nu=(s^m,r_1,\ldots,r_b)\vdash n$, there is an exact classical
randomized algorithm that samples
\begin{equation}
  P_\nu(\lambda)
  =
  \frac{[\chi^\lambda(\nu)]^2}{z_\nu}
\end{equation}
in expected time $n^{4b+O(1)}$. Consequently, the sampler runs in polynomial
time when $b=O(1)$ and in quasipolynomial time when $b=O(\log n)$. For any
given $\lambda\vdash n$, the character $\chi^\lambda(\nu)$ and the probability
$P_\nu(\lambda)$ can be computed exactly in time $n^{2b+O(1)}$.
\end{theorem}

The algorithm adds the exceptional cycles one at a time. At step $j$, it
converts an exact sample from the distribution for $\nu_{j-1}$ into an exact
sample from the distribution for $\nu_j$. We first bound the number of
possible border-strip moves.

\begin{lemma}[Number of border-strip moves]
\label{lem:border-strip-count}
Let $r\ge1$ and let $\mu$ be a partition with $|\mu|+r\le n$. There are at
most $n^2$ partitions $\lambda$ such that $\lambda/\mu$ is a border strip of
size $r$. The same bound holds for the number of partitions $\mu$ obtained by
removing a border strip of size $r$ from a fixed partition $\lambda$ with
$|\lambda|\le n$. Moreover, all these partitions can be listed in time
polynomial in $n$.
\end{lemma}

\begin{proof}
Pad every partition with zero parts so that it has $n$ rows. A nonempty
border strip occupies a consecutive set of rows. Let $a$ and $c$ be its top
and bottom rows. If it occupies more than one row, connectedness and the
absence of a $2\times2$ block imply that its boxes in two consecutive rows
overlap in exactly one column. The boxes added to row $i$ form the interval
$\{\mu_i+1,\ldots,\lambda_i\}$, and since both $\mu$ and $\lambda$ are
partitions the one-column overlap condition gives
$\lambda_i=\mu_{i-1}+1$ for $a<i\leq c$. The total size $r$ then determines
the remaining part through
\begin{equation}
  \lambda_a-\mu_a
  =
  r-
  \sum_{i=a+1}^{c}
  \big(\mu_{i-1}+1-\mu_i\big)\,.
\end{equation}
All rows outside $a,\ldots,c$ are unchanged, so $(a,c)$ determines at most
one $\lambda$. For removal from a fixed $\lambda$, the same overlap
condition gives $\mu_i=\lambda_{i+1}-1$ for $a\leq i<c$, and the total size
determines $\mu_c$, so $(a,c)$ also determines at most one removable strip.
There are at most $n^2$ pairs. Enumerating the pairs, constructing the
corresponding candidate, and checking the partition and border-strip
conditions takes polynomial time.
\end{proof}

We want to choose each border-strip move with the same
probability, independent of the current partition. Set $K:=n^2$. For every partition and every strip length, use
Lemma~\ref{lem:border-strip-count} to list the valid additions and pad this
list with invalid labels until it has exactly $K$ entries, so that each valid
addition occurs exactly once.

Suppose inductively that $\mu\vdash n_{j-1}$ is sampled from
\begin{equation}
  P_{j-1}(\mu)
  =
  \frac{[\chi^\mu(\nu_{j-1})]^2}{z_{\nu_{j-1}}}\,.
  \label{eq:almost-equal-inductive-distribution}
\end{equation}
Choose one of the $K$ labels uniformly. If it is invalid, reject the proposal
and repeat the step with a fresh sample $\mu$. If it is valid, add the corresponding border strip of
size $r_j$ and call the resulting partition $\lambda$. For a proposed
$\lambda$, let
\begin{equation}
  \mathcal R_j(\lambda)
  :=
  \left\{
    \mu'\vdash n_{j-1}:
    \lambda/\mu'
    \text{ is a border strip of size }r_j
  \right\}\,,
  \qquad
  S_j(\lambda)
  :=
  \sum_{\mu'\in\mathcal R_j(\lambda)}
  [\chi^{\mu'}(\nu_{j-1})]^2\,.
  \label{eq:almost-equal-S}
\end{equation}
For every candidate that can actually be proposed, $S_j(\lambda)>0$, since
the sampled $\mu$ belongs to $\mathcal R_j(\lambda)$ and has
$\chi^\mu(\nu_{j-1})\neq0$ with probability one under
Eq.~\eqref{eq:almost-equal-inductive-distribution}. The
Murnaghan-Nakayama rule gives
\begin{equation}
  \chi^\lambda(\nu_j)
  =
  \sum_{\mu'\in\mathcal R_j(\lambda)}
  (-1)^{\operatorname{ht}(\lambda/\mu')}
  \chi^{\mu'}(\nu_{j-1})\,,
  \label{eq:almost-equal-MN}
\end{equation}
where $\operatorname{ht}(\lambda/\mu')$ is one less than the number of rows
occupied by the border strip~\cite[Theorem~7.17.3]{Stanley1999}. Accept the
proposed partition with probability
\begin{equation}
  A_j(\lambda)
  :=
  \frac{[\chi^\lambda(\nu_j)]^2}
       {K S_j(\lambda)}\,.
  \label{eq:almost-equal-acceptance}
\end{equation}
If the proposal is rejected, sample a new $\mu$ from
Eq.~\eqref{eq:almost-equal-inductive-distribution} and repeat the entire
step.

\begin{lemma}[One-cycle update]
\label{lem:almost-equal-update}
The acceptance probability \eqref{eq:almost-equal-acceptance} lies in
$[0,1]$. Conditioned on acceptance, the output has distribution
\begin{equation}
  P_j(\lambda)
  =
  \frac{[\chi^\lambda(\nu_j)]^2}{z_{\nu_j}}\,.
  \label{eq:almost-equal-updated-distribution}
\end{equation}
\end{lemma}

\begin{proof}
Lemma~\ref{lem:border-strip-count} gives $|\mathcal R_j(\lambda)|\leq K$,
and Cauchy-Schwarz applied to Eq.~\eqref{eq:almost-equal-MN} gives
\begin{equation}
  [\chi^\lambda(\nu_j)]^2
  \leq
  |\mathcal R_j(\lambda)|
  \sum_{\mu'\in\mathcal R_j(\lambda)}
  [\chi^{\mu'}(\nu_{j-1})]^2
  \leq
  K S_j(\lambda)\,,
  \label{eq:almost-equal-CS}
\end{equation}
so $0\leq A_j(\lambda)\leq1$.

Before the acceptance test, the probability of proposing $\lambda$ is
\begin{equation}
  q_j(\lambda)
  =
  \frac{S_j(\lambda)}
       {Kz_{\nu_{j-1}}}\,,
  \label{eq:almost-equal-proposal}
\end{equation}
since each $\mu'\in\mathcal R_j(\lambda)$ is sampled with probability
$[\chi^{\mu'}(\nu_{j-1})]^2/z_{\nu_{j-1}}$ and its unique label leading to
$\lambda$ is chosen with probability $1/K$. It follows that
\begin{equation}
  q_j(\lambda)A_j(\lambda)
  =
  \frac{[\chi^\lambda(\nu_j)]^2}
       {K^2z_{\nu_{j-1}}}\,.
  \label{eq:almost-equal-accepted-mass}
\end{equation}
Column orthogonality gives $\sum_{\lambda\vdash n_j}[\chi^\lambda(\nu_j)]^2
=z_{\nu_j}$, so the total probability of acceptance in one trial is
\begin{equation}
  \frac{z_{\nu_j}}
       {K^2z_{\nu_{j-1}}}\,.
  \label{eq:almost-equal-total-acceptance}
\end{equation}
Dividing Eq.~\eqref{eq:almost-equal-accepted-mass} by
Eq.~\eqref{eq:almost-equal-total-acceptance} gives
Eq.~\eqref{eq:almost-equal-updated-distribution}.
\end{proof}

\begin{proof}[Proof of Theorem~\ref{thm:almost-equal-classical}]
Theorem~\ref{thm:equal-cycle-classical} supplies the exact sampler for
$P_0$. Apply Lemma~\ref{lem:almost-equal-update} successively for
$j=1,\ldots,b$. Induction shows that the final output has distribution
$P_b(\lambda)=[\chi^\lambda(\nu_b)]^2/z_{\nu_b}$.

It remains to bound the runtime. Write $\operatorname{poly}(n)$ for a fixed
polynomial in $n$, independent of $b$, that bounds every elementary step
below: listing the at most $K=n^2$ border-strip moves of
Lemma~\ref{lem:border-strip-count}, computing an $s$-core and $s$-quotient,
evaluating the hook-length formula, and arithmetic on integers with
$O(n\log n)$ bits.

Let $C_j$ be the time to compute $\chi^\lambda(\nu_j)$ exactly for one
partition $\lambda\vdash n_j$. For $j=0$ the value is the signed
uniform-cycle character. Starting from the Young diagram of $\lambda$, list
all border strips of length $s$ whose removal leaves a valid Young diagram,
and remove any one of them. Repeat until the diagram is empty. When the
$s$-core is empty, such a removal is available at every step, and the product
of the Murnaghan-Nakayama signs does not depend on which valid strips are
chosen \cite[Section~2.7]{JamesKerber1981}. There are $m\leq n$ removals,
each listed in polynomial time by Lemma~\ref{lem:border-strip-count}, and the
magnitude is given by Eq.~\eqref{eq:equal-cycle-character-factorization}.
Hence $C_0=\operatorname{poly}(n)$. For $j\geq1$, the recursion
\eqref{eq:almost-equal-MN} has at most $K$ terms, so
\begin{equation}
  C_j
  \leq
  K\,C_{j-1}+\operatorname{poly}(n)\,,
  \qquad\text{hence}\qquad
  C_j
  \leq
  (j+1)\,K^{j}\operatorname{poly}(n)
  =
  n^{2j+O(1)}\,.
  \label{eq:almost-equal-eval-cost}
\end{equation}
The centralizer $z_{\nu_j}$ is given explicitly by the cycle multiplicities,
so $P_{\nu_j}(\lambda)$ is computable within the same bound. Taking $j=b$
proves the evaluation claim.

A trial at step $j$ draws one sample from $P_{j-1}$, chooses one of the $K$
labels, and evaluates the acceptance probability
\eqref{eq:almost-equal-acceptance}. The latter requires
$\chi^\lambda(\nu_j)$, at cost $C_j$, and $S_j(\lambda)$, which is a sum of
at most $K$ characters at level $j-1$, at cost $K\,C_{j-1}\leq C_j$. All
quantities in the acceptance test are integers with $O(n\log n)$ bits, up to
$O(b\log n)$ further bits from the intermediate sums, so the rational
probability, and the uniform choice among the $K$ labels, can be sampled
exactly from random bits in expected time polynomial in this bit length.
Writing $T_j$ for the expected time to produce one sample from $P_j$, a
trial at step $j$ therefore costs at most $T_{j-1}+2C_j+\operatorname{poly}(n)$.

Let $c_j$ be the multiplicity of the part $r_j$ in $\nu_{j-1}$. From
$z_\nu=\prod_{r\geq1}r^{m_r}m_r!$,
\begin{equation}
  \frac{z_{\nu_j}}{z_{\nu_{j-1}}}
  =
  r_j(c_j+1)
  \geq 1\,,
  \label{eq:almost-equal-z-ratio}
\end{equation}
so by Eq.~\eqref{eq:almost-equal-total-acceptance} a trial at step $j$ is
accepted with probability $r_j(c_j+1)/K^2\geq K^{-2}$. The number of trials
is geometric with this success probability, and its expectation is at most
$K^2$. Since the trials are independent of one another, Wald's identity
gives
\begin{equation}
  T_j
  \leq
  K^2\big(T_{j-1}+2C_j+\operatorname{poly}(n)\big)\,,
  \qquad
  T_0=\operatorname{poly}(n)\,.
  \label{eq:almost-equal-recursion}
\end{equation}
Substituting Eq.~\eqref{eq:almost-equal-eval-cost} and iterating gives
\begin{align}
  T_b
  &\leq
  K^{2b}\,T_0
  +
  \sum_{j=1}^{b}
  K^{2(b-j+1)}\big(2C_j+\operatorname{poly}(n)\big)
  \nonumber\\
  &\leq
  K^{2b}\operatorname{poly}(n)
  +
  \sum_{j=1}^{b}
  K^{2b-2j+2}\,(j+1)\,K^{j}\operatorname{poly}(n)
  \leq
  (b+1)^2\,K^{2b+1}\operatorname{poly}(n)\,,
\end{align}
where the last step uses the fact that the exponent $2b-j+2$ is largest at $j=1$.
Since $K=n^2$ and $b\leq n$, this is $T_b=n^{4b+O(1)}$.
\end{proof}

Theorem~\ref{thm:almost-equal-classical} shows that a constant number of exceptional cycles does not produce a classically hard sampling problem. Any candidate family for quantum advantage built from perturbations of uniform cycles must, therefore, have a number of exceptional cycles that grows with $n$.

\section{Discussion and Outlook}
We have given a Hamiltonian-based algorithm for character state preparation, and extended it coherently to realize the quantum character transform over $S_n$. 
Our construction replaces the existing quantum Fourier transform-based approach~\cite{BravyiGossetHavlicekSchatzki2025} by a sequence of transfer evolutions under the Hamiltonians $H_s$, each run for an exactly known time.
Simulated using QSVT, the resulting circuit uses $\widetilde O(n^{5/2} \log \delta^{-1})$ gates and $O(n)$ qubits for CSP at any cycle type and for the full QCT. 
This is below the $\widetilde O(n^3)$ gates and $\widetilde O(n^{3/2})$ qubits of the best known implementation~\cite{Bruinsma2026} of the quantum Fourier transform over $S_n$.

Towards near-term hardware-friendly realizations, we also consider implementing the Hamiltonian evolution subroutines of our algorithm by Trotter decomposition. Using fast fermionic permutations with reconfigurable qubits~\cite{Maskara_2025}, the Hamiltonian evolution can be implemented with geometrically local gates in 1D or 2D geometry.
While the general operator-norm Trotter error bound indicates a gate complexity of $\widetilde O(n^4/\delta)$, we improved the bound by focusing on the states that our algorithm actually visits.
Looking at the leading-order Trotter error, our CSP algorithm has a sufficient gate complexity of $\widetilde O(n^{7/4}/\delta)$ for uniform cycles with the $\Theta(\sqrt n/\log n)$ cycle count that the hardness argument of Bravyi \emph{et al.}\cite{BravyiGossetHavlicekSchatzki2025} requires for the granularity condition. 
To better understand classical hardness, we showed that such uniform cycle types, as well as their perturbations by a constant number of exceptional cycles, admit exact polynomial-time classical samplers despite satisfying the granularity condition.
Nevertheless, when the number of exceptional cycles grows with $n$, or for generic cycle types with $\Theta(\sqrt n/\log n)$ cycles, no efficient classical algorithms are known, and quantum advantage is still possible.
Our numerical simulations up to $n=48$ suggest that the Trotter method remains efficient for these general cycle types, and generally have far better resource scaling than what our bounds predict.
These results motivate further study of the Trotter-based CSP, which we put forth as a candidate quantum advantage task on near-term reconfigurable-qubit quantum hardware.

Several interesting questions remain open:
\begin{itemize}
    \item Can the full Trotter error on the states visited by the algorithm
    be bounded, beyond the leading-order term? Our state-dependent bound
    controls only the error at order in $1/M$, where $M$ is the number of Trotter steps.
    A bound on the full error would turn our $\widetilde O(n^{7/4}/\delta)$
    estimate into a rigorous complexity bound. It would also be valuable to
    extend the analysis to cycle types with a growing number of exceptional
    cycles, where known classical algorithms become inefficient. For these
    cycle types, our numerics show smaller costs than both the worst-case
    bounds and the estimated bounds from the state-dependent error analysis, suggesting there is much room for improvements.
    
    \item Can our Hamiltonian transfer mechanism also prepare the Kostka and
    Littlewood-Richardson states described in Crichigno and
    Prakash~\cite{CrichignoPrakash2024}? These quantum spin-chain states encode
    Kostka numbers and Littlewood--Richardson coefficients as amplitudes, with
    the currents replaced by noncommutative complete symmetric polynomials and
    noncommutative Schur polynomials. Bravyi \emph{et al.}\cite{BravyiGossetHavlicekSchatzki2025} turned the complete symmetric operators into a classical MPS algorithm for Kostka numbers and left the Littlewood-Richardson case open. Neither work gives a quantum algorithm for preparing these states. Larocca and Havl\'i\v{c}ek~\cite{LaroccaHavlicek2025} compute Kostka, Littlewood-Richardson, plethysm and Kronecker coefficients with QFT-based quantum algorithms whenever the ratio of the relevant representation dimensions is polynomial. In this regime, they give an efficient classical algorithm for Kostka numbers and conjecture one for Littlewood-Richardson coefficients. For Kronecker coefficients, a classical algorithm of Panova~\cite{Panova2025} reduces the conjectured advantage to a polynomial gap. Exact Hamiltonian transfers for these nonunitary operators would give a QFT-free route to the same quantities.
    
    \item Can we generalize our CSP and QCT algorithms to other
    finite groups? In particular, QCT for general finite groups has applications in symmetry-adapted quantum simulations~\cite{Bastidas_2025}. The transfer mechanism (Lemma~\ref{lem:transfer}) can be applied to any family of linear operators with the ladder property of
    Proposition~\ref{prop:ladder}, so groups whose characters obey a
    Murnaghan-Nakayama-type rule, such as the wreath products
   $\mathbb Z_k\wr S_n$, are natural candidates.

	\item Can our Hamiltonian-based approach for the
    QCT be extended to realize the full quantum Fourier transform over $S_n$? The QCT is the restriction of the QFT to class functions, with the representation-matrix
    indices traced out. It would be useful to investigate whether the Hamiltonian approach can
    reach the remaining $O(n\log n)$ qubits of the Fourier basis at a cost
    below the $\widetilde O(n^3)$ gates in the state-of-the-art QFT circuits~\cite{Bruinsma2026}.
    
\end{itemize}

\section*{Acknowledgments}
We thank Jason Cong, David Gosset, and Yuan Liu for helpful discussions. 
DR is supported by the Mani L. Bhaumik Institute for Theoretical Physics. 
LZ is supported by the Google Research Scholar Award and Department of Navy Award No. N000142612399 issued by the Office of Naval Research.

\subsection*{Use of generative AI}
OpenAI GPT-5.4 Pro, accessed through ChatGPT in April 2026, was used to
explore approaches to the state-dependent Trotter error bound
(Section~\ref{sec:state-dependent} and Appendix~\ref{app:state-dependent}).
Each step of the final proof was derived and verified independently by the
authors. Anthropic Claude Opus 4.8, accessed through Claude Code in June 2026,
was used to check the QSVT-based implementation of the quantum character
transform (Section~\ref{sec:qct-cost}). Anthropic Claude Fable 5, accessed
through Claude Code in August 2026, was used to explore the classical sampling
algorithms for uniform and almost-uniform cycle types
(Section~\ref{sec:classical-hardness}) and to code the numerical simulations
of the Trotterized CSP circuit for uniform cycle types and their perturbations
(Section~\ref{sec:numerics}). Each step of these algorithms and their proofs
was verified independently by the authors, and the simulation code was checked
against a brute-force simulation on the full Fock space for $n=4$ and $n=5$.


\begin{appendices}

\section{Proof of the eigenvalue condition}
\label{app:eigenvalue-proof}

This appendix proves Proposition~\ref{prop:ladder} and the normalization in
Eq.~\eqref{eq:ladder-states}. The proof has three steps. First, we introduce
a displacement operator, which measures how far the fermions have hopped
relative to the seed state $\seed=\ket{1^n0^n}$, and show that each intermediate state of the algorithm has a
fixed displacement (Lemma~\ref{lem:displacement}). Second, we show that a
state of small displacement has its lowest $s$ sites occupied and its
highest $s$ sites empty (Lemma~\ref{lem:frozen}). Third, on such states the
commutators of the current operators take a simple form
(Lemma~\ref{lem:comm-frozen}), and the ladder property follows by induction.

Define the displacement operator
\begin{equation}
  \widehat{\mathsf D}
  :=
  \sum_{j=1}^{2n}j\hat n_j
  -\frac{n(n+1)}{2}\,\mathbb{1}\,.
\end{equation}
For every $s\in\{1,\dots,n\}$,
\begin{equation}
  [\widehat{\mathsf D},J_s]=sJ_s\,,
  \qquad
  [\widehat{\mathsf D},J_s^\dagger]=-sJ_s^\dagger\,.
  \label{eq:displacement-comm}
\end{equation}
Write $\ket{\phi_k} := J_{\nu_k}\cdots J_{\nu_1}\seed$ for the (unnormalized)
state after depositing the first $k$ currents, so $\ket{\phi_0}=\seed$.

\begin{lemma}[Displacement of intermediate states]
\label{lem:displacement}
Let $S_k=\sum_{j=1}^k\nu_j$. Then
\begin{equation}
  \widehat{\mathsf D}|\phi_k\rangle
  =S_k|\phi_k\rangle\,.
\end{equation}
Equivalently, every occupation-basis state in the support of $|\phi_k\rangle$ has total displacement $S_k$ from the seed configuration.
\end{lemma}

\begin{proof}
The claim follows by induction using Eq.~\eqref{eq:displacement-comm}. The base
case is $\widehat{\mathsf D}\seed=0$. If the claim holds for $k-1$, then
\begin{align}
\widehat{\mathsf D}J_{\nu_k}|\phi_{k-1}\rangle
&=
(J_{\nu_k}\widehat{\mathsf D}+\nu_kJ_{\nu_k})|\phi_{k-1}\rangle
\nonumber \\
&=(S_{k-1}+\nu_k)|\phi_k\rangle\,.
\end{align}
\end{proof}

For an occupation-basis state $\ket{p_1<\cdots<p_n}$,
\begin{equation}
  \widehat{\mathsf D}\ket{p_1<\cdots<p_n}
  =
  \left(
    \sum_{r=1}^{n}(p_r-r)
  \right)
  \ket{p_1<\cdots<p_n}\,.
\end{equation}
Thus, $\widehat{\mathsf D}$ measures how far the occupied sites have moved
from their positions $1,\ldots,n$ in the seed state. Define
$d_r:=p_r-r$. Since $p_{r+1}\geq p_r+1$, the sequence
$d_1\leq\cdots\leq d_n$ is nondecreasing.

\begin{definition}\label{def:s-state}
An $n$-particle state $|\psi\rangle$ on $2n$ sites is called an \emph{$s$-state} if every
occupation-basis state in its support has sites $1,\dots,s$ occupied and sites
$2n-s+1,\dots,2n$ empty.
\end{definition}

\begin{lemma}
\label{lem:frozen}
If every occupation-basis state in the support of $\ket{\psi}$ satisfies
\begin{equation}
  \sum_{r=1}^{n}d_r\leq n-s\,,
\end{equation}
then $\ket{\psi}$ is an $s$-state.
\end{lemma}

\begin{proof}
Suppose some site $j\leq s$ is empty. Then $p_j\geq j+1$, so
$d_j\geq1$. Since the $d_r$ are nondecreasing, $d_r\geq1$ for every
$r\geq j$. Therefore,
\begin{equation}
  \sum_{r=1}^{n}d_r
  \geq
  \sum_{r=j}^{n}d_r
  \geq
  n-j+1
  \geq
  n-s+1\,,
\end{equation}
which contradicts the assumption.

Thus, sites $1,\ldots,s$ are occupied. Moreover, $d_r\geq0$ for every
$r$, so
\begin{equation}
  d_n
  \leq
  \sum_{r=1}^{n}d_r
  \leq
  n-s\,.
\end{equation}
It follows that
\begin{equation}
  p_n=n+d_n\leq2n-s\,,
\end{equation}
so sites $2n-s+1,\ldots,2n$ are empty. Hence, $\ket{\psi}$ is an
$s$-state.
\end{proof}

\begin{lemma}
\label{lem:prefix-is-s-state}
For $s=\nu_k$, the state $|\phi_{k-1}\rangle$ is an $s$-state.
\end{lemma}

\begin{proof}
By Lemma~\ref{lem:displacement}, $|\phi_{k-1}\rangle$ has displacement $S_{k-1}$. Since $\nu$ is nonincreasing,
\begin{align}
  S_{k-1}+s=S_k\leq n\,.
\end{align}
Thus $S_{k-1}\leq n-s$, and Lemma~\ref{lem:frozen} applies.
\end{proof}

\begin{lemma}
\label{lem:comm-same}
For $s\in\{1,\dots,n\}$,
\begin{equation}
  [J_s,J_s^\dagger]
  =
  \sum_{j=s+1}^{2n}\hat n_j
  -
  \sum_{j=1}^{2n-s}\hat n_j\,.
  \label{eq:same-comm}
\end{equation}
\end{lemma}

\begin{proof}
Using the canonical anticommutation relations,
\begin{align}
[J_s,J_s^\dagger]
&=
\sum_{i,j=1}^{2n-s}
[a_{i+s}^\dagger a_i,a_j^\dagger a_{j+s}] =
\sum_{j=1}^{2n-s}
(a_{j+s}^\dagger a_{j+s}-a_j^\dagger a_j)\,,
\end{align}
which is Eq.~\eqref{eq:same-comm}.
\end{proof}

\begin{lemma}
\label{lem:comm-diff}
For $u>s$,
\begin{equation}
  [J_u,J_s^\dagger]
  =
  \sum_{j=s+1}^{2n-u+s}a_{j+u-s}^\dagger a_j
  -
  \sum_{i=1}^{2n-u}a_{i+u-s}^\dagger a_i\,.
  \label{eq:diff-comm}
\end{equation}
\end{lemma}

\begin{proof}
Again, using the canonical anticommutation relations,
\begin{align}
[J_u,J_s^\dagger]
&=
\sum_{i=1}^{2n-u}\sum_{j=1}^{2n-s}
[a_{i+u}^\dagger a_i,a_j^\dagger a_{j+s}] =
\sum_{i=1}^{2n-u}a_{i+u}^\dagger a_{i+s}
-
\sum_{\substack{i=1\\ i+u-s\leq 2n-s}}^{2n-u}a_{i+u-s}^\dagger a_i\,,
\end{align}
which gives Eq.~\eqref{eq:diff-comm} after reindexing.
\end{proof}

\begin{lemma}
\label{lem:comm-frozen}
The following hold.
\begin{enumerate}[label=(\roman*)]
\item If $|\psi\rangle$ is an $s$-state, then
\begin{align}
  [J_s^\dagger,J_s]|\psi\rangle=s|\psi\rangle\,.
\end{align}
\item If $u>s$ and $|\psi\rangle$ is a $u$-state, then
\begin{align}
  [J_u,J_s^\dagger]|\psi\rangle=0\,.
\end{align}
\end{enumerate}
\end{lemma}

\begin{proof}
For (i), Eq.~\eqref{eq:same-comm} gives
\begin{align}
[J_s^\dagger,J_s]
&=
\sum_{j=1}^{2n-s}\hat n_j-\sum_{j=s+1}^{2n}\hat n_j \nonumber \\
&=
\sum_{j=1}^{s}\hat n_j-\sum_{j=2n-s+1}^{2n}\hat n_j.
\end{align}
On an $s$-state this equals $s$.

For (ii), rewrite Eq.~\eqref{eq:diff-comm} as
\begin{align}
[J_u,J_s^\dagger]
&=
\sum_{j=2n-u+1}^{2n-u+s}a_{j+u-s}^\dagger a_j
-
\sum_{i=1}^{s}a_{i+u-s}^\dagger a_i\,.
\end{align}
On a $u$-state, the first sum annihilates the state because the annihilation
operators act on empty sites in $\{2n-u+1,\dots,2n\}$. The second sum
annihilates the state because $i+u-s\in\{1,\dots,u\}$ is already occupied, so
the creation operator vanishes.
\end{proof}

\begin{proof}[Proof of Proposition~\ref{prop:ladder}]
Fix $k$ and set $s=\nu_k$. Let $r_k$ be the number of earlier factors equal to
$s$. Since $\nu$ is nonincreasing, the
factors equal to $s$ among $\nu_1,\dots,\nu_{k-1}$ occur after all parts larger
than $s$, so
\begin{align}
  J_{\nu_{k-1}}\cdots J_{\nu_1}=J_s^{r_k}B\,,
\end{align}
where $B=J_{u_t}\cdots J_{u_1}$ with every $u_j>s$.
In the notation of Section~\ref{sec:ladder},
\begin{equation}
  \ket{\Phi^{>s}}
  =
  \frac{B\seed}{\|B\seed\|}\,.
\end{equation}

First, $J_s^\dagger B\seed=0$. Commute $J_s^\dagger$ through $B$ toward the
seed. Write $B_{<j}=J_{u_{j-1}}\cdots J_{u_1}$ for the partial
product. By Lemma~\ref{lem:displacement}, $B_{<j}\seed$ has displacement
$\sum_{i<j}u_i$. Since $\ket{\phi_{k-1}}$ deposits total length
$S_{k-1}\le n-s< n$ and the $u_i$ are among its factors,
$\sum_{i\le j}u_i\le n$, hence $\sum_{i<j}u_i\le n-u_j$. By
Lemma~\ref{lem:frozen}, $B_{<j}\seed$ is therefore a $u_j$-state, and
Lemma~\ref{lem:comm-frozen}(ii) gives $[J_{u_j},J_s^\dagger]=0$ on it. Passing
$J_s^\dagger$ through every factor of $B$ leaves
$J_s^\dagger B\seed=B J_s^\dagger\seed=0$, because every term in
$J_s^\dagger=\sum_i a_i^\dagger a_{i+s}$ annihilates $\seed$.

Second, the states $J_s^qB\seed$ are $s$-states for $0\leq q\leq r_k$, because
their displacement is at most $\sum_{j=1}^{k-1}\nu_j\leq n-s$, so
Lemma~\ref{lem:frozen} applies.

Let $|\psi_q\rangle=J_s^qB\seed$. Since $J_s^\dagger B\seed=0$,
\begin{align}
  J_s^\dagger|\psi_{q+1}\rangle
  &=
  [J_s^\dagger,J_s^{q+1}]B\seed \nonumber \\
  &=
  \sum_{r=0}^{q}J_s^r[J_s^\dagger,J_s]J_s^{q-r}B\seed\,.
\end{align}
Each state $J_s^{q-r}B\seed$ is an $s$-state, so
Lemma~\ref{lem:comm-frozen}(i) gives
\begin{align}
  J_s^\dagger|\psi_{q+1}\rangle=(q+1)s|\psi_q\rangle\,.
\end{align}
Since $|\psi_{q+1}\rangle=J_s|\psi_q\rangle$, this also gives
\begin{equation}
  \|\psi_{q+1}\|^2
  =
  (q+1)s\,\|\psi_q\|^2\,.
\end{equation}
Thus
\begin{equation}
  \|J_s^qB\seed\|
  =
  \sqrt{s^q q!}\,\|B\seed\|\,,
\end{equation}
for every $0\leq q\leq r_k+1$. Taking $k$ to be the last occurrence of $s$
proves the normalization in Eq.~\eqref{eq:ladder-states} for all
$0\leq q\leq m_s$.
Taking $q=r_k$ gives
\begin{align}
  J_s^\dagger J_s|\phi_{k-1}\rangle=(r_k+1)s|\phi_{k-1}\rangle\,,
\end{align}
and the norm formula gives
\begin{equation}
  \ket{\Phi_{s,r_k}}
  =
  \frac{\ket{\phi_{k-1}}}{\|\phi_{k-1}\|}\,.
\end{equation}
This is Eq.~\eqref{eq:ladder-eigen} with $j=r_k+1$. Since $k$ can be any
occurrence of the part $s$, it proves the claim for every
$1\leq j\leq m_s$. The two relations in Eq.~\eqref{eq:ladder-raiselower}
then follow from Eq.~\eqref{eq:ladder-states} and the lowering relation above.
\end{proof}

\section{Trotter error: commutator calculations}\label{app:trotter}

\subsection{Commutator norms}\label{app:commutators}

Write $h_{q,a} = \sigma_a^+\otimes A_q + \sigma_a^-\otimes A_q^\dagger$ where
$A_q = \sigma_{q+1}^+ \sigma_q^-$. Using $\sigma_a^+\sigma_a^-=|1\rangle\!\langle 1|_a$,
$\sigma_a^-\sigma_a^+=|0\rangle\!\langle 0|_a$ and $(\sigma_a^\pm)^2 = 0$,
\begin{align}
    [h_{q,a}, h_{q',a}] = |1\rangle\!\langle 1|_a\otimes(A_qA_{q'}^\dagger - A_{q'}A_q^\dagger)
    + |0\rangle\!\langle 0|_a\otimes(A_q^\dagger A_{q'} - A_{q'}^\dagger A_q)\,.
\end{align}

\paragraph{Adjacent bonds ($|q-q'|=1$).}
Take $q' = q+1$. Then
\begin{align}
    A_q A_{q+1}^\dagger = (\sigma_{q+1}^+\sigma_q^-)(\sigma_{q+1}^+\sigma_{q+2}^-)
    = (\sigma_{q+1}^+)^2 \sigma_q^- \sigma_{q+2}^- = 0\,,
\end{align}
and the remaining products vanish the same way, so $[h_{q,a}, h_{q+1,a}] = 0$.

\paragraph{Non-adjacent bonds ($|q-q'| \geq 2$).}
The four sites $\{q, q+1, q', q'+1\}$ are distinct. This includes bonds in
different chains, which act on disjoint site pairs and are treated by the same
argument. The operators act as
\begin{align}
    A_qA_{q'}^\dagger &: \quad |1_q\,0_{q+1}\,0_{q'}\,1_{q'+1}\rangle \to
    |0_q\,1_{q+1}\,1_{q'}\,0_{q'+1}\rangle\,, \nonumber \\
    A_{q'}A_q^\dagger &: \quad |0_q\,1_{q+1}\,1_{q'}\,0_{q'+1}\rangle \to
    |1_q\,0_{q+1}\,0_{q'}\,1_{q'+1}\rangle\,,
\end{align}
with all other basis states annihilated. In the two-dimensional subspace
spanned by $|1_q\,0_{q+1}\,0_{q'}\,1_{q'+1}\rangle$ and
$|0_q\,1_{q+1}\,1_{q'}\,0_{q'+1}\rangle$,
the difference $A_qA_{q'}^\dagger - A_{q'}A_q^\dagger$ is represented by
$\bigl(\begin{smallmatrix} 0 & -1 \\ 1 & 0 \end{smallmatrix}\bigr)$, which has
eigenvalues $\pm i$ and norm $1$. The same holds for $A_q^\dagger A_{q'} -
A_{q'}^\dagger A_q$. Since the commutator is block-diagonal in the ancilla basis,
\begin{align}
    \|[h_{q,a}, h_{q',a}]\| = 1 \qquad \text{for } |q-q'| \geq 2\,.
\end{align}

\subsection{Bond counting}\label{app:bond-counting}

Fix the length $s = \nu_k$ being deposited, and let $2n - s$ be the total number
of bonds in $H_s'$. In a chain of length $L_\mu$ there are $L_\mu - 2$ adjacent
bond pairs, each with vanishing commutator. Summing over the $s$ chains gives
$\sum_\mu (L_\mu - 2) = 2n - 2s$ commuting pairs. Every remaining pair has
commutator norm $1$, so
\begin{align}
    \sum_{q<q'} \|[h_{q,a}, h_{q',a}]\|
    = \binom{2n-s}{2} - (2n - 2s) \equiv C(\nu_k, n)\,.
\end{align}

\section{Leading-order state-dependent Trotter error for uniform cycle partitions}
\label{app:state-dependent}

This appendix derives the state-dependent bound of
Section~\ref{sec:state-dependent}. We first reduce the leading-order Trotter error
of one transfer to the norms $\|\mathcal A_s\ket{\rho_r}\|$. We then write
$\ket{\rho_r}$ explicitly, using the multinomial structure of uniform cycles,
and finally evaluate the contributions of commutators between different
chains and within the same chain.

We specialize to a uniform-cycle partition $\nu=(s^m)$ with $n=ms$. A single
permutation $P_s$ is used at every layer. In the permuted basis, there are $s$
chains, each of length $2m$. Throughout this appendix, $\seed$ denotes the seed
in the following basis
\begin{equation}
  \seed
  =
  \bigotimes_{\alpha=1}^{s}\ket{1^m0^m}_{\alpha}\,.
\end{equation}
For chain $\alpha\in\{1,\ldots,s\}$ and bond
$q\in\{1,\ldots,2m-1\}$, define
\begin{equation}
  A_{\alpha,q}
  :=
  \sigma^+_{\alpha,q+1}\sigma^-_{\alpha,q}\,.
\end{equation}
The current on chain $\alpha$ and the total length-$s$ current are
\begin{equation}
  J_\alpha
  :=
  \sum_{q=1}^{2m-1}A_{\alpha,q}\,,
  \qquad
  J_s
  =
  \sum_{\alpha=1}^{s}J_\alpha\,.
\end{equation}

\subsection{Telescoped state-dependent bound}

The input to the $k$-th transfer is
\begin{equation}
  \ket{\Psi_{k-1}}
  =
  \ket{0}_a\ket{\rho_{k-1}}\,,
\end{equation}
where $\ket{\rho_r}$ is the normalized system state after $r$ transfers of
the length-$s$ current. Write
\begin{equation}
  U_k
  :=
  e^{-it_kH_s'}
\end{equation}
for the exact evolution in the $k$-th transfer.

Order the bond terms lexicographically by their chain and bond labels
$(\alpha,q)$, and write them as $h_1,\ldots,h_{2n-s}$. One first-order
Trotter step is
\begin{equation}
  S_s(\tau)
  :=
  e^{-i\tau h_1}\cdots e^{-i\tau h_{2n-s}}\,.
\end{equation}
For
\begin{equation}
  \tau_k:=\frac{t_k}{M_k}\,,
\end{equation}
the Trotter approximation is
\begin{equation}
  \widetilde U_k
  :=
  S_s(\tau_k)^{M_k}\,.
\end{equation}
The $iX_a$ correction applied after the evolution is omitted here because it
is applied to both the exact and Trotterized states and therefore does not
change their distance.

Expanding one Trotter step gives
\begin{equation}
  S_s(\tau)-e^{-i\tau H_s'}
  =
  -\frac{\tau^2}{2}\mathcal E_s
  +O(\tau^3)\,,
  \qquad
  \mathcal E_s
  :=
  \sum_{w<w'}[h_w,h_{w'}]\,.
  \label{eq:app-step}
\end{equation}
Telescoping over the $M_k$ steps gives
\begin{align}
  \widetilde U_k-U_k
  &=
  \sum_{j=0}^{M_k-1}
  S_s(\tau_k)^{M_k-1-j}
  \left(
    S_s(\tau_k)-e^{-i\tau_kH_s'}
  \right)
  e^{-ij\tau_kH_s'}\,.
  \label{eq:app-telescope}
\end{align}
Define the state error of this transfer by
\begin{equation}
  \varepsilon_k^{\mathrm{SD}}
  :=
  \left\|
    \left(\widetilde U_k-U_k\right)\ket{\Psi_{k-1}}
  \right\|\,.
\end{equation}
Using the unitarity of $S_s(\tau_k)$ and $e^{-i\tau_kH_s'}$ gives
\begin{align}
  \varepsilon_k^{\mathrm{SD}}
  &\leq
  \frac{\tau_k^2}{2}
  \sum_{j=0}^{M_k-1}
  \left\|
    \mathcal E_s
    e^{-ij\tau_kH_s'}
    \ket{\Psi_{k-1}}
  \right\|
  +O(M_k\tau_k^3)\,.
  \label{eq:app-state-error}
\end{align}
The estimates below control the above leading term.

For two non-adjacent bonds
$w=(\alpha,q)$ and $w'=(\beta,q')$, define
\begin{equation}
  \mathcal A_{w w'}
  :=
  A_w^\dagger A_{w'}
  -
  A_{w'}^\dagger A_w\,.
\end{equation}
Adjacent bonds in the same chain commute. Any two non-adjacent bonds act on
disjoint system qubits, and
\begin{equation}
  [h_w,h_{w'}]
  =
  Z_a\otimes\mathcal A_{w w'}\,.
  \label{eq:app-comm}
\end{equation}
Set
\begin{equation}
  \mathcal A_s
  :=
  \sum_{\substack{w<w'\\
  w,w'\ \mathrm{non\text{-}adjacent}}}
  \mathcal A_{ww'}\,.
\end{equation}
It follows that
\begin{equation}
  \mathcal E_s
  =
  Z_a\otimes\mathcal A_s\,.
\end{equation}

The exact evolution preserves the two-dimensional space
\begin{equation}
  \operatorname{span}
  \left\{
    \ket{0}_a\ket{\rho_{k-1}},
    \ket{1}_a\ket{\rho_k}
  \right\}\,.
\end{equation}
On the input state,
\begin{align}
  e^{-i\theta H_s'}\ket{\Psi_{k-1}}
  &=
  \cos\!\left(\sqrt{ks}\,\theta\right)
  \ket{0}_a\ket{\rho_{k-1}}
  \nonumber\\
  &\quad
  -
  i\sin\!\left(\sqrt{ks}\,\theta\right)
  \ket{1}_a\ket{\rho_k}\,.
  \label{eq:app-rabi}
\end{align}
Since
\begin{equation}
  Z_a\ket{0}_a=\ket{0}_a\,,
  \qquad
  Z_a\ket{1}_a=-\ket{1}_a\,,
\end{equation}
the two terms obtained after applying $\mathcal E_s$ have orthogonal ancilla
states. Therefore,
\begin{align}
  \left\|
    \mathcal E_s
    e^{-i\theta H_s'}
    \ket{\Psi_{k-1}}
  \right\|^2
  &=
  \cos^2\!\left(\sqrt{ks}\,\theta\right)
  \left\|\mathcal A_s\ket{\rho_{k-1}}\right\|^2
  \nonumber\\
  &\quad
  +
  \sin^2\!\left(\sqrt{ks}\,\theta\right)
  \left\|\mathcal A_s\ket{\rho_k}\right\|^2.
  \label{eq:app-norm-split}
\end{align}
Define
\begin{equation}
  \Gamma_r
  :=
  \left\|\mathcal A_s\ket{\rho_r}\right\|\,,
  \qquad
  \Delta_k
  :=
  \max\{\Gamma_{k-1},\Gamma_k\}\,.
\end{equation}
Substituting this bound into Eq.~\eqref{eq:app-state-error} and using
$M_k\tau_k=t_k$ gives
\begin{equation}
  \varepsilon_{k,\mathrm{lead}}^{\mathrm{SD}}
  \leq
  \frac{t_k^2}{2M_k}\Delta_k\,.
  \label{eq:app-eps}
\end{equation}
The transfer time is
\begin{equation}
  t_k
  =
  \frac{\pi}{2\sqrt{ks}}\,.
\end{equation}
Thus, making the leading error at most $\delta_k$ requires
\begin{equation}
  M_k
  \geq
  \frac{\pi^2}{8ks\,\delta_k}\Delta_k\,.
  \label{eq:app-Mk}
\end{equation}

\subsection{Multinomial structure of
\texorpdfstring{$\ket{\rho_r}$}{rho\_r}}

Because the currents $J_\alpha$ act on different chains, they commute. The
multinomial theorem gives
\begin{equation}
  J_s^r\seed
  =
  \sum_{\substack{d_1,\ldots,d_s\geq0\\
  d_1+\cdots+d_s=r}}
  \frac{r!}{d_1!\cdots d_s!}
  \prod_{\alpha=1}^{s}J_\alpha^{d_\alpha}\seed\,.
  \label{eq:app-multinomial}
\end{equation}

For one chain, define
\begin{equation}
  \ket{\xi_d}_{\alpha}
  :=
  \frac{
    J_\alpha^d\ket{1^m0^m}_{\alpha}
  }{\sqrt{d!}}\,.
\end{equation}
By Proposition~\ref{prop:ladder}, specialized to a single chain, these states
satisfy
\begin{equation}
  J_\alpha\ket{\xi_d}_{\alpha}
  =
  \sqrt{d+1}\ket{\xi_{d+1}}_{\alpha}\,,
  \qquad
  J_\alpha^\dagger\ket{\xi_d}_{\alpha}
  =
  \sqrt d\ket{\xi_{d-1}}_{\alpha}\,,
\end{equation}
and
\begin{equation}
  {}_{\alpha}\!\braket{\xi_d|\xi_{d'}}_{\alpha}
  =
  \delta_{dd'}\,.
\end{equation}

The product states belonging to different vectors
$\vec d=(d_1,\ldots,d_s)$ are orthogonal. Consequently,
\begin{align}
  \left\|J_s^r\seed\right\|^2
  &=
  (r!)^2
  \sum_{|\vec d|=r}
  \frac{1}{d_1!\cdots d_s!}
  \nonumber\\
  &=
  r!s^r\,.
\end{align}
The normalized state after $r$ transfers is therefore
\begin{equation}
  \ket{\rho_r}
  =
  \sum_{|\vec d|=r}
  \sqrt{
    \frac{r!}{
      s^r d_1!\cdots d_s!
    }
  }
  \bigotimes_{\alpha=1}^{s}
  \ket{\xi_{d_\alpha}}_{\alpha}\,.
  \label{eq:app-rho}
\end{equation}

\subsection{Same-chain and cross-chain terms}

The operator $\mathcal A_s$ splits as
\begin{equation}
  \mathcal A_s
  =
  \mathcal A_{s,\mathrm{same}}
  +
  \mathcal A_{s,\mathrm{cross}}\,,
\end{equation}
where
\begin{align}
  \mathcal A_{s,\mathrm{same}}
  &:=
  \sum_{\alpha=1}^{s}\mathcal A_{s,\alpha},
  \nonumber\\
  \mathcal A_{s,\mathrm{cross}}
  &:=
  \sum_{1\leq\alpha<\beta\leq s}
  \left(
    J_\alpha^\dagger J_\beta
    -
    J_\beta^\dagger J_\alpha
  \right)\,,
  \label{eq:app-Dsplit}
\end{align}
and
\begin{equation}
  \mathcal A_{s,\alpha}
  :=
  \sum_{\substack{
    1\leq q<q'\leq2m-1\\
    q'-q\geq2
  }}
  \left(
    A_{\alpha,q}^\dagger A_{\alpha,q'}
    -
    A_{\alpha,q'}^\dagger A_{\alpha,q}
  \right)\,.
\end{equation}
The expression for $\mathcal A_{s,\mathrm{cross}}$ follows because every bond
on chain $\alpha$ is non-adjacent to every bond on a different chain $\beta$.

\begin{lemma}[Orthogonality]
\label{lem:app-orth}
For every $0\leq r\leq m$,
\begin{equation}
  \mathcal A_{s,\mathrm{same}}\ket{\rho_r}
  \perp
  \mathcal A_{s,\mathrm{cross}}\ket{\rho_r}\,.
\end{equation}
Consequently,
\begin{equation}
  \Gamma_r^2
  =
  \left\|
    \mathcal A_{s,\mathrm{same}}\ket{\rho_r}
  \right\|^2
  +
  \left\|
    \mathcal A_{s,\mathrm{cross}}\ket{\rho_r}
  \right\|^2\,.
\end{equation}
\end{lemma}

\begin{proof}
Let
\begin{equation}
  \mathcal L_r
  :=
  \operatorname{span}
  \left\{
    \bigotimes_{\alpha=1}^{s}
    \ket{\xi_{d_\alpha}}_{\alpha}
    :
    d_1+\cdots+d_s=r
  \right\}\,.
\end{equation}
Each term $J_\alpha^\dagger J_\beta$ maps a product of ladder states to
another product of ladder states with the same total value of
$d_1+\cdots+d_s$. Hence,
\begin{equation}
  \mathcal A_{s,\mathrm{cross}}\ket{\rho_r}
  \in
  \mathcal L_r\,.
\end{equation}

The operator $\mathcal A_{s,\alpha}$ preserves the value of $d$ on chain
$\alpha$. In the occupation basis, it is represented by a real
skew-symmetric matrix, while $\ket{\xi_d}_{\alpha}$ has real coefficients.
Therefore,
\begin{equation}
  {}_{\alpha}\!\bra{\xi_d}
  \mathcal A_{s,\alpha}
  \ket{\xi_d}_{\alpha}
  =
  0\,.
\end{equation}
Using Eq.~\eqref{eq:app-rho},
$\mathcal A_{s,\mathrm{same}}$ acts on one chain factor at a time. Its inner
product with any product state in $\mathcal L_r$ therefore contains a factor
of the form
\begin{equation}
  {}_{\alpha}\!\bra{\xi_d}
  \mathcal A_{s,\alpha}
  \ket{\xi_d}_{\alpha}\,,
\end{equation}
which vanishes. Thus,
\begin{equation}
  \mathcal A_{s,\mathrm{same}}\ket{\rho_r}
  \perp
  \mathcal L_r\,,
\end{equation}
which proves the result.
\end{proof}

\begin{lemma}[Cross-chain norm]
\label{lem:app-cross}
For every $0\leq r\leq m$,
\begin{equation}
  \left\|
    \mathcal A_{s,\mathrm{cross}}\ket{\rho_r}
  \right\|^2
  =
  \frac{r(s^2-1)}{3}\,.
\end{equation}
\end{lemma}

\begin{proof}
On the ladder subspace, identify
\begin{equation}
  J_\alpha
  \longleftrightarrow
  b_\alpha^\dagger\,,
  \qquad
  J_\alpha^\dagger
  \longleftrightarrow
  b_\alpha\,.
\end{equation}
Under this identification,
\begin{equation}
  \ket{\rho_r}
  \longleftrightarrow
  \frac{
    (b_u^\dagger)^r
  }{\sqrt{r!}}
  \ket{\Omega}\,,
  \qquad
  u
  :=
  \frac{1}{\sqrt s}(1,\ldots,1)^T\,.
\end{equation}
The cross-chain operator becomes
\begin{equation}
  T
  :=
  \sum_{\alpha,\beta=1}^{s}
  K_{\beta\alpha}
  b_\beta^\dagger b_\alpha\,,
  \qquad
  K_{\mu\eta}
  :=
  \operatorname{sgn}(\mu-\eta)\,.
\end{equation}
Since
\begin{equation}
  [T,b_u^\dagger]
  =
  b_{Ku}^\dagger\,,
  \qquad
  T\ket{\Omega}
  =
  0\,,
\end{equation}
the Leibniz rule gives
\begin{equation}
  T
  \frac{(b_u^\dagger)^r}{\sqrt{r!}}
  \ket{\Omega}
  =
  \sqrt r\,
  b_{Ku}^\dagger
  \frac{(b_u^\dagger)^{r-1}}{\sqrt{(r-1)!}}
  \ket{\Omega}\,.
\end{equation}
The matrix $K$ is real and skew-symmetric, so
\begin{equation}
  u^T K u=0\,.
\end{equation}
Thus $Ku\perp u$, and
\begin{equation}
  \left\|
    \mathcal A_{s,\mathrm{cross}}\ket{\rho_r}
  \right\|^2
  =
  r\|Ku\|^2\,.
\end{equation}
The components of $Ku$ are
\begin{equation}
  (Ku)_\mu
  =
  \frac{2\mu-s-1}{\sqrt s}\,.
\end{equation}
Therefore,
\begin{align}
  \|Ku\|^2
  &=
  \frac{1}{s}
  \sum_{\mu=1}^{s}
  (2\mu-s-1)^2
  \nonumber\\
  &=
  \frac{s^2-1}{3}\,,
\end{align}
which proves the claim.
\end{proof}

\begin{lemma}[Same-chain contribution]
\label{lem:app-same}
For every $0\leq r\leq m$,
\begin{equation}
  \left\|
    \mathcal A_{s,\mathrm{same}}\ket{\rho_r}
  \right\|^2
  =
  s
  \sum_{d=0}^{r}
  \binom{r}{d}
  \left(\frac1s\right)^d
  \left(1-\frac1s\right)^{r-d}
  F(d)\,,
  \label{eq:app-same}
\end{equation}
where
\begin{equation}
  F(d)
  :=
  \left\|
    \mathcal A_{s,\alpha}\ket{\xi_d}_{\alpha}
  \right\|^2\,.
\end{equation}
The value of $F(d)$ is independent of $\alpha$ because all chains are
identical.
\end{lemma}

\begin{proof}
For $\alpha\neq\beta$, the operator
$\mathcal A_{s,\alpha}^\dagger\mathcal A_{s,\beta}$ preserves the displacement
on every chain. Expanding $\ket{\rho_r}$ using Eq.~\eqref{eq:app-rho}, its
expectation value separates into one-chain factors. The chain-$\alpha$ factor
contains
\begin{equation}
  {}_{\alpha}\!\bra{\xi_d}
  \mathcal A_{s,\alpha}
  \ket{\xi_d}_{\alpha}
  =
  0\,.
\end{equation}
Therefore, all cross terms vanish and
\begin{equation}
  \left\|
    \mathcal A_{s,\mathrm{same}}\ket{\rho_r}
  \right\|^2
  =
  \sum_{\alpha=1}^{s}
  \left\|
    \mathcal A_{s,\alpha}\ket{\rho_r}
  \right\|^2\,.
\end{equation}

For a fixed chain $\alpha$, the number $d$ of deposited hops has the binomial
distribution
\begin{equation}
  w_d
  =
  \binom{r}{d}
  \left(\frac1s\right)^d
  \left(1-\frac1s\right)^{r-d}\,.
\end{equation}
Displacement sectors with different values of $d$ are orthogonal, so
\begin{equation}
  \left\|
    \mathcal A_{s,\alpha}\ket{\rho_r}
  \right\|^2
  =
  \sum_{d=0}^{r}w_dF(d)\,.
\end{equation}
All $s$ chains give the same contribution, which proves the result.
\end{proof}

Combining Lemmas~\ref{lem:app-orth}--\ref{lem:app-same} gives
\begin{equation}
  \Gamma_r^2
  =
  \frac{r(s^2-1)}{3}
  +
  s
  \sum_{d=0}^{r}
  \binom{r}{d}
  \left(\frac1s\right)^d
  \left(1-\frac1s\right)^{r-d}
  F(d)\,.
  \label{eq:app-gamma}
\end{equation}

\begin{lemma}[One-chain bound]
\label{lem:app-F-quadratic}
For every $0\leq d\leq m$,
\begin{equation}
  F(d)\leq d^2\,.
  \label{eq:app-F-quadratic}
\end{equation}
\end{lemma}

\begin{proof}
The case $d=0$ is immediate. For $d\geq1$, use the Young-diagram basis for one
chain. Because $d\leq m$, every partition of $d$ fits inside the
$m\times m$ diagram allowed by the chain.

For $\lambda\vdash d$, let $f^\lambda$ be the number of standard Young
tableaux of shape $\lambda$. The one-chain ladder state is
\begin{equation}
  \ket{\xi_d}
  =
  \frac{1}{\sqrt{d!}}
  \sum_{\lambda\vdash d}
  f^\lambda\ket{\lambda}\,.
\end{equation}
Write
\begin{equation}
  \sqrt{d!}\,
  \mathcal A_{s,\alpha}\ket{\xi_d}
  =
  \sum_{\lambda\vdash d}
  c_\lambda\ket{\lambda}\,.
\end{equation}

The operator $\mathcal A_{s,\alpha}$ removes one corner box and adds another.
Each such move has coefficient $0$, $+1$, or $-1$. Therefore,
\begin{equation}
  c_\lambda
  =
  \sum_{\substack{\eta\subset\lambda\\|\eta|=d-1}}
  \sum_{\substack{\mu\supset\eta\\|\mu|=d}}
  \epsilon_{\lambda,\mu,\eta}\,f^\mu\,,
  \qquad
  |\epsilon_{\lambda,\mu,\eta}|\leq1\,.
\end{equation}
Taking absolute values and using the standard counting identities
\begin{equation}
  \sum_{\substack{\eta\subset\lambda\\|\eta|=d-1}} f^\eta
  =
  f^\lambda\,,
  \qquad
  \sum_{\substack{\mu\supset\eta\\|\mu|=d}} f^\mu
  =
  d f^\eta\,,
\end{equation}
we obtain
\begin{align}
  |c_\lambda|
  &\leq
  \sum_{\substack{\eta\subset\lambda\\|\eta|=d-1}}
  \sum_{\substack{\mu\supset\eta\\|\mu|=d}}
  f^\mu
  \nonumber\\
  &=
  d
  \sum_{\substack{\eta\subset\lambda\\|\eta|=d-1}}
  f^\eta
  =
  d f^\lambda\,.
\end{align}
Hence,
\begin{align}
  F(d)
  &=
  \left\|
    \mathcal A_{s,\alpha}\ket{\xi_d}
  \right\|^2
  =
  \frac{1}{d!}
  \sum_{\lambda\vdash d}
  c_\lambda^2
  \nonumber\\
  &\leq
  \frac{d^2}{d!}
  \sum_{\lambda\vdash d}
  (f^\lambda)^2
  =
  d^2\,,
\end{align}
where the last equality uses
\begin{equation}
  \sum_{\lambda\vdash d}(f^\lambda)^2=d!\,.
\end{equation}
\end{proof}

Combining Lemma~\ref{lem:app-F-quadratic} with Eq.~\eqref{eq:app-gamma}
gives
\begin{equation}
  \Gamma_r^2
  \leq
  \frac{r(s^2-1)}{3}
  +r\left(1-\frac1s\right)
  +\frac{r^2}{s}\,.
  \label{eq:app-gamma-universal}
\end{equation}

\subsection{Gate count from state-dependent Trotter error}

We now use Eq.~\eqref{eq:app-gamma-universal} to bound the leading
state-dependent error. For $0\leq r\leq m$,
\begin{equation}
  \Gamma_r
  =
  O\!\left(
    s\sqrt r+\frac{r}{\sqrt s}
  \right)\,.
\end{equation}
It follows that
\begin{equation}
  \Delta_k
  =
  O\!\left(
    s\sqrt k+\frac{k}{\sqrt s}
  \right)\,.
\end{equation}
Substituting this into Eq.~\eqref{eq:app-eps} and using
$t_k=\pi/(2\sqrt{ks})$ gives
\begin{equation}
  \varepsilon_{k,\mathrm{lead}}^{\mathrm{SD}}
  =
  O\!\left[
    \frac{1}{M_k}
    \left(
      \frac{1}{\sqrt k}
      +
      \frac{1}{s^{3/2}}
    \right)
  \right]\,.
  \label{eq:app-leading-error-final}
\end{equation}
The first term comes from commutators between different chains. The second
comes from commutators within the same chain.

Let $\delta$ be the allowed total leading-order error. Assigning leading
error $\delta/m$ to each of the $m$ transfers, it is enough to choose
\begin{equation}
  M_k
  =
  \left\lceil
    C\frac{m}{\delta}
    \left(
      \frac{1}{\sqrt k}
      +
      \frac{1}{s^{3/2}}
    \right)
  \right\rceil
\end{equation}
for a sufficiently large constant $C$. Since
\begin{equation}
  \sum_{k=1}^{m}\frac{1}{\sqrt k}
  \leq
  2\sqrt m\,,
\end{equation}
the total number of Trotter steps is
\begin{equation}
  M_{\mathrm{tot}}
  :=
  \sum_{k=1}^{m}M_k
  =
  O\!\left[
    \frac{1}{\delta}
    \left(
      m^{3/2}
      +
      \frac{m^2}{s^{3/2}}
    \right)
  \right]\,.
  \label{eq:app-total-steps}
\end{equation}
The $O(m)$ contribution from the ceilings is absorbed into this bound because
$0<\delta<1$.

Each Trotter step contains
\begin{equation}
  s(2m-1)=2n-s=O(n)
\end{equation}
bond evolutions. Each bond evolution uses a constant number of elementary
gates. The resulting leading-order gate count is
\begin{equation}
  G_{\mathrm{lead}}^{\mathrm{SD}}(s^m)
  =
  O\!\left[
    \frac{n}{\delta}
    \left(
      m^{3/2}
      +
      \frac{m^2}{s^{3/2}}
    \right)
  \right]\,.
  \label{eq:app-general-gates}
\end{equation}

The two terms in Eq.~\eqref{eq:app-general-gates} give the following cases.
\begin{itemize}
  \item If $m=O(s^3)$, then
  \begin{equation}
    \frac{m^2}{s^{3/2}}
    =
    O(m^{3/2})\,.
  \end{equation}
  The same-chain term does not change the asymptotic gate count, and
  \begin{equation}
    G_{\mathrm{lead}}^{\mathrm{SD}}(s^m)
    =
    O\!\left(
      \frac{n m^{3/2}}{\delta}
    \right)\,.
    \label{eq:app-cross-dominated-gates}
  \end{equation}
  If $m=o(s^3)$ and $s$ grows, the same-chain term is strictly lower order.

  \item If $m\gg s^3$, the same-chain term cannot be dropped using the proved
  bound $F(d)\leq d^2$. The second term in
  Eq.~\eqref{eq:app-general-gates} is then larger, giving
  \begin{equation}
    G_{\mathrm{lead}}^{\mathrm{SD}}(s^m)
    =
    O\!\left(
      \frac{n m^2}{\delta s^{3/2}}
    \right)\,.
  \end{equation}

  \item Our numerical results are consistent with linear growth
  $F(d)=O(d)$, with an observed slope of approximately $0.6$ over the tested
  range. If this stronger bound holds for all $d$, then the same-chain
  contribution is $O(r)$ rather than $O(r+r^2/s)$. In that case,
  \begin{equation}
    G_{\mathrm{lead}}^{\mathrm{SD}}(s^m)
    =
    O\!\left(
      \frac{n m^{3/2}}{\delta}
    \right)
  \end{equation}
  for every uniform-cycle partition, without requiring $m=O(s^3)$. This
  statement is conditional on the unproven linear bound.
\end{itemize}

These estimates control the leading $1/M_k$ contribution to the
state-dependent Trotter error. They do not give a complete finite-$M_k$ error
bound because the higher-order remainder in
Eq.~\eqref{eq:app-state-error} has not been bounded.

\bibliographystyle{unsrt}
\bibliography{refs.bib}

\begin{thebibliography}{10}

\bibitem{BravyiGossetHavlicekSchatzki2025}
Sergey Bravyi, David Gosset, Vojt{\v{e}}ch Havl{\'\i}{\v{c}}ek, and Louis
  Schatzki.
\newblock Classical and quantum algorithms for characters of the symmetric
  group.
\newblock {\em PRX Quantum}, 6:030323, 2025.
\newblock arXiv:2501.12579.

\bibitem{MezardParisiVirasoro1987}
Marc M{\'e}zard, Giorgio Parisi, and Miguel~Angel Virasoro.
\newblock {\em Spin Glass Theory and Beyond: An Introduction to the Replica
  Method and Its Applications}, volume~9 of {\em World Scientific Lecture Notes
  in Physics}.
\newblock World Scientific, 1987.

\bibitem{ZhouNahum2019}
Tianci Zhou and Adam Nahum.
\newblock Emergent statistical mechanics of entanglement in random unitary
  circuits.
\newblock {\em Physical Review B}, 99:174205, 2019.

\bibitem{Dunne1993}
Gerald~V. Dunne.
\newblock {Slater} decomposition of {Laughlin} states.
\newblock {\em International Journal of Modern Physics B}, 7(28):4783--4813,
  1993.

\bibitem{DijkgraafMooreVerlindeVerlinde1997}
Robbert Dijkgraaf, Gregory Moore, Erik Verlinde, and Herman Verlinde.
\newblock Elliptic genera of symmetric products and second quantized strings.
\newblock {\em Communications in Mathematical Physics}, 185:197--209, 1997.

\bibitem{Thorngren}
Ryan Thorngren and Yifan Wang.
\newblock Fusion category symmetry. {Part II}. {Categoriosities} at $c=1$ and
  beyond.
\newblock {\em Journal of High Energy Physics}, 07:051, 2024.

\bibitem{GutperleLiRathoreRoumpedakis2024_2}
Michael Gutperle, Yan-Yan Li, Dikshant Rathore, and Konstantinos Roumpedakis.
\newblock Non-invertible symmetries in {$S_N$} orbifold {CFTs} and holography.
\newblock {\em Journal of High Energy Physics}, 09:110, 2024.

\bibitem{GutperleLiRathoreRoumpedakis2024}
Michael Gutperle, Yan-Yan Li, Dikshant Rathore, and Konstantinos Roumpedakis.
\newblock A note on entanglement entropy and topological defects in symmetric
  orbifold {CFTs}.
\newblock {\em Journal of High Energy Physics}, 09:010, 2024.

\bibitem{BintanjaGesteau2026}
Suzanne Bintanja and Elliott Gesteau.
\newblock The $\alpha$-states of a string worldsheet.
\newblock 2026.
\newblock arXiv:2607.15336.

\bibitem{Hepler1994}
Charles~T. Hepler.
\newblock On the complexity of computing characters of finite groups.
\newblock Master's thesis, University of Calgary, 1994.

\bibitem{CrichignoPrakash2024}
Marcos Crichigno and Anupam Prakash.
\newblock Quantum spin chains and symmetric functions.
\newblock 2024.
\newblock arXiv:2404.04322.

\bibitem{Ikenmeyer_2023}
Christian Ikenmeyer, Igor Pak, and Greta Panova.
\newblock Positivity of the symmetric group characters is as hard as the
  polynomial time hierarchy.
\newblock {\em International Mathematics Research Notices},
  2024(10):8442--8458, 2024.

\bibitem{Bastidas_2025}
Victor~M. Bastidas, Nathan Fitzpatrick, K.~J. Joven, Zane~M. Rossi, Shariful
  Islam, Troy Van~Voorhis, Isaac~L. Chuang, and Yuan Liu.
\newblock Unification of finite symmetries in the simulation of many-body
  systems on quantum computers.
\newblock {\em Physical Review A}, 111:052433, 2025.

\bibitem{Bruinsma2026}
Carli Bruinsma, Dmitry Grinko, and Maris Ozols.
\newblock Quantum {Fourier} transform for the symmetric group.
\newblock 2026.
\newblock arXiv:2608.28569.

\bibitem{Bluvstein_2023}
Dolev Bluvstein, Simon~J. Evered, Alexandra~A. Geim, Sophie~H. Li, Hengyun
  Zhou, Tom Manovitz, Sepehr Ebadi, Madelyn Cain, Marcin Kalinowski, Dominik
  Hangleiter, J.~Pablo Bonilla~Ataides, Nishad Maskara, Iris Cong, Xun Gao,
  Pedro Sales~Rodriguez, Thomas Karolyshyn, Giulia Semeghini, Michael~J.
  Gullans, Markus Greiner, Vladan Vuleti{\'c}, and Mikhail~D. Lukin.
\newblock Logical quantum processor based on reconfigurable atom arrays.
\newblock {\em Nature}, 626:58--65, 2024.

\bibitem{Maskara_2025}
Nishad Maskara, Marcin Kalinowski, Daniel Gonzalez-Cuadra, and Mikhail~D.
  Lukin.
\newblock Fast simulation of fermions with reconfigurable qubits, 2025.
\newblock arXiv:2509.08898.

\bibitem{Beals_1997}
Robert Beals.
\newblock Quantum computation of {Fourier} transforms over symmetric groups.
\newblock In {\em Proceedings of the Twenty-Ninth Annual ACM Symposium on
  Theory of Computing (STOC '97)}, pages 48--53. Association for Computing
  Machinery, 1997.

\bibitem{Gilyen_2019}
Andr{\'a}s Gily{\'e}n, Yuan Su, Guang~Hao Low, and Nathan Wiebe.
\newblock Quantum singular value transformation and beyond: exponential
  improvements for quantum matrix arithmetics.
\newblock In {\em Proceedings of the 51st Annual ACM SIGACT Symposium on Theory
  of Computing (STOC '19)}, pages 193--204. ACM, 2019.

\bibitem{ChildsSuTranWiebeZhu2021}
Andrew~M. Childs, Yuan Su, Minh~C. Tran, Nathan Wiebe, and Shuchen Zhu.
\newblock Theory of {Trotter} error with commutator scaling.
\newblock {\em Physical Review X}, 11:011020, 2021.

\bibitem{Stanley1999}
Richard~P. Stanley.
\newblock {\em Enumerative Combinatorics, Volume 2}, volume~62 of {\em
  Cambridge Studies in Advanced Mathematics}.
\newblock Cambridge University Press, 1999.

\bibitem{JamesKerber1981}
Gordon James and Adalbert Kerber.
\newblock {\em The Representation Theory of the Symmetric Group}, volume~16 of
  {\em Encyclopedia of Mathematics and its Applications}.
\newblock Addison-Wesley, Reading, Massachusetts, 1981.

\bibitem{FrameRobinsonThrall1954}
J.~S. Frame, G.~de~B. Robinson, and R.~M. Thrall.
\newblock The hook graphs of the symmetric group.
\newblock {\em Canadian Journal of Mathematics}, 6:316--324, 1954.

\bibitem{LaroccaHavlicek2025}
Mart{\'\i}n Larocca and Vojt{\v{e}}ch Havl{\'\i}{\v{c}}ek.
\newblock Quantum algorithms for representation-theoretic multiplicities.
\newblock {\em Physical Review Letters}, 135:010602, 2025.

\bibitem{Panova2025}
Greta Panova.
\newblock Polynomial time classical versus quantum algorithms for
  representation theoretic multiplicities, 2025.
\newblock arXiv:2502.20253.

\end{thebibliography}

\end{appendices}
\end{document}